\pdfoutput=1
\documentclass[12pt]{article}

\usepackage[T1]{fontenc}
\usepackage[utf8]{inputenc}
\usepackage[english]{babel}

\usepackage{amsmath}
\usepackage{amssymb}
\usepackage{amsthm}
\usepackage{mathtools}
\mathtoolsset{showonlyrefs,showmanualtags}

\allowdisplaybreaks

\usepackage[natbibapa,nodoi]{apacite}
\usepackage{authblk}

\usepackage{bbm}
\newcommand{\bbone}{\mathbbm{1}}
\usepackage{booktabs}
\usepackage[format=plain,font=small,labelfont=bf,textfont=up]{caption}
\usepackage{enumitem}
\setlist[enumerate]{itemsep=0.1ex}
\setlist[itemize]{itemsep=0.1ex}
\newcommand\expcommand\newcommand
\usepackage[letterpaper,margin=1.05in,bottom=1.65in]{geometry}
\usepackage{graphicx}
\usepackage[hidelinks]{hyperref}
\usepackage{setspace}
\usepackage[nohints,tight]{minitoc}

\usepackage{tikz}
\usetikzlibrary{arrows.meta,positioning}

\newcommand\refmain\ref

\DeclarePairedDelimiter\paren\lparen\rparen
\DeclarePairedDelimiter\bracket\lbrack\rbrack
\DeclarePairedDelimiter\braces\lbrace\rbrace

\DeclarePairedDelimiter\abs\lvert\rvert

\providecommand{\bbone}{\mathbf{1}}
\DeclarePairedDelimiterXPP\indicator[1]{\bbone}{\lbrack}{\rbrack}{}{#1}

\DeclarePairedDelimiterXPP\expf[1]{\exp}{\lparen}{\rparen}{}{#1}
\DeclarePairedDelimiterXPP\logf[1]{\log}{\lparen}{\rparen}{}{#1}
\DeclarePairedDelimiterXPP\maxf[1]{\max}{\lparen}{\rparen}{}{#1}
\DeclarePairedDelimiterXPP\minf[1]{\min}{\lparen}{\rparen}{}{#1}

\DeclarePairedDelimiterXPP\func[2]{#1}{\lparen}{\rparen}{}{#2}

\newcommand{\set}[1]{#1}

\newcommand\suchthat{:\allowbreak\mathopen{}}
\DeclarePairedDelimiter\setb\lbrace\rbrace

\newcommand{\Reals}{\mathbb{R}}

\newcommand{\Naturals}{\mathbb{N}}
\newcommand{\bset}{\setb{0, 1}}

\newcommand{\comp}{\mathsf{c}}

\DeclareMathOperator*{\argmin}{arg\,min}
\DeclarePairedDelimiter\card\lvert\rvert

\makeatletter
\newcommand*{\tran}{{\mathpalette\@tran{}}}
\newcommand*{\@tran}[2]{\raisebox{\depth}{$\m@th#1\intercal$}}
\makeatother

\DeclarePairedDelimiterXPP\tnorm[2]{}{\lVert}{\rVert_{1}}{}{#2}
\DeclarePairedDelimiterXPP\enorm[2]{}{\lVert}{\rVert_{2}}{}{#2}
\DeclarePairedDelimiterXPP\inorm[2]{}{\lVert}{\rVert_{\infty}}{}{#2}
\DeclarePairedDelimiterXPP\pnorm[2]{}{\lVert}{\rVert_{#1}}{}{#2}

\DeclarePairedDelimiterXPP\detf[1]{\det}{\lparen}{\rparen}{}{#1}

\DeclareMathOperator{\trsym}{tr}
\DeclarePairedDelimiterXPP\tr[1]{\trsym}{\lparen}{\rparen}{}{#1}

\providecommand\given{}
\newcommand\givensymbol[1]{\nonscript\:#1\vert\allowbreak\nonscript\:\mathopen{}}

\let\Pr\relax
\DeclareMathOperator{\Prsym}{pr}
\DeclarePairedDelimiterXPP\Pr[1]{\Prsym}{\lparen}{\rparen}{}{%
	\renewcommand\given{\givensymbol{\delimsize}}%
	#1}

\DeclarePairedDelimiterXPP\Prs[1]{\Prsym}{\lbrace}{\rbrace}{}{%
	\renewcommand\given{\givensymbol{\delimsize}}%
	#1}

\DeclarePairedDelimiterXPP\Prt[1]{\Prsym}{\lbrack}{\rbrack}{}{%
	\renewcommand\given{\givensymbol{\delimsize}}%
	#1}

\DeclareMathOperator{\Esym}{E}
\DeclarePairedDelimiterXPP\E[1]{\Esym}{\lparen}{\rparen}{}{%
	\renewcommand\given{\givensymbol{\delimsize}}%
	#1}

\DeclarePairedDelimiterXPP\Es[1]{\Esym}{\lbrace}{\rbrace}{}{%
	\renewcommand\given{\givensymbol{\delimsize}}%
	#1}

\DeclarePairedDelimiterXPP\Et[1]{\Esym}{\lbrack}{\rbrack}{}{%
	\renewcommand\given{\givensymbol{\delimsize}}%
	#1}

\DeclareMathOperator{\Varsym}{var}
\DeclarePairedDelimiterXPP\Var[1]{\Varsym}{\lparen}{\rparen}{}{%
	\renewcommand\given{\givensymbol{\delimsize}}%
	#1}

\DeclarePairedDelimiterXPP\EstVar[1]{\widehat{\Varsym}}{\lparen}{\rparen}{}{%
	\renewcommand\given{\givensymbol{\delimsize}}%
	#1}

\DeclareMathOperator{\Covsym}{cov}
\DeclarePairedDelimiterXPP\Cov[1]{\Covsym}{\lparen}{\rparen}{}{%
	\renewcommand\given{\givensymbol{\delimsize}}%
	#1}

\makeatletter
\newcommand{\indep}{\protect\mathpalette{\protect\@indep}{\perp}}
\newcommand*{\@indep}[2]{\mathrel{\rlap{$#1#2$}\mkern3mu{#1#2}}}
\makeatother

\newcommand{\bigOsym}{\mathcal{O}}
\DeclarePairedDelimiterXPP\bigO[1]{\bigOsym}{\lparen}{\rparen}{}{#1}

\newcommand{\littleOsym}{o}
\DeclarePairedDelimiterXPP\littleO[1]{\littleOsym}{\lparen}{\rparen}{}{#1}

\newcommand{\bigOpsym}{\bigOsym_p}
\DeclarePairedDelimiterXPP\bigOp[1]{\bigOpsym}{\lparen}{\rparen}{}{#1}

\newcommand{\littleOpsym}{\littleOsym_p}
\DeclarePairedDelimiterXPP\littleOp[1]{\littleOpsym}{\lparen}{\rparen}{}{#1}

\newcommand{\bigOmegasym}{\Omega}
\DeclarePairedDelimiterXPP\bigOmega[1]{\bigOmegasym}{\lparen}{\rparen}{}{#1}

\newcommand{\littleOmegasym}{\omega}
\DeclarePairedDelimiterXPP\littleOmega[1]{\littleOmegasym}{\lparen}{\rparen}{}{#1}

\newcommand{\bigThetasym}{\Theta}
\DeclarePairedDelimiterXPP\bigTheta[1]{\bigThetasym}{\lparen}{\rparen}{}{#1}

\newcommand{\diff}{\mathop{}\!\mathrm{d}}

\theoremstyle{plain}

\newtheorem{corollary}{Corollary}
\newtheorem{lemma}{Lemma}
\newtheorem{proposition}{Proposition}
\newtheorem{theorem}{Theorem}

\newtheorem{innerreflemma}{Lemma}
\newenvironment{reflemma}[1]
{\renewcommand\theinnerreflemma{#1}\innerreflemma}
{\endinnerreflemma}

\newtheorem{innerrefproposition}{Proposition}
\newenvironment{refproposition}[1]
{\renewcommand\theinnerrefproposition{#1}\innerrefproposition}
{\endinnerrefproposition}

\theoremstyle{definition}

\newtheorem{condition}{Condition}
\newtheorem{definition}{Definition}

\theoremstyle{remark}

\renewcommand{\set}[1]{\mathcal{#1}}

\newcommand{\sumin}{\sum_{i=1}^{n}}

\newcommand{\avgin}{\frac{1}{n}\sumin}
\newcommand{\cavgin}{n^{-1}\sumin}

\newcommand{\Sample}{\set{U}}

\newcommand{\att}{\tau_{\normalfont\textsc{att}}}
\newcommand{\attest}{\hat{\tau}_{\normalfont\textsc{att}}}

\newcommand{\Xvpop}{X}
\newcommand{\Xv}[1]{\Xvpop_{#1}}
\newcommand{\Xvi}{\Xv{i}}
\newcommand{\Xvj}{\Xv{j}}

\newcommand{\allX}{\set{X}}

\newcommand{\Wpop}{W}
\newcommand{\W}[1]{\Wpop_{#1}}
\newcommand{\Wi}{\W{i}}
\newcommand{\Wj}{\W{j}}

\DeclarePairedDelimiterXPP\POpop[1]{\Ypop}{\lparen}{\rparen}{}{#1}
\newcommand{\POtpop}{\POpop{1}}
\newcommand{\POcpop}{\POpop{0}}

\DeclarePairedDelimiterXPP\PO[2]{\Ypop_{#1}}{\lparen}{\rparen}{}{#2}
\newcommand{\POi}[1]{\PO{i}{#1}}
\newcommand{\POti}{\POi{1}}
\newcommand{\POci}{\POi{0}}

\newcommand{\Ypop}{Y}
\newcommand{\Y}[1]{\Ypop_{#1}}
\newcommand{\Yi}{\Y{i}}
\newcommand{\Yj}{\Y{j}}

\newcommand{\Treated}{\set{T}}
\newcommand{\Controls}{\set{C}}
\newcommand{\Matched}{\set{M}^*}

\newcommand{\nTreated}{N_1}
\newcommand{\nControls}{N_0}

\newcommand{\matchsym}{m}
\newcommand{\allM}{\mathcal{M}}
\newcommand{\optM}{\allM^*}

\newcommand{\mgensym}{\matchsym}
\DeclarePairedDelimiterXPP\mgen[1]{\mgensym}{\lparen}{\rparen}{}{#1}
\newcommand{\mgeni}{\mgen{i}}
\newcommand{\mgenj}{\mgen{j}}

\newcommand{\moptsym}{\matchsym^*}
\DeclarePairedDelimiterXPP\mopt[1]{\moptsym}{\lparen}{\rparen}{}{#1}
\newcommand{\mopti}{\mopt{i}}

\newcommand{\distsym}{d}
\DeclarePairedDelimiterXPP\dist[2]{\distsym}{\lparen}{\rparen}{}{#1, #2}

\newcommand{\pscoresym}{\pi}
\DeclarePairedDelimiterXPP\pscore[1]{\pscoresym}{\lparen}{\rparen}{}{#1}
\newcommand{\avgpscore}{\bar{\pscoresym}}

\newcommand{\genpar}{\theta}
\newcommand{\genest}{\hat{\theta}}

\newcommand{\rPSpop}{\Pi}
\newcommand{\rPS}[1]{\rPSpop_{#1}}
\newcommand{\rPSi}{\rPS{i}}
\newcommand{\rPSj}{\rPS{j}}

\newcommand{\PSpoint}{p^*}

\newcommand{\mdiscrep}{D}

\newcommand{\suppPS}{\mathbf{\Pi}_{\normalfont\textsc{supp}}}
\newcommand{\suppPSp}{\Pi^+_{\normalfont\textsc{supp}}}
\newcommand{\suppPSm}{\Pi^-_{\normalfont\textsc{supp}}}

\DeclarePairedDelimiterXPP\EPOPS[2]{\mu_{#1}}{\lparen}{\rparen}{}{#2}

\newcommand{\EPOPSc}[1]{\EPOPS{0}{#1}}

\newcommand{\lemmaconsistentunbiased}{%
	If an estimator $\genest$ is consistent for a parameter $\genpar$, then it is asymptotically unbiased or its variance is asymptotically unbounded:
	\begin{equation}
		\lim_{\varepsilon \to 0} \lim_{n \to \infty} \Pr[\big]{\abs{\genest - \genpar} \geq \varepsilon} = 0 \quad \implies \quad
		\lim_{n \to \infty} \E{\genest} = \genpar \quad\text{or}\quad \limsup_{n \to \infty}\Var{\genest} = \infty.
	\end{equation}%
}

\newcommand{\lemmaboundedvar}{%
	Given Condition~\refmain{cond:reg-cond},
	\begin{equation}
		\limsup_{n \to \infty}\Var{\attest} \leq \frac{4\E{\Ypop^2}}{\avgpscore^2} < \infty.
	\end{equation}%
}

\newcommand{\coromdiscrepinconsistent}{%
	Given Condition~\refmain{cond:reg-cond}, if the matching estimator is asymptotically biased with respect to the average treatment effect of the treated, then it is inconsistent:
	\begin{equation}
		\limsup_{n \to \infty} \abs[\big]{\E{\attest} - \att} > 0
		\quad \implies \quad
		\lim_{\varepsilon \to 0} \limsup_{n \to \infty} \Pr[\big]{\abs{\attest - \att} \geq \varepsilon} > 0.
	\end{equation}%
}

\newcommand{\lemmanocrossingmatches}{%
	An optimal propensity score matching contains no crossing matches.%
}

\newcommand{\proppscorebias}{%
	Given Conditions~\refmain{cond:reg-cond},~\refmain{cond:pscore-lipschitz} and~\refmain{cond:left-closed}, when the matching is constructed without replacement using the true propensity score,
	\begin{equation}
		\lim_{n \to \infty} \E{\attest} = \att + \frac{\Pr{\rPSpop \geq \PSpoint}}{2\avgpscore} \bracket[\Big]{\Es[\big]{\POcpop \given \Wpop = 1, \rPSpop \geq \PSpoint} - \Es[\big]{\POcpop \given \Wpop = 0, \rPSpop \geq \PSpoint}}.
	\end{equation}%
}

\title{\textbf{On the inconsistency of\\matching without replacement}}

\author{Fredrik Sävje}

\affil{Department of Political Science and Department of Statistics \& Data Science, Yale University.}

\date{\today}

\begin{document}

\maketitle

\begin{abstract}
\begin{singlespace}
\noindent The paper shows that matching without replacement on propensity scores produces estimators that generally are inconsistent for the average treatment effect of the treated.
To achieve consistency, practitioners must either assume that no units exist with propensity scores greater than one-half or assume that there is no confounding among such units.
The result is not driven by the use of propensity scores, and similar artifacts arise when matching on other scores as long as it is without replacement.

\end{singlespace}
\end{abstract}

\doparttoc
\faketableofcontents

\bigskip
\section{Introduction}

Matching aims to adjust for confounded treatment assignment when estimating treatment effects with observational data.
Each treated unit is matched to one or more similar control units according to some similarity measure based on the units' observed characteristics.
The treatment effect is estimated by the average difference between the outcomes of the treated units and their matched controls.
There are many variations to this basic recipe.
An important consideration is whether to match with or without replacement.
In the first case, several treated units can be matched to the same control.
In the second case, at most one treated unit can be matched to each control.
Matching with replacement produces matches of higher quality, but the information provided by the controls may be used inefficiently.
This is because fewer controls will be matched with treated units when the matching is with replacement, and unmatched units are discarded from the analysis.
Practitioners sometimes opt for matching without replacement to avoid such issues.
This type of matching has previously been discussed and studied by \citet{Dehejia2002Propensity}, \citet{Rosenbaum2002Observational}, \citet{Stuart2010Matching} and \citet{Abadie2012Martingale}, among others.

The purpose of this paper is to investigate the asymptotic properties of the matching estimator when the matching is done without replacement.
The main result is that the matching estimator generally is inconsistent for the treatment effect it aims to estimate.
The underlying idea is conceptually straightforward.
The sample must contain more controls than treated units to construct a matching without replacement.
If not, one would run out of control units, and some treated units would be left unmatched.
For the matching estimator to be consistent, the quality of the matches must improve as the sample grows in size.
This is only possible if there are more controls than treated units for every possible value of the observed covariates.
If not, one would run out of control units for some covariate values, which would force some treated units to be matched with distant control units.
To ensure that control units are locally abundant, the propensity score must be less than one half everywhere on the support of the covariates.
This condition is considerably stronger than typically assumed when matching is used.

The result suggests that practitioners may want to consider alternatives to matching without replacement when adjusting for confounding.
These alternatives include matching with replacement and various weighting methods.
It is also possible to modify the matching procedure to mitigate the problem.
For example, a caliper that diminishes in the sample size would ensure consistency as long as the estimator is modified to account for the fact that the caliper implicit reweights the treated units.
A concern with this modification, and others like it, is that it would undo much of the potential efficiency gains that prompt practitioners to use matching without replacement in the first place.

The key technical contribution of the paper is that the analysis does not condition on the observed covariates.
To the best of my knowledge, all previous investigations of the behavior of matching without replacement are based on such conditioning, or they consider asymptotic regimes that produce a similar effect, as discussed in Section~\ref{sec:diminishing}.
Once one conditions on the observed covariates, the matching is deterministic, and one may ignore the stochastic behavior of the matching procedure.
The practice greatly simplifies the analysis, but it may be problematic for two reasons.
First, the stochastic behavior of the matching procedure may be a non-negligible source of imprecision of the matching estimator, which is ignored.
Second, to investigate large sample performance, the behavior of the matching method must be stipulated.
This may, for example, be an assumption that the distance between matched units diminishes asymptotically.
However, such assumptions leave practitioners wondering whether matching without replacement in fact behaves in this way when units are sampled from a population.
The worry is that matchings that satisfy the stipulated behavior may be rare, in which case the analyses in the previous literature would condition on an event that practitioners are unlikely to meet in practice.
The technical innovation of this paper is an approach that allows for an unconditional analysis of the matching function and its asymptotic behavior in a standard sampling framework.

\section{Illustration}

An illustration using a categorical covariate will fix ideas.
Consider a population in which $10\%$ of the units belong to a certain covariate category $A$, and $3 / 4$ of these units are treated.
Among the treated units in category $A$, only an expected fraction of $1 / 3$ will be matched to controls in the same category.
This is because controls in the category run out after the first third of the treated units have been matched.
The remaining $2 / 3$, which are $5\%$ of the total sample in expectation, must be matched to controls in other categories.
Unless these other units are representative of the treated units in category $A$, which we have no reason to believe that they are, the poor quality matches will prevent the estimator from concentrating around the treatment effect.
The argument applies as soon as more than half of the units in any category are treated.
Thus, to achieve consistency, there cannot be any such categories.
The purpose of the rest of this paper is to demonstrate that this phenomenon occurs more generally and to derive the asymptotic bias in closed form.

\section{Preliminaries}

\subsection{Notation and regularity conditions}

Consider a population described by a distribution $\braces{\Xvpop, \Wpop, \POcpop, \POtpop, \Ypop}$, where $\Xvpop \in \allX$ is a covariate in some, possibly multi-dimensional, covariate set, $\Wpop \in \bset$ is an indicator of treatment assignment, $\POcpop$ and $\POtpop$ are real-valued potential outcomes, and $\Ypop = \POpop{\Wpop}$ is the realized outcome.
The notation requires that the potential outcomes are unambiguous for each treatment condition, ruling out, for example, that the treatment assigned to one unit affects the outcome of another unit.

A sample of $n$ units is drawn from the population, and the observations are indexed by $\Sample = \setb{1, \dotsc, n}$.
The potential outcomes are never directly observed, and the available information for unit $i$ is $\paren{\Xvi, \Wi, \Yi}$.
The population is assumed to be large, so the observations can be seen as independent and identically distributed according to the population distribution.

The parameter of interest is the treatment effect of the treated units in the population:
\begin{equation}
	\att = \Es[\big]{\POtpop - \POcpop \given \Wpop = 1}.
\end{equation}
The main inferential challenge is that treatment assignment is suspected to be confounded.
That is, the two conditional distributions of $\POcpop$ given $\Wpop$ may not be the same.
Our hope is that we have observed all confounding variables so that we can adjust for the difference in the conditional distributions.
The focus here is when this adjustment is done using matching without replacement.

To formalize the method, let $\Treated = \setb{i \in \Sample \suchthat \Wi = 1}$ and $\Controls = \setb{i \in \Sample \suchthat \Wi = 0}$ be the sets of treated and control units in the sample.
Let $\nTreated = \card{\Treated}$ and $\nControls = \card{\Controls}$ denote their sizes.
A matching can be described as an injective function $\mgensym : \Treated \to \Controls$ where $\mgeni$ gives the match for $i \in \Treated$.
The function is injective because the matching is without replacement: $\mgeni \neq \mgenj$ for all $i \neq j$.
Let $\allM$ collect all such functions, which is the set of all admissible matchings.

Several methods have been devised to select a suitable $\matchsym$ from $\allM$.
This paper considers optimal matching as described by \citet{Rosenbaum1989Optimal}.
The covariate set is here endowed with a metric, $\distsym : \allX \times \allX \to \Reals^+$.
The resulting metric space captures how similar the units are with respect to their covariates.
That is, if $\dist{\Xvi}{\Xvj} < \dist{\Xvi}{\Xv{k}}$, then unit $i$ is deemed more similar to unit $j$ than to $k$.
A continuity assumption will later give the metric meaning by connecting it to the potential outcomes.

Optimal matchings are those that minimize the sum of distances between matched units:
\begin{equation}
	\optM = \argmin_{\matchsym \in \allM} \sum_{i \in \Treated} \dist[\big]{\Xvi}{\Xv{\mgen{i}}}.
\end{equation}
If there is not a unique optimal matching, a matching is picked arbitrarily from the set of optimal matchings in a deterministic fashion.
The selected optimal matching $\moptsym \in \optM$ is thus completely determined by $\paren{\Xv{1}, \dotsc, \Xv{n}}$ and $\paren{\W{1}, \dotsc, \W{n}}$.
If one were to condition on the covariates and treatment assignments, $\moptsym$ is not random.
However, as noted in the introduction, no such conditioning will be done here, and the matching is random.
The sets $\Treated$, $\Controls$ and $\allM$ are also random.

With $\moptsym$ in hand, the treatment effect $\att$ is estimated by the average difference in observed outcomes between each treated unit and its matched control:
\begin{equation}
	\attest = \frac{1}{\nTreated} \sum_{i \in \Treated} \paren[\big]{\Yi - \Y{\mopt{i}}}.
\end{equation}
This is the estimator studied by \citet{Abadie2006Large,Abadie2012Martingale}, and it is widely used by practitioners.
The estimator as stated is not defined when the sample contains no treated units or when it contains more treated units than controls.
Practitioners do not tend to use matching without replacement in either of those two cases, and the conditions below ensure they happen with a probability approaching zero at an exponential rate.
Nevertheless, for completeness, the estimator is defined to be zero in these cases.

To motivate the use of matching adjustment, the population is assumed to satisfy a set of conditions.
A key aspect of these conditions is the propensity score: the fraction of units in the population assigned to treatment.
Let $\avgpscore = \Pr{\Wpop = 1}$ be the overall fraction of treated units, and let $\pscore{x} = \Pr{\Wpop = 1 \given \Xvpop = x}$ be the fraction conditional on the covariate.

\begin{condition}\label{cond:reg-cond}
	The population satisfies:
	\vspace{-6pt}
	\begin{enumerate}[label=\roman*.,left=16pt]
		\item Unconfoundedness: $\POcpop$ is conditionally independent of $\Wpop$ given $\Xvpop = x$ on the support of $\Xvpop$.
		\item Overlap: $\pscore{x}$ is bounded away from one on the support of $\Xvpop$.
		\item Existence of treated units: $\avgpscore$ is greater than zero.
		\item Abundance of control units: $\avgpscore$ is less than one-half.
		\item Well-behaved outcomes: $\E{\Ypop^2}$ exists.
	\end{enumerate}
\end{condition}

The first condition states that all confounding variables are observed.
This ensures that covariate adjustment in principle could resolve the confounding.
The second condition states that the support of the covariate for the treated units is in the interior of the support for the controls.
This ensures that there is enough information in the population for the adjustment.
The combination of the two conditions is sometimes called ignorable treatment assignment \citep{Rosenbaum1983Central}.
Ignorable assignment is, however, typically taken to also include unconfoundedness for the other potential outcome and a lower bound on the propensity score.
This is not needed here because the treatment effect of the treated units is the focus \citep{Heckman1997Matching}.

The third and fourth conditions are the ones mentioned above.
They ensure that large samples almost always contains more controls than treated units and at least some treated units.
This, in turn, ensures that it is possible to construct a matching without replacement with a probability approaching one.
The fifth condition ensures that the outcome distribution in the population does not have tails that are too heavy.

\subsection{Asymptotic bias and consistency}

The focus in the following section is the asymptotic bias of the estimator.
However, the bias itself is not of interest in this paper.
The reason for this focus is that asymptotic unbiasedness is a necessary condition for consistency for the matching estimator given Condition~\ref{cond:reg-cond}.
The following two lemmas provide the details.
All proofs are presented in the supplement.

\begin{lemma}\label{lem:consistent-unbiased}
	\lemmaconsistentunbiased{}
\end{lemma}

\begin{lemma}\label{lem:bounded-var}
	\lemmaboundedvar{}
\end{lemma}

Lemma~\ref{lem:consistent-unbiased} states that the expectation of a consistent estimator must concentrate around the parameter it aims to estimate unless its variance grows without limit in the sample size.
This holds for any estimator, and not only for the matching estimator.
The intuition is that consistency implies that only a negligible mass of the estimator's sampling distribution is outside a small neighborhood of the parameter.
If this small mass is enough to affect the estimator's expectation, the mass must move away from the bulk of the sampling distribution, leading to asymptotically unbounded variance.

The lemma implies that an estimator that is known to have asymptotically bounded variance can only be consistent if it is asymptotically unbiased.
That is, it is necessary (but not sufficient) for consistency that the expectation of the estimator concentrates around the parameter.
Lemma~\ref{lem:bounded-var} shows that the variance of the matching estimator is asymptotically bounded, so it is an estimator of this kind.
Of course, our hope is that the variance will approach zero under suitable conditions.
The purpose of Lemma~\ref{lem:bounded-var} is to show that the variance is bounded no matter what these additional conditions might be.

\begin{corollary}\label{coro:mdiscrep-inconsistent}
	\coromdiscrepinconsistent{}
\end{corollary}

\subsection{Matching on propensity scores}

Matching on high-dimensional covariates do not tend to perform well.
As shown by \citet{Beyer1999When} and others, random points in high-dimensional spaces tend to be equidistant, which makes their distances uninformative.
Indeed, \citet{Abadie2006Large} show that the rate of convergence of the matching estimator is negatively affected by the dimensionality of the covariates when the matching is done with replacement.
Practitioners using matching should for this reason be motivated to reduce the dimensionality of the covariates before matching.
To achieve this, a function can be applied to each unit's covariate, providing a coarser description of its characteristics.
A matching can then be constructed as above but with the coarsened covariates substituted for the raw covariates.

A concern when coarsening the covariates is that one may lose valuable information; unconfoundedness may not hold for the coarsened covariates even if it does so for the raw covariates.
A way to maintain unconfoundedness is to use a coarsening function that still balances the raw covariates.
Such a function is called a balancing score, and \citet{Rosenbaum1983Central} show that the propensity score, as defined above, is the coarsest balancing score.
A metric based on the propensity scores is therefore an attractive alternative to, for example, Euclidean distances on raw covariates, and it is the focus of this paper.
Matching on other scores, including the raw covariates, is discussed in the supplement.

The propensity score is often not known in advance, and it must then be estimated.
I will, however, disregard the estimation step in this paper, and take the propensity score as known.
This will avoid some technical complications of little interest for the current discussion.
In particular, consistent estimation of the propensity score is generally not possible unless additional assumptions are imposed \citep[see, e.g.,][]{Robins1997Curse}.
By focusing on the setting with a known propensity score, I hope to highlight that the results below are not driven by the challenges in this estimating step.
Put differently, it is not surprising that the matching estimator performs poorly when one fails to estimate the propensity score well, but that is not the point I make in this paper.\footnote{%
	It has been shown that one can improve efficiency by using an estimated propensity score even when the true score is known \citep[see, e.g.,][]{Heckman1998Matching,Hirano2003Efficient,Abadie2016Matching}.
	The intuition is that the estimation can implicitly adjust for chance imbalances between treated and control units in the sample.
	This consideration is not relevant here because the focus is on the asymptotic bias.%
}

Because the propensity score takes a prominent position, dedicated notation will expedite the discussion.
Let $\rPSpop = \pscore{\Xvpop}$ be a random variable describing the distribution of the propensity score in the population.
Similarly, let $\rPSi = \pscore{\Xvi}$ be the propensity score for unit $i$ in the sample.
The metric used for the matching is the absolute difference between units' propensity scores: $\dist{x}{x'} = \abs{\pscore{x} - \pscore{x'}}$.
That is, the optimal matching is the one that minimizes the sum of $\abs[\big]{\rPSi - \rPS{\mgen{i}}}$ over the treated units $i \in \Treated$.

A matching metric must be coupled with a restriction on the potential outcomes to give it meaning.
Matching adjustment is implicitly motivated by an assumption that units that are similar with respect to the matching metric are expected to be similar also with respect to their potential outcomes, but Condition~\ref{cond:reg-cond} only ensures that units with identical covariate values are comparable.
The following condition formalizes this assumption.

\begin{condition}[Continuity]\label{cond:pscore-lipschitz}
	$\Es{\POcpop \given \rPSpop = p}$ is Lipschitz continuous on the support of $\rPSpop$.
\end{condition}

\section{Inconsistency of the matching estimator}

The question at hand is how $\attest$ behaves when the sample is drawn from a population satisfying Conditions~\ref{cond:reg-cond} and~\ref{cond:pscore-lipschitz}.
As noted in the introduction, the reason this question may be beyond reach is that matching without replacement is not asymptotically stable.
This is in contrast to matching with replacement and many other adjustment methods.

When several treated units are allowed to be matched to the same control, small perturbations of the sample will have small effects.
For example, if we were to remove one unit and replace it with a new unit drawn from the population, the only affected units are those matched to the unit that was removed and those close to the unit that replaces it.
Asymptotically, these units will be a negligible fraction of the sample.
This stability makes an analysis tractable, which, for example, facilitated the investigation by \citet{Abadie2006Large}.

The concern when the matching is done without replacement is that such small perturbations can initiate chain effects.
This is because at most one treated unit can be matched to each control.
If we, as above, were to replace one matched control unit with a new unit drawn from the population, then the treated unit that was matched to the replaced unit needs to find a new match.
If this new match was previously matched to another treated unit, that treated unit must also find a new match.
This could potentially start a chain of units forced to find new matches, and the chain could potentially be long.
Thus, replacing just a single unit could induce large changes in the matching.
It might be reasonable to believe that such long chain effects are rare, and the set of matched controls will in any case be more stable than the matching itself.
The instability of the matching nevertheless complicates the analysis.

It turns out that the following lemma provides enough leverage to characterize the matching in large samples.

\begin{definition}\label{def:crossing-matches}
	A matching $\matchsym$ is said to contain crossing matches if
	\begin{equation}
		\maxf[\big]{\rPS{i}, \rPS{\mgen{j}}} < \minf[\big]{\rPS{j}, \rPS{\mgen{i}}}
		\qquad\text{for some}\qquad
		i, j \in \Treated.
	\end{equation}
\end{definition}

\begin{lemma}\label{lem:no-crossing-matches}
	\lemmanocrossingmatches{}
\end{lemma}

The intuition behind the lemma is that the matching objective, the sum of within-match propensity score differences, can be made smaller if two matches are crossing.
We simply need to switch the controls the two treated units are matched to.
An implication of Lemma~\ref{lem:no-crossing-matches} is that if we observe an unmatched control unit with propensity score $p$ in an optimal matching, then we know that no unit with a propensity score greater than $p$ is matched with a unit with a propensity score less than $p$.
The insight provides a way to separate the investigation into two parts, making an analysis tractable.
If there is a point $\PSpoint$ where an unmatched unit exists with high probability asymptotically, then the units to the left and right of $\PSpoint$ can be seen as two unconnected matching problems, which can be analyzed separately.

The point we are looking for is the smallest value $p$ such that at least half of the units in the population with propensity scores greater or equal to $p$ are treated.
In particular, consider $\Pr{\Wpop = 1 \given \rPSpop \geq p}$.
This is the probability that a unit with a propensity score greater or equal to $p$ is treated.
We know that this probability is greater than $p$ because $\Pr{\Wpop = 1 \given \rPSpop = p} = p$.
For example, at least half of the units with $\rPSpop \geq 1 / 2$ are treated, so $\Pr{\Wpop = 1 \given \rPSpop \geq 1 / 2} \geq 1 / 2$.
The point that partitions the matching problem is
\begin{equation}
	\PSpoint = \inf\setb[\big]{p \suchthat \Pr{\Wpop = 1 \given \rPSpop \geq p} \geq 1 / 2}.
\end{equation}
No such point exists when $\Pr{\rPSpop \geq 1 / 2} = 0$.
Let $\PSpoint = 1 / 2$ in that case, so $\PSpoint$ always is defined.

Provided that there are some units in the population with $\rPSpop \geq 1 / 2$, there will be an unmatched unit in a small neighborhood around $\PSpoint$ with high probability in large samples, and this gives us the partition we seek.
The first part consists of units with $\rPSi > \PSpoint$, and the second part consists of units with $\rPSi < \PSpoint$.
The properties of the matching are remarkably different in these two parts.
On the one hand, controls will be scarce above $\PSpoint$, so all units are matched.
This means that there here will be a limit to how much the quality of the matching can improve as the sample grows.
On the other hand, controls are abundant below $\PSpoint$, and the match quality improves without limit here, although it may be at a slow rate.

The argument does not apply to units with $\rPSi = \PSpoint$, and such units complicate the discussion without adding any profound insights.
A simple way to avoid the concern is to assume that there is no atom at $\PSpoint$, so $\Pr{\rPSpop = \PSpoint} = 0$.
The following condition does the same but is slightly weaker.
It effectively says that if there is an atom at $\PSpoint$, then we can consider those units to belong to the group with propensity scores greater than $\PSpoint$.

\begin{condition}\label{cond:left-closed}
	The set $\setb{p \suchthat \Pr{\Wpop = 1 \given \rPSpop \geq p} \geq 1 / 2}$ is left-closed or empty.
\end{condition}

\begin{proposition}\label{prop:pscore-bias}
	\proppscorebias{}
\end{proposition}

The proposition captures consequences of poor match quality among units with $\rPSi \geq \PSpoint$.
Specifically, it shows that the poor quality translates into bias.
There are two ways to achieve asymptotic unbiasedness, and thereby possibly consistency.

The first option is to assume that there is no confounding among units with $\rPSi \geq \PSpoint$:
\begin{equation}
	\Es[\big]{\POcpop \given \Wpop = 1, \rPSpop \geq \PSpoint} = \Es[\big]{\POcpop \given \Wpop = 0, \rPSpop \geq \PSpoint}.
\end{equation}
Note that these expectations condition on a range of propensity scores rather than an exact score, so Condition~\ref{cond:reg-cond} does not imply that the expectations are equal.
The second option is to assume that there are no units above $\PSpoint$, so that $\Pr{\rPSpop \geq \PSpoint} = 0$.
This is only possible if $\Pr{\rPSpop \geq 1 / 2} = 0$.
This is a strengthening of the overlap assumption in Conditions~\ref{cond:reg-cond}, requiring that $\pscore{x}$ is less than $1 / 2$ almost everywhere on the support of the covariate.
Neither option is attractive.

Some readers may object that the conditions used for Proposition~\ref{prop:pscore-bias} are weaker than those typically used in applications and claim that consistency is achieved under slightly stronger conditions.
However, the proposition holds for all populations satisfying Conditions~\ref{cond:reg-cond},~\ref{cond:pscore-lipschitz} and~\ref{cond:left-closed}.
A strengthening of the conditions will therefore not lead to consistency unless the strengthening is exactly one of the two options discussed in the previous paragraph.
Their disjunction is a necessary condition for consistency in this setting.

\section{Diminishing fraction of treated units}\label{sec:diminishing}

The paper has so far considered settings where the population is fixed throughout the asymptotic sequence.
Matching methods have also been studied in other asymptotic regimes.
For example, in addition to the regime used in this paper, \citet{Abadie2006Large} investigate matching with replacement when the treated units are a diminishing fraction of the population.
They assume that the treated and control units are sampled separately in proportions such that $\nTreated^r / \nControls \to k$ for some $r > 1$ and $k < \infty$.
In other words, they consider the case in which $\avgpscore \to 0$.

The alternative regime is a better approximation of situations in which it is considerably easier to sample additional controls than it is to sample additional treated units.
One example is when a novel medical treatment is evaluated using conventional treatment as comparison.
The only patients with $\Wpop = 1$ are those in the hospitals offering the new treatment.
By virtue of being a novel treatment, such patients are rare.
However, patients with $\Wpop = 0$ are easy to find because there are many hospitals that offer the conventional treatment.
In an imagined sequence of samples, it is here appropriate to assume that the fraction of treated units would approach zero.

\citet{Abadie2012Martingale} study matching without replacement in this regime, and they show that the matching estimator is consistent.
Their main contribution is a characterization of the large sample distribution of the estimator using a martingale representation.
The representation consists of a conditional bias term and a martingale array, of which the latter is the primary focus.
However, to achieve consistency, the conditional bias must be shown to diminish, which the authors demonstrate in an appendix.
This result critically depends on additional conditions on the asymptotic behavior of the population distribution.

In the standard asymptotic regime, all aspects of the population are fixed.
Besides $\avgpscore$, this includes the propensity score and the conditional densities of the treated and control units over the covariate set.
When we let $\avgpscore \to 0$, some of these other aspects must also change.
We cannot simultaneously hold the propensity score and the conditional densities fixed if the fraction of treated units approaches zero.
The route that \citet{Abadie2012Martingale} take is to fix the conditional densities.
The consequence is that the propensity score approaches zero everywhere on the support of the covariate.
That is, they implicitly assume that $\pscore{x} \to 0$ for almost all $x$ on the support of $\Xvpop$.

While it may be reasonable to consider the case $\avgpscore \to 0$, it may not always be reasonable to assume that $\pscore{x} \to 0$ holds everywhere.
The example with the hospitals and the novel treatment provides an illustration.
We can here sample patients with the conventional treatment much easier than we sample patients with the new treatment.
However, all these additional controls will be of a special type.
They will be patients in hospitals offering only the conventional treatment.
Among patients in hospitals offering the new treatment, a non-negligible fraction will be treated, so $\pscore{x} \to 0$ does not hold.

The illustration mirrors a sentiment that appears to be common among practitioners: when controls are abundant, most of them are not useful because they are too different from the treated units.
Put differently, it might be easy to find controls, but it is hard to find controls that are good.
A more appropriate regime might therefore be one that holds the propensity score fixed when $\avgpscore \to 0$, and adjusts the conditional densities as needed.
An inspection of the proof of Proposition~\ref{prop:pscore-bias} suggests that not much would change under this regime, except that the relevant scaling is $\avgpscore n$ rather than $n$.
This means that consistency would require $\Pr{\rPSpop \geq \PSpoint} / \avgpscore \to 0$, unless one assumes that there is no confounding among units above $\PSpoint$.

\section{Discussion}

Practitioners often see consistency as an integral property of an estimator, and the result in this paper is discouraging for matching without replacement.
The picture becomes even grimmer with the realization that the variance of the estimator may converge to zero even if the bias does not.
Confidence intervals based on estimated standard errors would in that case be dangerously misleading.
This may prompt practitioners to reconsider whether matching without replacement is appropriate for their studies.

It would, however, be too rash to categorically dismiss matching without replacement based on these results.
Practitioners may be willing to accept the bias introduced by the method in light of its benefits, such as ease of analysis and a potential reduction in variance.
Furthermore, practitioners often examine the balance between treatment groups after matching.
The asymptotic bias demonstrated above would be mirrored by an imbalance in the propensity scores between the treatment groups.
In other words, to examine whether the concern raised here applies to a specific context, one can estimate $\Pr{\rPSpop \geq 1 / 2}$ and test if it differs from zero.
Practitioners may also use calipers when they construct their matchings \citep{Cochran1973Controlling}.
The approach avoids matches of poor quality by excluding problematic units from the estimation.
Consistency would be achieved, at the cost of efficiency, if the caliper approaches zero as the sample grows, provided that the observations are weighted appropriately to account for the excluded treated units.

\citet{Abadie2011Bias} describe a bias adjustment approach when the matching is done with replacement.
The purpose is to addresses the concern raised by \citet{Abadie2006Large} that the unadjusted matching estimator sometimes converges at a slower than root-$n$ rate.
The adjustment requires estimates of the full response surfaces of the potential outcomes.
If these surfaces can be estimated consistently, a similar bias adjustment approach would address the concern raised in this paper.
However, consistency is then achieved solely because of our ability to estimate the response surfaces, and the matching step becomes largely redundant.
If the response surfaces are presumed to take some particular form, such as when they are estimated by a parametric model, then consistency would require functional form assumptions.

The feature that creates the asymptotic bias is the fact that the matching is done without replacement rather than that it is done using the propensity score.
This paper focused on propensity score matching for expositional reasons.
The supplement extends Proposition~\ref{prop:pscore-bias} to matching without replacement using arbitrary matching scores and metrics.
This includes other balancing scores, various prognostic scores as discussed by \citet{Hansen2008prognostic}, and Euclidean and Mahalanobis distances on raw or transformed covariates.
The form of the bias is similar to the form when using the propensity score, but the interpretation is somewhat more intricate, and additional regularity conditions are needed to derive the bias in closed form.
Nevertheless, the lesson is the same: matching without replacement is not generally consistent.

The focus of this paper was point estimation in the tradition of \citet{Abadie2006Large,Abadie2012Martingale}.
An alternative approach to analyze matched samples is the randomization-based inference framework.
\citet{Rosenbaum2002Observational} provides an overview.
Some features of this framework makes it difficult to judge whether the results in this paper also apply there.
For example, this literature rarely considers point estimators, and when they are considered, they are of the Hodges--Lehmann-type, which differ greatly from the type of estimator considered here.
Inherent to the permutation approach used in this literature is also that the analysis is conditional on the matching.
Still, the concerns raised here should motivate practitioners to be cautious also when working in the randomization-based framework.
Using the terminology of \citet{Rosenbaum2002Observational}, this paper shows that there can be substantial ``overt bias'' even if the sample is large and all conventional matching conditions hold.
Therefore, practitioners should make sure to follow the recommendations by \citet{Rosenbaum2002Observational} and others on how to address overt bias no matter the size of their sample.
Furthermore, practitioners should be wary about theoretical investigations that condition on the matching and assume the quality of the matching will improve without limit as the sample grows.

\begin{singlespace}
\bibliography{part-ref}
\end{singlespace}

\clearpage

\mtcaddpart[Supplement]

\newcommand\resetsuppcounters{%
\setcounter{section}{0}%
\setcounter{table}{0}%
\setcounter{figure}{0}%
\setcounter{equation}{0}%
\setcounter{conjecture}{0}%
\setcounter{corollary}{0}%
\setcounter{lemma}{0}%
\setcounter{proposition}{0}%
\setcounter{theorem}{0}%
\setcounter{assumption}{0}%
\setcounter{condition}{0}%
\setcounter{definition}{0}%
\setcounter{remark}{0}%
\setcounter{example}{0}%
}

\newcommand\updatesuppcounters[1]{%
\renewcommand\thesection{#1\arabic{section}}%
\renewcommand\thetable{#1\arabic{table}}%
\renewcommand\thefigure{#1\arabic{figure}}%
\renewcommand\theequation{#1\arabic{equation}}%
\renewcommand\theconjecture{#1\arabic{conjecture}}%
\renewcommand\thecorollary{#1\arabic{corollary}}%
\renewcommand\thelemma{#1\arabic{lemma}}
\renewcommand\theproposition{#1\arabic{proposition}}
\renewcommand\thetheorem{#1\arabic{theorem}}
\renewcommand\theassumption{#1\arabic{assumption}}
\renewcommand\thecondition{#1\arabic{condition}}
\renewcommand\thedefinition{#1\arabic{definition}}
\renewcommand\theremark{#1\arabic{remark}}
\renewcommand\theexample{#1\arabic{example}}
}

\newcommand\updatesuppanchors[1]{
\renewcommand\theHsection{#1.\arabic{section}}
\renewcommand\theHtable{#1.\arabic{table}}
\renewcommand\theHfigure{#1.\arabic{figure}}
\renewcommand\theHequation{#1.\arabic{equation}}
\renewcommand\theHconjecture{#1.\arabic{conjecture}}
\renewcommand\theHcorollary{#1.\arabic{corollary}}
\renewcommand\theHlemma{#1.\arabic{lemma}}
\renewcommand\theHproposition{#1.\arabic{proposition}}
\renewcommand\theHtheorem{#1.\arabic{theorem}}
\renewcommand\theHassumption{#1.\arabic{assumption}}
\renewcommand\theHcondition{#1.\arabic{condition}}
\renewcommand\theHdefinition{#1.\arabic{definition}}
\renewcommand\theHremark{#1.\arabic{remark}}
\renewcommand\theHexample{#1.\arabic{example}}
}

\resetsuppcounters
\updatesuppanchors{S}
\updatesuppcounters{S}

\begin{center}
	\Huge \textbf{Supplement}
\end{center}

\vspace{1in}

\parttoc

\clearpage

\section{Matching with arbitrary scores and metrics}

\subsection{Extension of main proposition}

\newcommand{\Svpop}{S}
\newcommand{\Sv}[1]{\Svpop_{#1}}
\newcommand{\Svi}{\Sv{i}}
\newcommand{\Svj}{\Sv{j}}

\newcommand{\allS}{\set{S}}
\newcommand{\suppS}{\allS_{\normalfont\textsc{supp}}}
\newcommand{\allSp}{\allS_{+}}
\newcommand{\optS}{\allS^{*}}
\newcommand{\optSc}{\allS^\comp}

\newcommand{\Q}{\set{Q}}
\newcommand{\Qc}{\set{Q}^\comp}

\newcommand{\allSuperS}{\set{A}}

\newcommand{\distSsym}{d_S}
\DeclarePairedDelimiterXPP\distS[2]{\distSsym}{\lparen}{\rparen}{}{#1, #2}

\newcommand{\spoint}{s^\dagger}

The main paper showed that the estimator of the average treatment effect of the treated was not generally consistent when matching without replacement using the propensity score.
That is, the metric $\distsym : \allX \times \allX \to \Reals^+$ used for the matching was the absolute differences in propensity scores: $\dist{x}{x'} = \abs{\pscore{x} - \pscore{x'}}$.
This section extends this result to matching without replacement using arbitrary metrics over arbitrary scores.
In the interest of space, the proof of the extension will be somewhat informal.

Let $s \colon \allX \to \allS$ be a score function summarizing the covariates.
This function could be any type of dimensionality reduction, including the propensity score, other balancing scores and prognostic scores.
The score can be multidimensional, so that $\func{s}{\Xvpop}$ is a vector.
The score can also be the identity function, so that $\Xvpop = \func{s}{\Xvpop}$.
Let $\Svpop = \func{s}{\Xvpop}$ be a random variable describing the distribution of the score in the population.
For convenience, we assume that the score $\Svpop$ is continuously distributed.
Let $\suppS$ be the support of $\Svpop$.

The score is associated with some metric $\distSsym : \allS \times \allS \to \Reals^+$ that captures similarity between different scores.
If the score is scalar, this will often be the absolute difference $\distS{s}{s'} = \abs{s - s'}$.
For multidimensional scores, it might be Euclidean or Mahalanobis distances.
However, the current discussion is not restricted to these particular metrics; it applies to any metric satisfying the conditions below.
Note that the score function and its associated metric induces a pseudometric on the raw covariates: $\dist{x}{x'} = \distS{\func{s}{x}}{\func{s}{x'}}$.

We will extend the standard matching assumptions specified in Conditions~\refmain{cond:reg-cond}~and~\refmain{cond:pscore-lipschitz} in the main paper to the score $\Svpop$ considered here.
This is formalized in the following three conditions.

\begin{condition}\label{cond:reg-cond-arb-score}
	The population satisfies:
	\vspace{-3pt}
	\begin{enumerate}[label=\roman*.,ref=\ref*{cond:reg-cond-arb-score}.\roman*]
		\item Unconfoundedness: $\POcpop$ is conditionally independent of $\Wpop$ given $\Svpop = s$ on $\suppS$.
		\item Overlap: $\Pr{\Wpop = 1 \given \Svpop = s}$ is bounded away from one on $\suppS$.
		\item Existence of treated units: $\Pr{\Wpop = 1}$ is greater than zero.
		\item Abundance of control units: $\Pr{\Wpop = 1}$ is less than one-half.\label{cond:arb-score-abundance-controls}
		\item Well-behaved outcomes: $\E{\Ypop^2}$ exists.
	\end{enumerate}
\end{condition}

\begin{condition}[Continuity]\label{cond:arb-lipschitz}
	$\Es{\POcpop \given \Svpop = s}$ is Lipschitz continuous on $\suppS$ with respect to metric $\distSsym$.
	That is, there exists a constant $c$ such that for all $s,s' \in \suppS$,
	\begin{equation}
		\abs[\big]{\Es{\POcpop \given \Svpop = s} - \Es{\POcpop \given \Svpop = s'}} \leq c \times \distS{s}{s'}.
	\end{equation}
\end{condition}

\begin{condition}[Assignment continuity]\label{cond:assign-lipschitz}
	$\Pr{\Wpop = 1 \given \Svpop = s}$ is Lipschitz continuous on $\suppS$ with respect to metric $\distsym_S$.
\end{condition}

Conditions~\ref{cond:reg-cond-arb-score}~and~\ref{cond:arb-lipschitz} correspond directly to Conditions~\refmain{cond:reg-cond}~and~\refmain{cond:pscore-lipschitz} in the main paper, but there is no condition in the main paper that corresponds to Condition~\ref{cond:assign-lipschitz}.
Indeed, the condition is satisfied by construction when $\Svpop$ is the propensity score, because $\Pr{\Wpop = 1 \given \Svpop = s} = s$ in that case.
Continuity of the conditional assignment probability function does not hold by construction for other scores, and it is therefore explicitly imposed here.

The purpose of Condition~\ref{cond:assign-lipschitz} is to rule out situations where the assignment mechanism is highly fragmented with respect to the score, so that minute perturbations of the score could lead to large changes in the probability of being treated.
Such fragmented assignment mechanisms could facilitate identification using local extrapolation rather than the conventional overlap motivation.
Indeed, consider a scalar, real-valued score $\Svpop$ and an assignment mechanism for which $\Pr{\Wpop = 1 \given \Svpop = s} = 0$ if $s$ is a rational number and $\Pr{\Wpop = 1 \given \Svpop = s} = 1$ otherwise.
Overlap does not hold in this case, but identification of the average treatment effect for the treated is possible because the rationals are dense in the reals, so Condition~\ref{cond:arb-lipschitz} allows us to extrapolate $\Es{\POcpop \given \Svpop = s}$ from the rationals to the remaining real numbers.
Given an appropriate marginal distribution of $\Svpop$, matching without replacement could be consistent in this setting.
Condition~\ref{cond:assign-lipschitz} rules out this alternative motivation for matching, instead focusing the discussion on the conventional setting.
The condition can be weakened at the cost of additional complexity.

To extend the result in the main paper to arbitrary matching scores, we will consider subsets of $\suppS$ defined based on the local probability of being treated.
Let
\begin{equation}
	\allSp = \braces[\Big]{ s \in \suppS : \Pr{\Wpop = 1 \given \Svpop = s} \geq 1/2}
\end{equation}
collect all scores for which the corresponding conditional probability of being treated is one half or greater.
Let $\allSuperS$ collect all supersets of $\allSp$ for which exactly one half of the constituent units are treated.
That is, $\allSuperS$ contains all sets $\Q \subseteq \suppS$ such that
\begin{equation}
\allSp \subseteq \Q
\qquad\text{and}\qquad
\Pr{\Wpop = 1 \given \Svpop \in \Q} = 1/2.
\end{equation}
Condition~\ref{cond:arb-score-abundance-controls} together with continuity of $\Svpop$ ensure that $\allSuperS$ is not empty.

For any $\Q \in \allSuperS$, let $\nu_{\Q,w}$ be the density of the score $\Svpop$ conditional on $\Svpop \in \set{Q}$ and $\Wpop = w$.
That is, if $f$ is the probability density function of $\Svpop$, then
\begin{equation}
	\func{\nu_{\set{Q},w}}{s} = \func[\big]{f}{s \mid \Svpop \in \Q, \Wpop = w}.
\end{equation}
For any $\Q \in \allSuperS$, let $\func{\Gamma}{\Q}$ collect all couplings between $\nu_{\Q,1}$ and $\nu_{\Q,0}$.
That is, $\func{\Gamma}{\Q}$ contains all joint densities over $\Q \times \Q$ such that the two marginal densities of each $\gamma \in \func{\Gamma}{\Q}$ equal $\nu_{\Q,1}$ and $\nu_{\Q,0}$, respectively.
This allows us to define the Wasserstein distance of order one between $\nu_{\Q,1}$ and $\nu_{\Q,0}$:
\begin{equation}
	\func{\mathcal{W}}{\Q} = \inf_{\gamma \in \func{\Gamma}{\Q}} \int_{\Q \times \Q} \distS{s}{s'} \mathrm{d} \gamma.
\end{equation}
We will consider a lower bound on a weighted version of the Wasserstein distance in the set $\allSuperS$:
\begin{equation}
	\mathcal{W}^* = \inf_{\Q \in \allSuperS} \Pr{\Svpop \in \Q \given \Wpop = 1} \times \func{\mathcal{W}}{\Q}.
\end{equation}
The weight $\Pr{\Svpop \in \Q \given \Wpop = 1}$ gives the proportion of treated units that are in $\Q$.

\begin{condition}[Existence of minimum]\label{cond:arb-left-closed}
	There exists some $\optS \in \allSuperS$ that attains the infimum, so that $\Pr{\Svpop \in \optS \given \Wpop = 1} \times \func{\mathcal{W}}{\optS} = \mathcal{W}^*$.
\end{condition}

Condition~\ref{cond:arb-left-closed} is akin to Condition~\refmain{cond:left-closed} in the main paper, which allowed us to disregard situations where the boundary of the partition needed to be split between the two parts of the partition.
Similar to Condition~\refmain{cond:left-closed}, it is possible to remove this assumption, at the cost of considerably technical complexity.
As this added complexity would bring no interesting insights, we proceed under the assumption that $\optS$ exists.

We are now ready to state the extension of Proposition~\refmain{prop:pscore-bias}.

\begin{proposition}\label{prop:extended-main-res}
	Given Conditions~\ref{cond:reg-cond-arb-score}, \ref{cond:arb-lipschitz}, \ref{cond:assign-lipschitz} and \ref{cond:arb-left-closed}, when the matching is constructed without replacement using the score $\Svpop$ and metric $\distSsym$,
	\begin{equation}
		\lim_{n \to \infty} \E{\attest} = \att + \frac{\Pr{\Svpop \in \optS}}{2\avgpscore} \bracket[\Big]{\Es[\big]{\POcpop \given \Wpop = 1, \Svpop \in \optS} - \Es[\big]{\POcpop \given \Wpop = 0, \Svpop \in \optS}}.
	\end{equation}
\end{proposition}

\begin{proof}[Sketch of proof]
The proof uses a partitioning argument similar to the one used in the main paper.
Consider constructing a matching in the population.
The set $\allSp$ is the region of $\suppS$ where, at every point, there are at least as many treated units as control units.
Unless $\Pr{\Wpop = 1 \given \Svpop = s} = 1/2$ for all $s \in \allSp$, we have that $\Pr{\Wpop = 1 \given \Svpop \in \allSp} > 1/2$.
That is, there will be an abundance of treated units in the region $\allSp$.

Because each treated unit must be matched to a unique control unit, there will not be enough control units in $\allSp$.
In other words, an abundance of treated units in $\allSp$ means that some of these units must be matched to control units outside of $\allSp$.
In particular, a share
\begin{equation}
	\frac{\Pr{\Wpop = 1 \given \Svpop \in \allSp} - \Pr{\Wpop = 0 \given \Svpop \in \allSp}}{\Pr{\Wpop = 1 \given \Svpop \in \allSp}}
\end{equation}
of the treated units in $\allSp$ will be matched to control units outside of $\allSp$.

\newcommand{\allSc}{\allS_{c}}

Consider which control units these excess treated units will be matched to.
Let $\allSc \subset \suppS$ be some set that does not overlap with $\allSp$, so that $\allSp \cap \allSc = \emptyset$.
Because $\allSp$ contains all points with $\Pr{\Wpop = 1 \given \Svpop = s} \geq 1/2$, the set $\allSc$ must contain points for which $\Pr{\Wpop = 1 \given \Svpop = s} < 1/2$.
Hence, there will be an abundance of control units in $\allSc$.
Naively, we could take all control units in $\allSc$ and match with the excess treated units in $\allSp$ to address the problem that there is an abundance of treated units in $\allSp$.
If there are sufficiently many control units in $\allSc$, so that
\begin{equation}
	\Pr{\Wpop = 0, \Svpop \in \allSc} \geq \Pr{\Wpop = 1, \Svpop \in \allSp} - \Pr{\Wpop = 0, \Svpop \in \allSp},
\end{equation}
then we can match all treated units in $\allSp$ with control units in either $\allSp$ or $\allSc$.

However, unless there are only control units in $\allSc$, so that $\Pr{\Wpop = 1 \given \Svpop \in \allSc} = 0$, this naive approach will potentially leave treated units in $\allSc$ unmatched.
Consider the union $\Q = \allSp \cup \allSc$.
Following the naive approach, this union may still contain an abundance of treated units, because the treated units in $\allSc$ must also be matched, but they cannot be matched with control units in $\allSc$ if all those units have been matched with treated units in $\allSp$.
Thus, to properly address the problem that there is an abundance of treated units in $\allSp$, we need to find a set $\allSc$ such that the union $\Q = \allSp \cup \allSc$ consists of exactly half treated units:
\begin{equation}
\Pr{\Wpop = 1 \given \Svpop \in \Q} = 1/2.
\end{equation}
The set $\allSuperS$ defined above collects all such unions $\Q$.
In order words, $\allSuperS$ collects all possible solutions to the problem that there is an abundance of treated units in $\allSp$.

We will now break the overall matching problem into two steps.
The first step is to find a partition of $\suppS$ into sets $\Q$ and $\Qc$, such that $\Q \in \allSuperS$.
The second step is to find the optimal matching within each of the two sets.
If there are no matches in the optimal matching bridging $\Q$ and $\Qc$, so that no unit in $\Q$ is matched with a unit in its complement $\Qc = \suppS \setminus \Q$, then we can consider these two matching problems separately.

\newcommand{\allSuperSalt}{\set{B}}

For this approach to work, we must verify that the optimal solution can be partitioned in this way.
To show this, let $\allSuperSalt$ collect all sets $\Q \subseteq \suppS$ satisfying three properties.
The first property is that $\allSp$ is a subset of $\Q$.
The second property is that at most half of the units in $\Q$ are treated: $\Pr{\Wpop = 1 \given \Svpop \in \Q} \leq 1/2$.
The third property is that no unit in $\Q$ is matched to a unit in $\Qc$ in the optimal matching solution.
Note that $\allSuperSalt$ is non-empty; for example, Condition~\ref{cond:arb-score-abundance-controls} ensures that $\suppS \in \allSuperSalt$.
We will show that $\allSuperSalt$ contains one set $\Q$ with $\Pr{\Wpop = 1 \given \Svpop \in \Q} = 1/2$, which implies $\Q \in \allSuperS$, as desired.

Consider some partition $\Q \in \allSuperSalt$ for which $\Pr{\Wpop = 1 \given \Svpop \in \Q} < 1/2$.
We will show that there is a set $\Q' \subset \Q$ such that $\Q' \in \allSuperSalt$.
Because $\Pr{\Wpop = 1 \given \Svpop \in \Q} < 1/2$, there will be a region $\set{R} \in \Q$ with unmatched control units in an infinitesimally small neighborhood around every point $s \in \set{R}$.
Let $M$ denote whether a control unit is matched.
That is, $M = 1$ denotes a control unit that is matched with a treated unit; for example, $\Pr{M = 0 \given \Wpop = 0}$ is the share of control units that are not matched.
The region $\set{R} \in \Q$ is such that $\Pr{M = 0, \Wpop = 0 \given \Svpop = s} > 0$ for all $s \in \set{R}$.

Condition~\ref{cond:assign-lipschitz} ensures that there exists a $\set{R}$ that is connected.
Note that all treated units in $\set{R}$ are matched with control units infinitesimally close to them.
If they were not, we could improve the matching by changing their matches to infinitesimally close unmatched control units, which we know exists by construction of $\set{R}$.
Hence, if any control units in $\set{R}$ are matched to treated units that are not infinitesimally close, they must be matched to units outside of $\set{R}$.
Consider the share of control units in $\set{R}$ that are matched with treated units outside of $\set{R}$:
\begin{equation}
\Pr{M = 1, \Svpop \in \set{R}} - \Pr{\Wpop = 1, \Svpop \in \set{R}}.
\end{equation}
Note that this difference cannot be negative, because all treated units in $\set{R}$ are matched with control units in $\set{R}$.
We will now show that the difference cannot be positive.

If the difference is positive, consider a subset $\set{K} \subset \set{R}$ such that
\begin{equation}
\Pr{M = 0, \Wpop = 0, \Svpop \in \set{R}}
= \Pr{\Wpop = 0, \Svpop \in \set{K}} - \Pr{\Wpop = 1, \Svpop \in \set{K}}.
\end{equation}
The left hand side is the share of unmatched control units in $\set{R}$.
The right hand side is the share of excess control units in $\set{K}$; that is, the leftover control units after all treated units in $\set{K}$ have been matched with control units in $\set{K}$.

Because $\set{R}$ is connected and all matched control units in $\set{R}$ are matched with treated units outside of $\set{R}$ if they are not matched with treated units that are not infinitesimally close, we can pick $\set{K}$ to be the region farthest from the excess treated units outside of $\set{R}$.
For example, if $\Svpop$ is uniformly distributed, $\set{R} = \braces{s' : \distS{s'}{s} \leq r}$ is a ball centered at $s$ with $r$ as radius, $\Pr{\Wpop = 1 \given \Svpop = s}$ is constant in $\set{R}$, and the treated units outside of $\set{R}$ that are matched with control units in $\set{R}$ are evenly distributed around $\set{R}$, then $\set{K} = \braces{s' : \distS{s'}{s} \leq r'}$ will be a ball centered at $s$ with some radius $r' < r$.
Because $\set{K}$ is farthest from the treated units outside of $\set{R}$, no control unit in $\set{K}$ will be matched to treated units outside of $\set{K}$, so
\begin{equation}
\Pr{M = 1, \Svpop \in \set{K}} = \Pr{\Wpop = 1, \Svpop \in \set{K}}.
\end{equation}
But this would imply that all control units that are in $\set{R} \setminus \set{K}$ are matched:
\begin{equation}
\Pr{\Wpop = 0, \Svpop \in \set{R} \setminus \set{K}} = \Pr{M = 1, \Svpop \in \set{R} \setminus \set{K}}.
\end{equation}
To see this, write
\begin{align}
	\Pr{\Wpop = 0, \Svpop \in \set{R} \setminus \set{K}} &= \Pr{\Wpop = 0, \Svpop \in \set{R}} - \Pr{\Wpop = 0, \Svpop \in \set{K}},
	\\
	\Pr{M = 1, \Svpop \in \set{R} \setminus \set{K}} &= \Pr{M = 1, \Svpop \in \set{R}} - \Pr{M = 1, \Svpop \in \set{K}},
\end{align}
so that
\begin{multline}
\Pr{\Wpop = 0, \Svpop \in \set{R} \setminus \set{K}} - \Pr{M = 1, \Svpop \in \set{R} \setminus \set{K}}
\\
= \Pr{\Wpop = 0, \Svpop \in \set{R}} - \Pr{M = 1, \Svpop \in \set{R}} + \Pr{M = 1, \Svpop \in \set{K}} - \Pr{\Wpop = 0, \Svpop \in \set{K}}.
\end{multline}
Note that
\begin{equation}
\Pr{\Wpop = 0, \Svpop \in \set{R}} - \Pr{M = 1, \Svpop \in \set{R}} = \Pr{M = 0, \Wpop = 0, \Svpop \in \set{R}},
\end{equation}
which by construction of $\set{K}$ is equal to
\begin{equation}
\Pr{\Wpop = 0, \Svpop \in \set{K}} - \Pr{\Wpop = 1, \Svpop \in \set{K}}.
\end{equation}
Furthermore, $\set{K}$ was constructed so that
\begin{equation}
\Pr{M = 1, \Svpop \in \set{K}} = \Pr{\Wpop = 1, \Svpop \in \set{K}}.
\end{equation}
It then follows that
\begin{multline}
\Pr{\Wpop = 0, \Svpop \in \set{R} \setminus \set{K}} - \Pr{M = 1, \Svpop \in \set{R} \setminus \set{K}}
\\
= \Pr{\Wpop = 0, \Svpop \in \set{K}} - \Pr{\Wpop = 1, \Svpop \in \set{K}} + \Pr{\Wpop = 1, \Svpop \in \set{K}} - \Pr{\Wpop = 0, \Svpop \in \set{K}},
\end{multline}
which is zero.

We have concluded that if
\begin{equation}
\Pr{M = 1, \Svpop \in \set{R}} - \Pr{\Wpop = 1, \Svpop \in \set{R}} > 0,
\end{equation}
then
\begin{equation}
\Pr{\Wpop = 0, \Svpop \in \set{R} \setminus \set{K}} = \Pr{M = 1, \Svpop \in \set{R} \setminus \set{K}}.
\end{equation}
However, this contradicts the definition of $\set{R}$, which stipulated that there were unmatched control units in an infinitesimally small neighborhood around every point $s \in \set{R}$.
The only possibility is therefore that $\set{K} = \set{R}$, but $\set{K}$ was such that
\begin{equation}
\Pr{M = 1, \Svpop \in \set{K}} = \Pr{\Wpop = 1, \Svpop \in \set{K}},
\end{equation}
which again contradicts
\begin{equation}
\Pr{M = 1, \Svpop \in \set{R}} - \Pr{\Wpop = 1, \Svpop \in \set{R}} > 0,
\end{equation}
because $\set{K} = \set{R}$.

We have now showed that for any $\Q \in \allSuperSalt$ such that $\Pr{\Wpop = 1 \given \Svpop \in \Q} < 1/2$, there exists a $\set{R}\in \Q$ with the properties stipulated above and
\begin{equation}
\Pr{M = 1, \Svpop \in \set{R}} = \Pr{\Wpop = 1, \Svpop \in \set{R}}.
\end{equation}
In other words, no unit in $\set{R}$ is matched with a unit in $\Q' = \Q \setminus \set{R}$.
We will now consider $\Q'$ in lieu of $\Q$.

Recall that $\set{R}$ contains an excess of control units.
Therefore, when we remove $\set{R}$ from $\Q$ to construct $\Q'$, we have
\begin{equation}
\Pr{\Wpop = 1 \given \Svpop \in \Q} < \Pr{\Wpop = 1 \given \Svpop \in \Q'},
\end{equation}
and because no treated unit in $\Q'$ is matched with a control unit outside of $\Q'$, we also have
\begin{equation}
\Pr{\Wpop = 1 \given \Svpop \in \Q'} \leq 1/2.
\end{equation}
It follows that $\Q' \in \allSuperSalt$.
Continuity of $\Svpop$ implies that we can apply this procedure recursively until we reach a $\Q'$ such that
\begin{equation}
\Pr{\Wpop = 1 \given \Svpop \in \Q'} = 1/2,
\end{equation}
which implies that $\Q' \in \allSuperSalt$ and $\Q' \in \allSuperS$, as desired.

We have now showed that we can reformulate the overall matching problem in the population into two parts.
The first part is to find a partition of $\suppS$ into sets $\Q$ and $\Qc$, such that $\Q \in \allSuperS$, and the second part is to find an optimal matching within each of the sets $\Q$ and $\Qc$ separately.
We will consider this problem in reverse.

Starting with the second step, we can use transport theory to characterize the optimal matching in $\Q$.
We use a bijection between the treated and control units in $\Q$ to describe the matching between the units.
For infinite populations, we can describe such a matching as a coupling between the conditional distributions of the score within $\Q$.
Let $\func{\nu_{\set{Q},w}}{s} = \func{f}{s \mid \Svpop \in \Q, \Wpop = w}$ be the conditional density of $\Svpop$ for units in $\Q$ that are assigned treatment $\Wpop = w$.
Consider a density $\gamma$ over $\Q \times \Q$ with the property that its marginals are equal to $\nu_{\set{Q},1}$ and $\nu_{\set{Q},0}$.
The matching objective for the optimal matching within $\Q$ is therefore given by
\begin{equation}
	\func{\mathcal{W}}{\Q} = \inf_{\gamma \in \func{\Gamma}{\Q}} \int_{\Q \times \Q} \distS{s}{s'} \mathrm{d} \gamma.
\end{equation}

It is straightforward to characterize the optimal matching in $\Qc = \suppS \setminus \Q$.
Recall that no control unit in this region is matched with any treated unit in $\Q$.
Furthermore, because $\allSp \subseteq \Q$, we have that $\Pr{\Wpop = 1 \given \Svpop = s} < 1/2$ for all $s \in \Qc$.
That is, at every point in $\Qc$, there is an abundance of control units in a small neighborhood around that point.
This means that all treated units in $\Qc$ can be matched to a control unit that is infinitesimally close with respect to the score.
For this reason, the optimal matching in $\Qc$ will attain a value of the matching objective that is zero.

Continuing with the first step, which was to find the partition into $\Q$ and $\Qc$, the overall matching objective is the weighted sum of the matching objective within $\Q$, with weight $\Pr{\Svpop \in \Q \given \Wpop = 1}$, and the matching objective within $\Qc$, with weight $\Pr{\Svpop \not\in \Q \given \Wpop = 1}$.
In other words, the overall matching objective is
\begin{equation}
	\Pr{\Svpop \in \Q \given \Wpop = 1} \times \func{\mathcal{W}}{\Q} + \Pr{\Svpop \not\in \Q \given \Wpop = 1} \times 0,
\end{equation}
which is lower bounded in $\allSuperS$ by
\begin{equation}
	\mathcal{W}^* = \inf_{\Q \in \allSuperS} \Pr{\Svpop \in \Q \given \Wpop = 1} \times \func{\mathcal{W}}{\Q}.
\end{equation}
Condition~\ref{cond:arb-left-closed} states that this lower bound is attainable by some $\optS \in \allSuperS$:
\begin{equation}
	\Pr{\Svpop \in \optS \given \Wpop = 1} \times \func{\mathcal{W}}{\optS} = \mathcal{W}^*.
\end{equation}
This means that $\optS$ is the partition of the optimal matching.

Note that the partition of $\suppS$ into $\optS$ and $\optSc = \suppS \setminus \optS$ resembles the partition in the main paper.
In particular, all control units in $\optS$ are matched with treated units, while there is an abundance of control units at every point in $\optSc$.
The behavior of the matching is therefore radically different in $\optS$ and $\optSc$.
We can use this fact to characterize the bias.

As above, let $M$ denote whether a control unit is matched.
The population version of the matching estimator considered in the main paper is
\begin{equation}
	\E{\Ypop \given \Wpop = 1} - \E{\Ypop \given M = 1} = \Es{\POtpop \given \Wpop = 1} - \Es{\POcpop \given M = 1},
\end{equation}
which gives a population bias with respect to $\att = \Es{\POtpop - \POcpop\given \Wpop = 1}$ equal to
\begin{equation}
	\Es{\POcpop\given \Wpop = 1} - \Es{\POcpop \given M = 1}.
\end{equation}
Use the law of total expectation and the partition of $\suppS$ into $\optS$ and $\optSc$ to decompose the terms of the bias:
\begin{align}
	\Es{\POcpop\given \Wpop = 1} &= \Pr{\Svpop \in \optS \given \Wpop = 1} \Es[\big]{\POcpop\given \Wpop = 1, \Svpop \in \optS}
	\\
										  &\qquad\qquad+ \Pr{\Svpop \in \optSc \given \Wpop = 1} \Es[\big]{\POcpop\given \Wpop = 1, \Svpop \in \optSc},
										  \\[0.5em]
	\Es{\POcpop \given M = 1} &= \Pr{\Svpop \in \optS \given M = 1} \Es[\big]{\POcpop\given M = 1, \Svpop \in \optS}
	\\
										  &\qquad\qquad+ \Pr{\Svpop \in \optSc \given M = 1} \Es[\big]{\POcpop\given M = 1, \Svpop \in \optSc}.
\end{align}
The fact that no unit in $\optS$ is matched with a unit in $\optSc$ implies that
\begin{equation}
	\Pr{\Svpop \in \optS \given \Wpop = 1} = \Pr{\Svpop \in \optS \given M = 1}
\end{equation}
and
\begin{equation}
	\Pr{\Svpop \in \optSc \given \Wpop = 1} = \Pr{\Svpop \in \optSc \given M = 1}.
\end{equation}
We can therefore write the bias as
\begin{multline}
	\Pr{\Svpop \in \optS \given \Wpop = 1} \bracket[\Big]{ \Es[\big]{\POcpop\given \Wpop = 1, \Svpop \in \optS} - \Es[\big]{\POcpop\given M = 1, \Svpop \in \optS} },
\\
+
	\Pr{\Svpop \in \optSc \given \Wpop = 1} \bracket[\Big]{ \Es[\big]{\POcpop\given \Wpop = 1, \Svpop \in \optSc} - \Es[\big]{\POcpop\given M = 1, \Svpop \in \optSc} }.
\end{multline}

Recall that each treated unit in $\optSc$ will be matched to a control unit that is infinitesimally close.
This fact together with Condition~\ref{cond:arb-lipschitz} implies that
\begin{equation}
	\Es[\big]{\POcpop\given \Wpop = 1, \Svpop \in \optSc} = \Es[\big]{\POcpop\given M = 1, \Svpop \in \optSc},
\end{equation}
so the second term of the bias is zero.
Also recall that all control units in $\optS$ are matched, which implies that
\begin{equation}
	\Es[\big]{\POcpop\given M = 1, \Svpop \in \optS} = \Es[\big]{\POcpop\given \Wpop = 0, \Svpop \in \optS}.
\end{equation}
Finally, use Bayes' rule to write
\begin{equation}
	\Pr{\Svpop \in \optS \given \Wpop = 1}
	= \frac{\Pr{\Svpop \in \optS} \Pr{\Wpop = 1 \given \Svpop \in \optS }}{\Pr{\Wpop = 1}}
	= \frac{\Pr{\Svpop \in \optS}}{2 \avgpscore},
\end{equation}
which follows from $\Pr{\Wpop = 1 \given \Svpop \in \optS } = 1/2$ and $\Pr{\Wpop = 1} = \avgpscore$.

Taken together, we have showed that the bias in the population is
\begin{equation}
	\frac{\Pr{\Svpop \in \optS}}{2 \avgpscore} \bracket[\Big]{ \Es[\big]{\POcpop\given \Wpop = 1, \Svpop \in \optS} - \Es[\big]{\POcpop\given \Wpop = 0, \Svpop \in \optS} },
\end{equation}
which is the statement in the proposition.
\end{proof}

As mentioned in the beginning of this section, the proof of Proposition~\ref{prop:extended-main-res} is informal compared to the proof of the result in the main paper presented below.
The most notable omission is that the proof here solely considers matching in the population.
To make the proof rigorous and complete, one would need to show that the sequence of matchings in the samples convergences in an appropriate sense to the population matching.
The bulk of the proof of Proposition~\refmain{prop:pscore-bias} consists of showing this type of convergence, and it has been omitted in the proof of the extension in the interest of space.

\subsection{Example: Prognostic score}

This section provides an example that illustrates the result in the previous section.
The focus is when the score used for to construct the matching is a prognostic score as described by \citet{Hansen2008prognostic}.

Consider a three-dimensional covariate $\Xvpop = \paren{\Xvpop_1, \Xvpop_2, \Xvpop_3}$, for which each coordinate is uniformly distributed on $[0, 1]$ and independent of the other two.
Let the propensity score be the product of the first and second covariates when the second covariate is taken to the power of $a \geq 1/3$:
\begin{equation}
	\pscore{x} = \Pr{\Wpop = 1 \given \Xvpop = x} = x_1x_2^a.
\end{equation}
Consider a score function $s \colon \allX \to \allS$ given by $\func{s}{x} = x_1 + x_3$, so that $\Svpop = \Xvpop_1 + \Xvpop_3$ with probability one.
The density of $\Svpop$ is $\func{f}{s} = s$ when $s \leq 1$, and $\func{f}{s} = 2 - s$ when $s > 1$.
Note that $\suppS = [0, 2]$.
The metric for the score is the absolute difference: $\distS{s}{s'} = \abs{s - s'}$.

Let the potential outcomes be distributed as
\begin{equation}
	\POcpop = \Svpop^2 + \varepsilon_0
	\qquad\text{and}\qquad
	\POtpop = 5/2 + \varepsilon_1,
\end{equation}
where $\varepsilon_0$ and $\varepsilon_1$ are independent standard normal variates.
This means that the average treatment effect of the treated is one, $\att = 1$, irrespectively of $a$.
Note that $\Svpop$ is a prognostic score: $\POcpop \indep \Xvpop \mid \Svpop$.

We will now verify that this data generating process satisfies Conditions~\ref{cond:reg-cond-arb-score},~\ref{cond:arb-lipschitz} and \ref{cond:assign-lipschitz}.
Unconfoundedness follows directly from the fact that $\Svpop$ is a prognostic score.
To verify overlap, consider the probability of being treated given the score:
\begin{equation}
	\Pr{\Wpop = 1 \given \Svpop = s}
	= \E{\Wpop \given \Svpop = s}
	= \E{\Xvpop_1\Xvpop_2^a \given \Svpop = s}
	= \E{\Xvpop_2^a} \E{\Xvpop_1 \given \Svpop = s},
\end{equation}
where the first equality follows from binary $\Wpop$, the second equality follows from the definition of the propensity score, and the third equality follows from that $\Xvpop_1$, $\Xvpop_2$ and $\Xvpop_3$ are mutually independent.
Note that
\begin{equation}
	\E{\Xvpop_2^a}
	= \int_0^1 u^a \diff u = \frac{1}{a + 1}.
\end{equation}
To derive $\E{\Xvpop_1 \given \Svpop = s}$, note that $\Xvpop_1 = \Svpop - \Xvpop_3$ with probability one, so we can write
\begin{equation}
	\E{\Xvpop_1 \given \Svpop = s}
	= \E{\Svpop - \Xvpop_3 \given \Svpop = s}
	= s - \E{\Xvpop_3 \given \Svpop = s}.
\end{equation}
We have $\E{\Xvpop_1 \given \Svpop = s} = \E{\Xvpop_3 \given \Svpop = s}$ because of symmetry, so $\E{\Xvpop_1 \given \Svpop = s} = s / 2$.
Taken together, we have
\begin{equation}
	\Pr{\Wpop = 1 \given \Svpop = s}
	= \E{\Xvpop_2^a} \E{\Xvpop_1 \given \Svpop = s}
	= \frac{s}{2a + 2}.
\end{equation}
Note that $\Pr{\Wpop = 1 \given \Svpop = s}$ is increasing in $s$, and that $\Svpop$ takes its highest value when $\Xvpop_1 = \Xvpop_3 = 1$, in which case $\Svpop = 2$.
The conditional probability $\Pr{\Wpop = 1 \given \Svpop = s}$ is therefore at most $1 / \paren{ a + 1 }$, which is bounded away from one as long as $a > 0$.
We have thus verified the overlap part of Condition~\ref{cond:reg-cond-arb-score}.

To verify the existence of treated units and abundance of control units, note that
\begin{equation}
	\Pr{\Wpop = 1}
	= \Es[\big]{\E{\Wpop \given \Svpop}}
	= \E[\bigg]{\frac{\Svpop}{2a + 2}}
	= \frac{1}{2a + 2},
\end{equation}
where the final equality follows from that $\E{\Svpop} = \E{\Xvpop_1} + \E{\Xvpop_3} = 2 \times 0.5$.
Hence, $\Pr{\Wpop = 1} > 0$ as long as $a < \infty$, and $\Pr{\Wpop = 1} < 1/2$ as long as $a > 0$.
The fifth part of Condition~\ref{cond:reg-cond-arb-score}, which is that $\E{\Ypop^2}$ exists, is satisfied by the fact that $\abs{\Svpop}$ is bounded by two.
Condition~\ref{cond:arb-lipschitz} is satisfied by independence of $\varepsilon_0$, which implies that $\Es{\POcpop \given \Svpop = s} = s^2$.
Condition~\ref{cond:assign-lipschitz} is satisfied by the fact that $\Pr{\Wpop = 1 \given \Svpop = s}$ here is a linear function.

We can conclude that this data generating process satisfies Conditions~\ref{cond:reg-cond-arb-score},~\ref{cond:arb-lipschitz} and \ref{cond:assign-lipschitz}.
This means that Proposition~\ref{prop:extended-main-res} applies, and we have
\begin{equation}
	\lim_{n \to \infty} \E{\attest} - \att = \frac{\Pr{\Svpop \in \optS}}{2\avgpscore} \bracket[\Big]{\Es[\big]{\POcpop \given \Wpop = 1, \Svpop \in \optS} - \Es[\big]{\POcpop \given \Wpop = 0, \Svpop \in \optS}}.
\end{equation}

We proceed by deriving this bias term.
The first step is to find the sets $\allSp$ and $\optS$.
Recall that
\begin{equation}
	\allSp = \braces[\Big]{ s \in \suppS : \Pr{\Wpop = 1 \given \Svpop = s} \geq 1/2},
\end{equation}
which for the current data generating process is $\allSp = [a + 1, 2]$.
If $a \geq 1$, then $\allSp$ has no mass.

Following the argument in the proof of Proposition~\ref{prop:extended-main-res}, we have $\optS = \bracket{b, 2}$, where $b$ is such that $\Pr{\Wpop = 1 \given \Svpop \in \optS} = 1/2$.
Note that
\begin{equation}
	\Pr{\Wpop = 1 \given \Svpop \in \bracket{b, 2}}
	= \Es[\big]{\E{\Wpop \given \Svpop} \given \Svpop \in \bracket{b, 2}}
	= \frac{1}{\Pr{\Svpop \in \bracket{b, 2}}} \int_b^2 \frac{s \func{f}{s}}{2a + 2} \diff s,
\end{equation}
where $\func{f}{s}$ is the density of $\Svpop$.
Calculate the integrals:
\begin{equation}
	\int_b^2 s \func{f}{s} \diff s
	= \begin{cases}
		\paren{b - 2}^2 \paren{b + 1} / 3 & \text{if } b \geq 1,
		\\
		\paren{3 - b^3} / 3 & \text{if } b < 1,
	\end{cases}
\end{equation}
and
\begin{equation}
	\Pr{\Svpop \in \bracket{b, 2}}
	= \int_b^2 \func{f}{s} \diff s
	= \begin{cases}
		\paren{b - 2}^2 / 2 & \text{if } b \geq 1,
		\\
		\paren{2 - b^2} / 2 & \text{if } b < 1.
	\end{cases}
\end{equation}
Hence, when $b \geq 1$,
\begin{equation}
	\Pr{\Wpop = 1 \given \Svpop \in \bracket{b, 2}}
	= \frac{b + 1}{3a + 3},
\end{equation}
and we attain $\Pr{\Wpop = 1 \given \Svpop \in \bracket{b, 2}} = 1/2$ with $b = \paren{3a + 1} / 2$, provided that $a \geq 1/3$.
Therefore, $\optS = \bracket{ \paren{3a + 1} / 2, 2 }$.

Finally, we need to calculate the difference between the conditional expectations of the potential outcomes in $\optS$:
\begin{equation}
\Es[\big]{\POcpop \given \Wpop = 1, \Svpop \in \optS} - \Es[\big]{\POcpop \given \Wpop = 0, \Svpop \in \optS}.
\end{equation}
Start by noting that
\begin{equation}
	2 \Es[\big]{\POcpop \given \Svpop \in \optS}
	= \Es[\big]{\POcpop \given \Wpop = 1, \Svpop \in \optS}
	+ \Es[\big]{\POcpop \given \Wpop = 0, \Svpop \in \optS}
\end{equation}
because $\Pr{\Wpop = 1 \given \Svpop \in \optS} = \Pr{\Wpop = 0 \given \Svpop \in \optS} = 1/2$.
This means that we can write
\begin{multline}
\Es[\big]{\POcpop \given \Wpop = 1, \Svpop \in \optS} - \Es[\big]{\POcpop \given \Wpop = 0, \Svpop \in \optS}
\\
= 2\Es[\big]{\POcpop \given \Wpop = 1, \Svpop \in \optS} - 2 \Es[\big]{\POcpop \given \Svpop \in \optS}.
\end{multline}
Starting with the expectation for the treated units, note
\begin{multline}
	\Es[\big]{\POcpop \given \Wpop = 1, \Svpop \in \optS}
	= \Et[\Big]{\Es[\big]{\POcpop \given \Svpop} \given \Wpop = 1, \Svpop \in \optS}
	\\
	= \E[\big]{\Svpop^2 \given \Wpop = 1, \Svpop \in \optS}
	= \int_b^2 s^2 \func{f}{s \mid \Wpop = 1, \Svpop \in \optS} \diff s,
\end{multline}
where $\func{f}{s \mid \Wpop = 1, \Svpop \in \optS}$ is the conditional density of $\Svpop$ given $\Wpop = 1$ and $\Svpop \in \optS$.
Using Bayes' rule, we can write
\begin{equation}
	\func{f}{s \mid \Wpop = 1, \Svpop \in \optS}
	= \frac{\func{f}{s} \Pr{\Wpop = 1, \Svpop \in \optS \given \Svpop = s}}{\Pr{\Wpop = 1, \Svpop \in \optS}}.
\end{equation}
As noted above, for $s \geq 1$, we have $\func{f}{s} = 2 - s$.
Furthermore,
\begin{equation}
\Pr{\Wpop = 1, \Svpop \in \optS \given \Svpop = s} = \Pr{\Wpop = 1 \given \Svpop = s}
\end{equation}
whenever $s \geq \paren{3a + 1} / 2$.
Recall that $\Pr{\Wpop = 1 \given \Svpop = s} = s / \paren{2a + 2}$.
When $a \geq 1/3$, so that $\optS = \bracket{ \paren{3a + 1} / 2, 2 }$, then
\begin{equation}
	\Pr{\Wpop = 1, \Svpop \in \optS}
	= \frac{9\paren{a - 1}^2}{16}.
\end{equation}
Hence, the conditional density is
\begin{equation}
	\func{f}{s \mid \Wpop = 1, \Svpop \in \optS}
	= \frac{8 s \paren{2 - s}}{9\paren{a + 1}\paren{a - 1}^2},
\end{equation}
and the condition expectation is
\begin{equation}
	\Es[\big]{\POcpop \given \Wpop = 1, \Svpop \in \optS}
	= \int_{\paren{3a + 1} / 2}^2 \frac{8 s^3 \paren{2 - s}}{9\paren{a + 1}\paren{a - 1}^2} \diff s
	= \frac{27 a^3 + 54 a^2 + 51 a + 28}{20 a + 20}.
\end{equation}

Next consider the conditional expectation for all units in $\optS$.
By a similar argument as above, when $a \geq 1/3$, so that $\optS = \bracket{ \paren{3a + 1} / 2, 2 }$, then
\begin{equation}
	\Es[\big]{\POcpop \given \Svpop \in \optS}
	= \frac{1}{\Pr{\Svpop \in \optS}} \int_{\paren{3a + 1} / 2}^2 s^2 \paren{2 - s} \diff s
	= \frac{9 a^2 + 14 a + 9}{8},
\end{equation}
which implies that
\begin{equation}
\Es[\big]{\POcpop \given \Wpop = 1, \Svpop \in \optS} - \Es[\big]{\POcpop \given \Wpop = 0, \Svpop \in \optS}
=
\frac{\paren{a - 1}^2 \paren{9 a + 11}}{20 a + 20}.
\end{equation}

Finally, when $1/3 \leq a \leq 1$, note that
\begin{equation}
	\frac{\Pr{\Svpop \in \optS}}{2\avgpscore} = \frac{9 \paren{a - 1}^2 \paren{a + 1} }{8}.
\end{equation}
Taken together, this gives
\begin{equation}
	\lim_{n \to \infty} \E{\attest} - \att
	= \braces[\bigg]{\frac{9 \paren{a - 1}^2 \paren{a + 1} }{8}} \times \braces[\bigg]{ \frac{\paren{a - 1}^2 \paren{9 a + 11}}{20 a + 20} }
	= \frac{9 \paren{a - 1}^4 \paren{9 a + 11}}{160}.
\end{equation}

Hence, we have $\lim_{n \to \infty} \E{\attest} - \att = 0.1556$ when $a = 1/3$, which is a substantial bias relative to the treatment effect $\att = 1$.
The asymptotic bias becomes smaller as $a$ grows.
For example, we have $\lim_{n \to \infty} \E{\attest} - \att = 0.0804$ when $a = 4/9$.
When $a = 1$, we have $\Pr{\Svpop \in \optS} = 0$, so $\lim_{n \to \infty} \E{\attest} - \att = 0$.

To confirm these theoretical results, I ran a small Monte Carlo study with the current data generating process and the three values of $a$ mentioned in the previous paragraph.
I used sample sizes ranging from one hundred to one million units, and ten thousand simulation rounds were drawn for each setting.
To facilitate the large sample sizes, I used an approximate version of optimal matching.
The error introduced by this approximation algorithm, compared to true optimal matching, is small relative to the bias introduced by the fact that the matching was done without replacement.
The error stemming from the approximation algorithm also converges to zero as the sample grows, unlike the bias from without replacement matching.

Table~\ref{tab:mc-results} reports the results.
The first two columns specify the simulation setting.
The third column ``Asymp.\ bias'' gives the asymptotic bias as predicted by the theory.
The fourth column ``Emp.\ bias'' presents the empirical bias in the Monte Carlo simulation.
The empirical bias is considerably larger than the asymptotic bias for smaller sample sizes, but the empirical bias approaches the asymptotic bias as the sample grows.
For samples with one million units, the empirical and asymptotic biases are the same up to the fourth decimal.
The fifth column ``Emp.\ SE'' presents the empirical standard error of the estimator in the Monte Carlo simulation.
We see that the standard error approaches zero as the sample grows no matter the value of $a$.
Hence, the estimator is convergent in this setting, but it converges to the wrong limit.

\begin{table}
\centering
\caption{Monte Carlo results for prognostic score example}\label{tab:mc-results}
\begin{tabular}{rrrrr}
$a$   & $n$ & Asymp.\ bias & Emp.\ bias & Emp.\ SE \\ \toprule
$1/3$ & $    100$ & $0.1556$ & $0.2666$ & $0.2708$ \\
$1/3$ & $   1000$ & $0.1556$ & $0.1660$ & $0.0866$ \\
$1/3$ & $  10000$ & $0.1556$ & $0.1567$ & $0.0274$ \\
$1/3$ & $ 100000$ & $0.1556$ & $0.1558$ & $0.0088$ \\
$1/3$ & $1000000$ & $0.1556$ & $0.1556$ & $0.0027$ \\[0.3em]

$4/9$ & $    100$ & $0.0804$ & $0.2132$ & $0.2752$ \\
$4/9$ & $   1000$ & $0.0804$ & $0.0981$ & $0.0868$ \\
$4/9$ & $  10000$ & $0.0804$ & $0.0819$ & $0.0277$ \\
$4/9$ & $ 100000$ & $0.0804$ & $0.0804$ & $0.0088$ \\
$4/9$ & $1000000$ & $0.0804$ & $0.0804$ & $0.0028$ \\[0.3em]

$1  $ & $    100$ & $0.0000$ & $0.1094$ & $0.3315$ \\
$1  $ & $   1000$ & $0.0000$ & $0.0184$ & $0.1036$ \\
$1  $ & $  10000$ & $0.0000$ & $0.0026$ & $0.0330$ \\
$1  $ & $ 100000$ & $0.0000$ & $0.0003$ & $0.0104$ \\
$1  $ & $1000000$ & $0.0000$ & $0.0000$ & $0.0033$ \\
\bottomrule
\end{tabular}
\end{table}

\section[Additional remarks]{Additional remarks}

\subsection*{Estimation of average treatment effects}

The main paper considered estimation of the average treatment effect of the treated units (ATT).
Practitioners are in some cases interested in the overall average treatment effect (ATE).
However, matching without replacement is generally not used to estimate this effect.
The reason is that to estimate the overall average treatment effect, one would need to match all treated units with control units, and then match all control units with treated units.
But if these matchings are to be constructed without replacement, there must be an equal number of treated units and control units in the sample.
This is rarely the case, and it will then not be possible to construct a matching without replacement.
Even if it is mechanically possible to construct such a matching, the resulting estimator would coincide with the unadjusted difference in average outcomes between the two treatment groups.
Hence, the approach implicitly presumes that the assignment mechanism is unconfounded without adjustment.
The conclusion is that matching without replacement cannot be used to estimate the overall average treatment effect (ATE).

\subsection*{Relation to criticisms of the use of propensity score}

Propensity score matching has recently been criticized by \citet{King2019Why}.
Using an analogy with experimental design, the authors argue that propensity score matching fail to emulate blocked randomized experiments, unlike some other matching methods.
This argument is unrelated to the discussion in the main paper.
Indeed, because the propensity score is the coarsest balancing score, matching on any other balancing score (including the raw covariates) would only aggravate the concerns highlighted here, so the recommendations by \citet{King2019Why} would not be a solution.
A possible solution would be to match on a score that is coarser than the propensity score, but such a score cannot be a balancing score, so this approach would run the risk of not fully adjust for the confounding, and may fail to achieve consistency for that reason.

\subsection*{Variance estimation}

The main paper solely considers point estimation.
However, the investigation suggests a possible approach to variance estimation.
Recall that $\PSpoint$ partitions the sample into two groups based on the propensity scores.
Because the matching behaves vastly differently in these two groups, we can use the partition to estimate the variance.
In particular, asymptotically, all control units with $\rPSpop \geq \PSpoint$ will be matched, so we can simply estimate the variance for these units, disregarding the matching.
Furthermore, for units with $\rPSpop < \PSpoint$, the matching will asymptotically behave as if it was with replacement, so variance estimation techniques used for matching with replacement can be used here \citep[see, e.g.,][]{Abadie2006Large}.
Combining the variance estimators in the two parts of the partition would produce a variance estimator for the overall sample.
This may therefore be an alternative to bootstrap estimators, which have been shown to not perform well for matching \citep{Abadie2008Failure}.

However, it is questionable how useful such a variance estimator would be given the bias exhibited by the point estimator.
This bias may prompt practitioners not to use matching without replacement in the first place, in which case they are in no need of a variance estimator.
Even if they decided to use matching without replacement, the variance estimator would not capture the bias, so confidence intervals and hypothesis tests based on the variance estimator would be dangerously misleading, and should therefore not be used in practice.

\subsection*{Implications for weighting estimators}

The matching estimator discussed in the main paper can be reinterpreted as a weighting estimator.
In particular, for a matching $\mgensym : \Treated \to \Controls$, we can define weights for the control units as
\begin{equation}
	\nu_i = \sum_{j \in \Treated} \indicator{\mgenj = i},
\end{equation}
which allows us to write the matching estimator as
\begin{equation}
	\attest = \frac{1}{\nTreated} \sum_{i \in \Treated} \Yi - \frac{1}{\nTreated} \sum_{i \in \Controls} \nu_i \Yi.
\end{equation}
Matching without replacement imposes restrictions on these weights.
The first restriction is that $\nu_i \in \braces{0, 1}$, which ensures that each control is matched with at most one treated unit.
The second restriction is that $\sum_i^n \nu_i = \nTreated$, which ensures that all treated units are matched.

Using this representation of the estimator, it is possible to generalize the insights of this paper to other weighting estimators, such as kernel estimators.
The feature of matching without replacement that introduces bias is that the weights $\nu_i$ are restricted to $\braces{0, 1}$.
This means that we are prevented to upweight units in sparse regions of the covariate space to the degree we would like.
Compare this with matching with replacement, for which the weights are restricted to $\braces{0, 1, \dotsc, \nTreated}$.
While this also restricts us to integer weights, we are free to put as much weight as we would like on any one unit.

This suggests that similar issues would arise for other weighting estimators if they impose strict upper limits on the weights.
However, commonly used weighting estimators, such as the kernel estimator by \citet*{Heckman1997Matching,Heckman1998Matching}, do not impose such restrictions, so they should not exhibit the type of bias demonstrated in the main paper.

This representation suggests a way to mitigate the bias without necessary fully adopt matching with replacement.
If we relax the restriction on the weights so they can take values in $\braces{0, 1, \dotsc, k}$ for some $k < \nTreated$, we effectively allow more than one treated unit to be matched to the same control unit, although there cannot be more than $k$ units matched to the same unit.
This will facilitate consistency under weaker conditions than those in the main paper.
For example, if we impose $\nu_i \in \braces{0, 1, 2}$, it will be enough to assume that $\pscore{x} \leq 2/3$ on the support of $\Xvpop$, rather than $\pscore{x} \leq 1/2$.

\section{Miscellaneous definitions and propositions}

\begin{definition}\label{def:rPS-support}
	Let $\suppPS$ be the support of $\rPSpop$.
\end{definition}

\begin{definition}\label{def:rPS-lower-upper-support}
	Let $\suppPSm = \inf \suppPS$ and $\suppPSp = \sup \suppPS$.
\end{definition}

\begin{definition}\label{def:cond-exp-PO-given-PS}
	Let $\EPOPS{z}{p} = \E{\POpop{z} \given \rPSpop = p}$ be the conditional expectation of the potential outcome for treatment $z$ given propensity score $p$.
\end{definition}

\newcommand{\hoeffdingrv}{X}

\begin{theorem}\label{thm:hoeffdings-inequality}
	Let $\hoeffdingrv_1, \hoeffdingrv_2, \dotsc, \hoeffdingrv_n$ be $n$ independent random variables such that $0 \leq \hoeffdingrv_i \leq 1$ with probability one.
	Let $\bar{\hoeffdingrv} = \cavgin \hoeffdingrv_i$ and $\mu = \E{\bar{\hoeffdingrv}}$.
	For $0 < t < 1 - \mu$ and $0 < s < \mu$:
	\begin{equation}
		\Pr{\bar{\hoeffdingrv} - \mu \geq t} \leq \expf{-2nt^2} \qquad\text{and}\qquad \Pr{\bar{\hoeffdingrv} - \mu \leq -s} \leq \expf{-2ns^2}.
	\end{equation}
\end{theorem}

\begin{proof}
The first inequality is Theorem~1 in \citet{Hoeffding1963Probability}.
For the second inequality:
\begin{equation}
	\Pr{\bar{\hoeffdingrv} - \mu \leq -s}
	=
	\Pr{-\bar{\hoeffdingrv} + \mu \geq s}
	=
	\Pr[\big]{\paren{1 - \bar{\hoeffdingrv}} - \paren{1 - \mu} \geq s},
\end{equation}
which is bounded by $\expf{-2ns^2}$ when $0 < s < 1 - \paren{1 - \mu}$ using the first inequality.
\end{proof}

\begin{theorem}\label{thm:pscore-balancing}
	Under Condition~\refmain{cond:reg-cond}, $\POcpop$ is conditionally independent of $\Wpop$ given $\rPSpop = p$ on the support of $\rPSpop$.
\end{theorem}

\begin{proof}
The statement is Theorem~3 in \citet{Rosenbaum1983Central}.
\end{proof}

\section{Proofs of lemmas}

\subsection{Proof of Lemma~\refmain{lem:consistent-unbiased}}

\begin{reflemma}{\refmain{lem:consistent-unbiased}}
	\lemmaconsistentunbiased{}
\end{reflemma}

\newcommand{\tmppar}{\theta}
\newcommand{\tmpest}{\hat{\theta}}

\begin{proof}
Assume $\tmpest$ is consistent for $\tmppar$ but that the implication does not hold.
In other words, for some constant $c \geq 1$:
\begin{equation}
	\limsup_{n \to \infty} \abs[\big]{\E{\tmpest} - \tmppar} \geq 1 / c \qquad \text{and} \qquad \limsup_{n \to \infty}\Var{\tmpest} \leq c.
\end{equation}

Let $\paren{\varepsilon_n}$ be a sequence in $\Reals^+$ such that $\varepsilon_n \to 0$ and $\Pr{\abs{\tmpest - \tmppar} \geq \varepsilon_n} \to 0$.
Consistency ensures that such a sequence exists.
Let $A_n = \indicator[\big]{\abs{\tmpest - \tmppar} \geq \varepsilon_n}$ be a sequence of random variables and let $\delta_n = \E{A_n} = \Pr{A_n = 1} = \Pr{\abs{\tmpest - \tmppar} \geq \varepsilon_n}$.
By the law of total variance:
\begin{equation}
	\Var{\tmpest} = \E[\big]{\Var{\tmpest \given A_n}} + \Var[\big]{\E{\tmpest \given A_n}} \geq \Var[\big]{\E{\tmpest \given A_n}}.
\end{equation}
Note that $\E[\big]{\E{\tmpest \given A_n}} = \E{\tmpest}$, so:
\begin{equation}
	\Var[\big]{\E{\tmpest \given A_n}}
	= \E[\Big]{\paren[\big]{\E{\tmpest \given A_n} - \E{\tmpest}}^2}
	= \frac{1 - \delta_n}{\delta_n} \paren[\Big]{\E{\tmpest} - \E{\tmpest \given A_n = 0}}^2.
\end{equation}
The last equality may need some elaboration.
Note:
\begin{multline}
	\frac{1 - \delta_n}{\delta_n} \paren[\Big]{\E{\tmpest} - \E{\tmpest \given A_n = 0}}^2
	=
	\frac{\paren{1 - \delta_n} \paren{1 - \delta_n + \delta_n} }{\delta_n} \paren[\Big]{\E{\tmpest} - \E{\tmpest \given A_n = 0}}^2
	\\
	=
	\frac{\paren{1 - \delta_n}^2}{\delta_n} \paren[\Big]{\E{\tmpest} - \E{\tmpest \given A_n = 0}}^2
	+
	\paren{1 - \delta_n} \paren[\Big]{\E{\tmpest} - \E{\tmpest \given A_n = 0}}^2.
\end{multline}
By the law of total expectation:
\begin{equation}
	\E{\tmpest} = \delta_n \E{\tmpest \given A_n = 1} + \paren{1 - \delta_n} \E{\tmpest \given A_n = 0},
\end{equation}
so:
\begin{multline}
	\frac{\paren{1 - \delta_n}^2}{\delta_n} \paren[\Big]{\E{\tmpest} - \E{\tmpest \given A_n = 0}}^2
	= \frac{\paren{1 - \delta_n}^2}{\delta_n} \paren[\bigg]{\E{\tmpest} - \frac{\E{\tmpest} - \delta_n \E{\tmpest \given A_n = 1}}{1 - \delta_n}}^2
	\\
	= \frac{\paren{1 - \delta_n}^2}{\delta_n} \paren[\bigg]{\frac{\delta_n \E{\tmpest \given A_n = 1} - \delta_n\E{\tmpest}}{1 - \delta_n}}^2 = \delta_n \paren[\Big]{\E{\tmpest \given A_n = 1} - \E{\tmpest}}^2,
\end{multline}
and the equality follows from the law of total expectation and $\delta_n = \Pr{A_n = 1}$.

Focusing on the second factor, add and subtract $\tmppar$ to get:
\begin{equation}
	\paren[\big]{\E{\tmpest} - \E{\tmpest \given A_n = 0}}^2 \geq \paren[\big]{\E{\tmpest} - \tmppar}^2 + 2\paren[\big]{\E{\tmpest} - \tmppar}\paren[\big]{\tmppar - \E{\tmpest \given A_n = 0}}.
\end{equation}
Let $b_n = \abs{\E{\tmpest} - \tmppar}$ be the magnitude of the bias.
Recall that $\abs{\tmpest - \tmppar} < \varepsilon_n$ when $A_n = 0$, so $\abs[\big]{\tmppar - \E{\tmpest \given A_n = 0}} < \varepsilon_n$.
It follows that:
\begin{equation}
	2\paren[\big]{\E{\tmpest} - \tmppar}\paren[\big]{\tmppar - \E{\tmpest \given A_n = 0}} \geq - 2 b_n \varepsilon_n,
\end{equation}
and:
\begin{equation}
	\paren[\big]{\E{\tmpest} - \tmppar}^2 + 2\paren[\big]{\E{\tmpest} - \tmppar}\paren[\big]{\tmppar - \E{\tmpest \given A_n = 0}} \geq b_n\paren{b_n - 2 \varepsilon_n}.
\end{equation}
Taken together:
\begin{equation}
	\Var{\tmpest} \geq \frac{b_n \paren{b_n - 2 \varepsilon_n} \paren{1 - \delta_n}}{\delta_n}.
\end{equation}

Recall that the proof started by assuming:
\begin{equation}
	\limsup_{n \to \infty} \abs[\big]{\E{\tmpest} - \tmppar} \geq 1 / c \qquad \text{and} \qquad \limsup_{n \to \infty}\Var{\tmpest} \leq c,
\end{equation}
for some constant $c \geq 1$.
Let $n'$ be such that $\varepsilon_n \leq 1 / 8 c$ and $\delta_n \leq 1 / 32c^3 \leq 1 / 2$ for all $n \geq n'$.
Consistency ensures that such an integer exists.
Asymptotic biasedness implies that $b_n \geq 1 / 2c$ an infinite number of times for $n \geq n'$, and in these cases:
\begin{equation}
	\Var{\tmpest}
	\geq \frac{b_n \paren{b_n - 2 \varepsilon_n}\paren{1 - \delta_n}}{\delta_n}
	\geq \frac{\paren{1 / 2c} \paren{1 / 2c - 2 / 8 c}\paren{1 - 1 / 2}}{1 / 32c^3}
	= 2c,
\end{equation}
which contradicts $\limsup_{n \to \infty}\Var{\tmpest} \leq c$.
\end{proof}

\subsection{Proof of Lemma~\refmain{lem:bounded-var}}

\begin{reflemma}{\refmain{lem:bounded-var}}
	\lemmaboundedvar{}
\end{reflemma}

\begin{proof}
Consider the expectation of the squared estimator:
\begin{equation}
	\Var{\attest} = \E{\attest^2} - \paren[\big]{\E{\attest}}^2 \leq \E{\attest^2}.
\end{equation}
As noted in the proof of Lemma~\ref{lem:mdiscrep-bias-step3}, the estimator can be written using $\Matched$ when $1 \leq \nTreated \leq \nControls$.
Thus, in that case:
\begin{equation}
	\attest^2
	= \paren[\Bigg]{\frac{1}{\nTreated} \sum_{i \in \Treated} \Yi - \frac{1}{\nTreated} \sum_{i \in \Matched} \Yi}^2
	\leq \frac{1}{\nTreated^2} \paren[\Bigg]{\sum_{i \in \mathbf{A}} \abs{\Yi}}^2,
\end{equation}
where $\mathbf{A} = \Treated \cup \Matched$.
Recall that $\attest = 0$ when $\nTreated = 0$ or $\nTreated > \nControls$, so:
\begin{equation}
	\attest^2 \leq \frac{1}{\maxf{1, \nTreated^2}} \paren[\Bigg]{\sum_{i \in \mathbf{A}} \abs{\Yi}}^2,
\end{equation}
holds no matter how many treated units there are in the sample.
Noting that $2ab \leq a^2 + b^2$ for any $a, b \in \Reals$:
\begin{equation}
	\paren[\Bigg]{\sum_{i \in \mathbf{A}} \abs{\Yi}}^2 = \frac{1}{2}\sum_{i \in \mathbf{A}}\sum_{j \in \mathbf{A}} 2 \abs{\Yi\Yj} \leq \frac{1}{2}\sum_{i \in \mathbf{A}}\sum_{j \in \mathbf{A}}\paren{\Yi^2 + \Yj^2} \leq 2 \nTreated \sum_{i \in \mathbf{A}}\Yi^2,
\end{equation}
because $\card{\mathbf{A}} = 2\nTreated$ when $\nTreated \leq \nControls$ and $\card{\mathbf{A}} < 2\nTreated$ when $\nTreated > \nControls$.
Separating the sum again gives:
\begin{equation}
	\E{\attest^2} \leq \E[\Bigg]{\frac{2}{\maxf{1, \nTreated}} \sum_{i \in \Treated}\Yi^2} + \E[\Bigg]{\frac{2}{\maxf{1, \nTreated}} \sum_{i \in \Matched}\Yi^2}.
\end{equation}

As noted in the previous proofs, $\Treated$ contains no more information about $\Yi$ than $\Wi$, so:
\begin{equation}
	\E[\Bigg]{\frac{2}{\maxf{1, \nTreated}} \sum_{i \in \Treated}\Yi^2}
	= 2 \E[\Bigg]{\frac{\nTreated}{\maxf{1, \nTreated}}} \E{\Ypop^2 \given \Wpop = 1},
\end{equation}
and:
\begin{equation}
	\E[\Bigg]{\frac{2}{\maxf{1, \nTreated}} \sum_{i \in \Matched}\Yi^2}
	\leq \E[\Bigg]{\frac{2}{\maxf{1, \nTreated}} \sum_{i \in \Controls}\Yi^2}
	= 2 \E[\Bigg]{\frac{\nControls}{\maxf{1, \nTreated}}} \E{\Ypop^2 \given \Wpop = 0}.
\end{equation}
Using the same argument as in the proof of Lemma~\ref{lem:mdiscrep-bias-step1}:
\begin{equation}
	\E{\Ypop^2 \given \Wpop = 1} \leq \frac{\E{\Ypop^2}}{\avgpscore}
	\qquad\text{and}\qquad
	\E{\Ypop^2 \given \Wpop = 0} \leq \frac{\E{\Ypop^2}}{\avgpscore},
\end{equation}
so:
\begin{equation}
	\E{\attest^2}
	\leq \E[\Bigg]{\frac{\nTreated + \nControls}{\maxf{1, \nTreated}}} \frac{2\E{\Ypop^2}}{\avgpscore}
	= \E[\Bigg]{\frac{n}{\maxf{1, \nTreated}}} \frac{2\E{\Ypop^2}}{\avgpscore}.
\end{equation}
Finally:
\begin{equation}
	\E[\Bigg]{\frac{n}{\maxf{1, \nTreated}}}
	\leq
	n \Pr[\big]{\nTreated = 0}
	+ \Pr[\big]{\nTreated > \nControls}
	+ \E[\Bigg]{\frac{n}{\nTreated} \given 1 \leq \nTreated \leq \nControls}.
\end{equation}
The first term is $n \paren{1 - \avgpscore}^n$ and converges to zero.
The second term was shown to converge to zero in the proof of Lemma~\ref{lem:mdiscrep-bias-step2}.
That proof also showed:
\begin{equation}
	\limsup_{n \to \infty} \E[\bigg]{\frac{n}{\nTreated} \given 1 \leq \nTreated \leq \nControls} \leq \frac{2}{\avgpscore} \tag*{\qedhere}.
\end{equation}
\end{proof}

\subsection{Proof of Lemma~\refmain{lem:no-crossing-matches}}

\begin{reflemma}{\refmain{lem:no-crossing-matches}}
	\lemmanocrossingmatches{}
\end{reflemma}

\newcommand{\maltsym}{\matchsym'}
\DeclarePairedDelimiterXPP\malt[1]{\maltsym}{\lparen}{\rparen}{}{#1}

\begin{proof}
The lemma is proven by demonstrating the contrapositive.
Consider a matching $\mgensym$ containing at least one pair of matches that are crossing according to Definition~\refmain{def:crossing-matches}.
That is, for some $k, \ell \in \Treated$:
\begin{equation}
	\maxf{\rPS{k}, \rPS{\mgen{\ell}}} < \minf{\rPS{\ell}, \rPS{\mgen{k}}}.
\end{equation}
Fix these indices throughout the proof, so $k$ and $\ell$ refer to two specific treated units.

Consider an alternative matching $\maltsym$ that swaps the matched controls for $k$ and $\ell$:
\begin{equation}
	\malt{i}
	= \begin{cases}
		\mgen{k} & \text{if } i = \ell,
		\\
		\mgen{\ell} & \text{if } i = k,
		\\
		\mgen{i} & \text{otherwise}.
	\end{cases}
\end{equation}
The sum of within-match differences in the propensity scores for the two matchings are:
\begin{equation}
	\sum_{i \in \Treated} \abs{\rPSi - \rPS{\mgeni}}
	\qquad\text{and}\qquad
	\sum_{i \in \Treated} \abs{\rPSi - \rPS{\malt{i}}},
\end{equation}
and their difference is:
\begin{equation}
	\sum_{i \in \Treated} \paren[\big]{\abs{\rPSi - \rPS{\mgeni}} - \abs{\rPSi - \rPS{\malt{i}}}}
	= \abs{\rPS{k} - \rPS{\mgen{k}}} - \abs{\rPS{k} - \rPS{\mgen{\ell}}} + \abs{\rPS{\ell} - \rPS{\mgen{\ell}}} - \abs{\rPS{\ell} - \rPS{\mgen{k}}},
\end{equation}
because they are identical apart from the matches for $k$ and $\ell$.
It remains to show that this difference is positive.

Because the matches are crossing, $\rPS{\mgen{k}} > \rPS{k}$ and $\rPS{\ell} > \rPS{\mgen{\ell}}$, so:
\begin{equation}
	\abs{\rPS{k} - \rPS{\mgen{k}}} = \rPS{\mgen{k}} - \rPS{k} \qquad\qquad \text{and} \qquad\qquad \abs{\rPS{\ell} - \rPS{\mgen{\ell}}} = \rPS{\ell} - \rPS{\mgen{\ell}}.
\end{equation}
Define $B_k, B_\ell \in \setb{-1, 1}$ so that $B_k\rPS{k} \geq B_k\rPS{\mgen{\ell}}$ and $B_\ell\rPS{\mgen{k}} \geq B_\ell\rPS{\ell}$, and write:
\begin{equation}
	\abs{\rPS{k} - \rPS{\mgen{\ell}}} = B_k\rPS{k} - B_k\rPS{\mgen{\ell}} \qquad\qquad \text{and} \qquad\qquad \abs{\rPS{\ell} - \rPS{\mgen{k}}} = B_\ell\rPS{\mgen{k}} - B_\ell\rPS{\ell},
\end{equation}
so the difference can be written as:
\begin{multline}
	\abs{\rPS{k} - \rPS{\mgen{k}}} - \abs{\rPS{k} - \rPS{\mgen{\ell}}} + \abs{\rPS{\ell} - \rPS{\mgen{\ell}}} - \abs{\rPS{\ell} - \rPS{\mgen{k}}}
	\\
	=
	\paren[\big]{1 + B_\ell} \rPS{\ell} + \paren[\big]{1 - B_\ell} \rPS{\mgen{k}} - \paren[\big]{1 + B_k} \rPS{k} - \paren[\big]{1 - B_k} \rPS{\mgen{\ell}}.
\end{multline}

Note that $\setb{1 + B_\ell, 1 - B_\ell} = \setb{0, 2}$ because $B_\ell \in \setb{-1, 1}$.
It follows that:
\begin{equation}
	\paren[\big]{1 + B_\ell}\rPS{\ell} + \paren[\big]{1 - B_\ell}\rPS{\mgen{k}} \geq 2 \minf{\rPS{\ell}, \rPS{\mgen{k}}} > 2 \maxf{\rPS{k}, \rPS{\mgen{\ell}}},
\end{equation}
where the last inequality follows from $k$ and $\ell$ having crossing matches.
By a similar argument:
\begin{equation}
	- \paren[\big]{1 + B_k}\rPS{k} - \paren[\big]{1 - B_k}\rPS{\mgen{\ell}} \geq -2 \maxf{\rPS{k}, \rPS{\mgen{\ell}}},
\end{equation}
which implies:
\begin{multline}
	\paren[\big]{1 + B_\ell} \rPS{\ell} + \paren[\big]{1 - B_\ell} \rPS{\mgen{k}} - \paren[\big]{1 + B_k} \rPS{k} - \paren[\big]{1 - B_k} \rPS{\mgen{\ell}}
	\\
	> 2 \maxf{\rPS{k}, \rPS{\mgen{\ell}}} - 2 \maxf{\rPS{k}, \rPS{\mgen{\ell}}} = 0. \tag*{\qedhere}
\end{multline}
\end{proof}

\section{Proof of Proposition~\refmain{prop:pscore-bias}}

\newcommand{\SamplePSm}{\Sample_-}
\newcommand{\SamplePSp}{\Sample_+}
\newcommand{\nSamplePSm}{\card{\SamplePSm}}
\newcommand{\nSamplePSp}{\card{\SamplePSp}}

\newcommand{\TreatedPSm}{\Treated_-}
\newcommand{\TreatedPSp}{\Treated_+}
\newcommand{\nTreatedPSm}{\card{\TreatedPSm}}
\newcommand{\nTreatedPSp}{\card{\TreatedPSp}}

\newcommand{\ControlsPSm}{\Controls_-}
\newcommand{\ControlsPSp}{\Controls_+}
\newcommand{\nControlsPSm}{\card{\ControlsPSm}}
\newcommand{\nControlsPSp}{\card{\ControlsPSp}}

\newcommand{\MatchedPSm}{\Matched_-}
\newcommand{\MatchedPSp}{\Matched_+}
\newcommand{\nMatchedPSm}{\card{\MatchedPSm}}
\newcommand{\nMatchedPSp}{\card{\MatchedPSp}}

\subsection{Overview of proof}

The proof of Proposition~\refmain{prop:pscore-bias} consists of four parts.
The first part is to show that the bias of the estimator is asymptotically equal the expectation of a random variable $\mdiscrep$, which is the normalized difference between the sum of $\POci$ among treated units and the sum of $\POci$ among matched controls (see Definition~\ref{def:mdiscrep} in Section~\ref{sec:app-definitions}).
This random variable has a non-random denominator, making it easier to analyze than the matching estimator itself.
The asymptotic equivalence is proven in Lemmas~\ref{lem:mdiscrep-bias},~\ref{lem:mdiscrep-bias-step1},~\ref{lem:bounded-POs},~\ref{lem:mdiscrep-bias-step2},~\ref{lem:mdiscrep-bias-step3}~and~\ref{lem:mdiscrep-bias-step4}.
While necessary to rigorously prove Proposition~\refmain{prop:pscore-bias}, the proofs of these lemmas are somewhat tedious and do not bring many interesting insights.

Next, $\mdiscrep$ is decomposed into three terms.
The remaining three parts of the proof consider these three terms in turn.
The proofs of these lemmas are also somewhat tedious, but they provide several insights, and readers may find interesting to study these proof somewhat more carefully.

The first term is the normalized difference between the sums of the potential outcomes of all control units with $\rPSi \geq \PSpoint$ and matched control with $\rPSi \geq \PSpoint$.
Lemmas~\ref{lem:pscore-mdiscrep-term1}, \ref{lem:pscore-mdiscrep-term1-no-upper-unmatched}, \ref{lem:pscore-mdiscrep-term1-no-upper-unmatched-middle-part}, \ref{lem:pscore-mdiscrep-term1-no-upper-unmatched-atom-point} and \ref{lem:pscore-mdiscrep-term1-upper-balance} show that this term converges to zero.
The intuition behind this result is that, asymptotically, all control units with $\rPSi \geq \PSpoint$ will be matched, so the two sums in the difference are over the same units.

The second term is the normalized difference between the sums of the potential outcomes of treated units with $\rPSi < \PSpoint$ and matched control units with $\rPSi < \PSpoint$.
Lemmas~\ref{lem:pscore-mdiscrep-term1-no-upper-unmatched}, \ref{lem:pscore-mdiscrep-term1-no-upper-unmatched-middle-part}, \ref{lem:pscore-mdiscrep-term1-no-upper-unmatched-atom-point}, \ref{lem:pscore-mdiscrep-term1-upper-balance}, \ref{lem:pscore-mdiscrep-term2} and \ref{lem:pscore-mdiscrep-term2-bin-overflow} show that this term converges to zero.
The intuition behind this result is that, asymptotically, each treated units with $\rPSi < \PSpoint$ will be matched to a control unit that has a propensity score that is infinitesimally close to the score of the treated unit.
Condition~\refmain{cond:pscore-lipschitz} thereby ensures that the average outcome of these matched control units is the same as the average potential outcome under the control condition for the treated units they are matched with.

The third term is the normalized difference between the sums of the potential outcomes of treated units with $\rPSi \geq \PSpoint$ and all control units with $\rPSi \geq \PSpoint$.
Lemma~\ref{lem:pscore-mdiscrep-term3} shows that this term converges to the quantity stipulated in the proposition.
This term does not depend on the matching, so the proof of this lemma is straightforward.

Figure~\ref{fig:dependency-diagram} is a diagram of the relationships between the proposition and its lemmas.

\begin{figure}
\centering
\begin{tikzpicture}
	\tikzstyle{lemmanode}=[draw,thick,circle,minimum size=1.2cm]
	\tikzstyle{arc}=[thick,-{Latex[scale=1.1]}]

	\node[draw,thick] (pscorebias) {\LARGE Proposition~\refmain{prop:pscore-bias}};
	\node (pscorebiaschild) [below=40pt of pscorebias] {};

	\node[lemmanode] (mdiscrepbias) [left=90pt of pscorebiaschild] {\ref{lem:mdiscrep-bias}};
	\node (mdiscrepbiaschild) [below=35pt of mdiscrepbias] {};
	\node[lemmanode] (mdiscrepbiasstep2) [left=5pt of mdiscrepbiaschild] {\ref{lem:mdiscrep-bias-step2}};
	\node[lemmanode] (mdiscrepbiasstep1) [left=10pt of mdiscrepbiasstep2] {\ref{lem:mdiscrep-bias-step1}};
	\node[lemmanode] (mdiscrepbiasstep3) [right=5pt of mdiscrepbiaschild] {\ref{lem:mdiscrep-bias-step3}};
	\node[lemmanode] (mdiscrepbiasstep4) [right=10pt of mdiscrepbiasstep3] {\ref{lem:mdiscrep-bias-step4}};
	\node[lemmanode] (boundedPOs) [below=20pt of mdiscrepbiasstep1] {\ref{lem:bounded-POs}};

	\draw[arc] (mdiscrepbias) -- (pscorebias);
	\draw[arc] (mdiscrepbiasstep1) -- (mdiscrepbias);
	\draw[arc] (mdiscrepbiasstep2) -- (mdiscrepbias);
	\draw[arc] (mdiscrepbiasstep3) -- (mdiscrepbias);
	\draw[arc] (mdiscrepbiasstep4) -- (mdiscrepbias);
	\draw[arc] (boundedPOs) -- (mdiscrepbiasstep1);

	\node[lemmanode] (pscoremdiscrepterm1) [right=15pt of pscorebiaschild] {\ref{lem:pscore-mdiscrep-term1}};
	\node (pscoremdiscrepterm1child) [below=35pt of pscoremdiscrepterm1] {};
	\node[lemmanode] (pscoremdiscrepterm1noupperunmatched) [left=5pt of pscoremdiscrepterm1child] {\ref{lem:pscore-mdiscrep-term1-no-upper-unmatched}};
	\node[lemmanode] (pscoremdiscrepterm1upperbalance) [right=10pt of pscoremdiscrepterm1child] {\ref{lem:pscore-mdiscrep-term1-upper-balance}};
	\node (pscoremdiscrepterm1noupperunmatchedchild) [below=33pt of pscoremdiscrepterm1noupperunmatched] {};
	\node[lemmanode] (pscoremdiscrepterm1noupperunmatchedmiddlepart) [left=5pt of pscoremdiscrepterm1noupperunmatchedchild] {\ref{lem:pscore-mdiscrep-term1-no-upper-unmatched-middle-part}};
	\node[lemmanode] (pscoremdiscrepterm1noupperunmatchedatompoint) [right=5pt of pscoremdiscrepterm1noupperunmatchedchild] {\ref{lem:pscore-mdiscrep-term1-no-upper-unmatched-atom-point}};

	\draw[arc] (pscoremdiscrepterm1) -- (pscorebias);
	\draw[arc] (pscoremdiscrepterm1noupperunmatched) -- (pscoremdiscrepterm1);
	\draw[arc] (pscoremdiscrepterm1upperbalance) -- (pscoremdiscrepterm1);
	\draw[arc] (pscoremdiscrepterm1noupperunmatchedmiddlepart) -- (pscoremdiscrepterm1noupperunmatched);
	\draw[arc] (pscoremdiscrepterm1noupperunmatchedatompoint) -- (pscoremdiscrepterm1noupperunmatched);

	\node[lemmanode] (pscoremdiscrepterm2) [right=80pt of pscorebiaschild] {\ref{lem:pscore-mdiscrep-term2}};
	\node[lemmanode] (pscoremdiscrepterm2binoverflow) [right=70pt of pscoremdiscrepterm1child] {\ref{lem:pscore-mdiscrep-term2-bin-overflow}};

	\draw[arc] (pscoremdiscrepterm2) -- (pscorebias);
	\draw[arc] (pscoremdiscrepterm2binoverflow) -- (pscoremdiscrepterm2);
	\draw[arc] (pscoremdiscrepterm1noupperunmatched) -- (pscoremdiscrepterm2);
	\draw[arc] (pscoremdiscrepterm1upperbalance) -- (pscoremdiscrepterm2);

	\node[lemmanode] (pscoremdiscrepterm3) [right=140pt of pscorebiaschild] {\ref{lem:pscore-mdiscrep-term3}};

	\draw[arc] (pscoremdiscrepterm3) -- (pscorebias);
\end{tikzpicture}
\caption{Dependency diagram over proof of Proposition~\refmain{prop:pscore-bias}}\label{fig:dependency-diagram}
\end{figure}

\subsection{Definitions}\label{sec:app-definitions}

\begin{definition}
	Let $\Matched$ collect all matched control units when the matching exists and all controls when it does not exist:
	\begin{equation}
		\Matched =
		\begin{cases}
			\setb{\mopt{i} \suchthat i \in \Treated} & \text{if } \nTreated \leq \nControls,
			\\
			\Controls & \text{if } \nTreated > \nControls,
		\end{cases}
	\end{equation}
\end{definition}

\begin{definition}\label{def:mdiscrep}
	Let $\mdiscrep$ be the sum of difference in the potential outcome under control between each treated unit and its matched control unit normalized by the expected number of treated units:
	\begin{equation}
		\mdiscrep = \frac{1}{\avgpscore n} \sum_{i \in \Treated} \POci - \frac{1}{\avgpscore n} \sum_{i \in \Matched} \POci.
	\end{equation}
\end{definition}

\begin{definition}\label{def:PSpoint-partition}
	Partition $\Sample$, $\Treated$ and $\Controls$ as:
	\begin{align}
		\SamplePSp &= \setb{i \in \Sample \suchthat \rPSi \geq \PSpoint}, \qquad &\TreatedPSp &= \setb{i \in \Treated \suchthat \rPSi \geq \PSpoint}, \qquad &\ControlsPSp &= \setb{i \in \Controls \suchthat \rPSi \geq \PSpoint},
		\\
		\SamplePSm &= \setb{i \in \Sample \suchthat \rPSi < \PSpoint}, \qquad &\TreatedPSm &= \setb{i \in \Treated \suchthat \rPSi < \PSpoint}, \qquad &\ControlsPSm &= \setb{i \in \Controls \suchthat \rPSi < \PSpoint},
	\end{align}
	and partition $\Matched$ as:
	\begin{equation}
		\MatchedPSp =
		\begin{cases}
			\setb{\mopt{i} \suchthat i \in \TreatedPSp} & \text{if } \nTreated \leq \nControls,
			\\
			\ControlsPSp & \text{if } \nTreated > \nControls,
		\end{cases}
		\qquad\quad
		\MatchedPSm =
		\begin{cases}
			\setb{\mopt{i} \suchthat i \in \TreatedPSm} & \text{if } \nTreated \leq \nControls,
			\\
			\ControlsPSm & \text{if } \nTreated > \nControls.
		\end{cases}
	\end{equation}
\end{definition}

\newcommand{\IMatchedPSp}[1]{M^+_{#1}}
\newcommand{\IMatchedPSpi}{\IMatchedPSp{i}}
\newcommand{\IMatchedPSpj}{\IMatchedPSp{j}}

\begin{definition}\label{def:PSpoint-match-indicator}
	Let $\IMatchedPSpi = \indicator{i \in \MatchedPSp}$.
\end{definition}

\subsection{Main proof}

\begin{refproposition}{\refmain{prop:pscore-bias}}
	\proppscorebias{}
\end{refproposition}

\begin{proof}
Recall the definition of $\mdiscrep$:
\begin{equation}
	\mdiscrep = \frac{1}{\avgpscore n} \sum_{i \in \Treated} \POci - \frac{1}{\avgpscore n} \sum_{i \in \Matched} \POci.
\end{equation}
Lemma~\ref{lem:mdiscrep-bias} shows that:
\begin{equation}
	\liminf_{n \to \infty} \E{\attest} = \att + \liminf_{n \to \infty} \E{\mdiscrep}
	\qquad\text{and}\qquad
	\limsup_{n \to \infty} \E{\attest} = \att + \limsup_{n \to \infty} \E{\mdiscrep},
\end{equation}
and the rest of the proof considers $\mdiscrep$.

Use the partitions of $\Treated$, $\Controls$ and $\Matched$ in Definition~\ref{def:PSpoint-partition} to write:
\begin{multline}
	\mdiscrep = \frac{1}{\avgpscore n} \sum_{i \in \TreatedPSp} \POci + \frac{1}{\avgpscore n} \sum_{i \in \TreatedPSm} \POci - \frac{1}{\avgpscore n} \sum_{i \in \MatchedPSp} \POci - \frac{1}{\avgpscore n} \sum_{i \in \MatchedPSm} \POci
	\\
	+ \frac{1}{\avgpscore n} \sum_{i \in \ControlsPSp} \POci - \frac{1}{\avgpscore n} \sum_{i \in \ControlsPSp} \POci.
\end{multline}
After rearranging terms, taking expectations and limits, we get:
\begin{align}
	\lim_{n \to \infty} \E{\mdiscrep}
	&= \lim_{n \to \infty} \E[\bigg]{\frac{1}{\avgpscore n} \sum_{i \in \ControlsPSp} \POci - \frac{1}{\avgpscore n} \sum_{i \in \MatchedPSp} \POci}
	\\
	&\qquad\quad + \lim_{n \to \infty} \E[\bigg]{\frac{1}{\avgpscore n} \sum_{i \in \TreatedPSm} \POci - \frac{1}{\avgpscore n} \sum_{i \in \MatchedPSm} \POci}
	\\
	&\qquad\quad + \lim_{n \to \infty} \E[\bigg]{\frac{1}{\avgpscore n} \sum_{i \in \TreatedPSp} \POci - \frac{1}{\avgpscore n} \sum_{i \in \ControlsPSp} \POci},
\end{align}
assuming the limit exists.
The first two terms are shown to be zero by Lemmas~\ref{lem:pscore-mdiscrep-term1} and~\ref{lem:pscore-mdiscrep-term2}.
Lemma~\ref{lem:pscore-mdiscrep-term3} completes the proof.
\end{proof}

\newcommand{\lemmamdiscrepbias}{%
	Given Condition~\refmain{cond:reg-cond}:
	\begin{equation}
		\liminf_{n \to \infty} \E{\mdiscrep} = \liminf_{n \to \infty} \paren[\big]{\E{\attest} - \att}
		\qquad\text{and}\qquad
		\limsup_{n \to \infty} \E{\mdiscrep} = \limsup_{n \to \infty} \paren[\big]{\E{\attest} - \att}.
	\end{equation}%
}

\subsection{Proofs of Lemmas~\ref*{lem:mdiscrep-bias},~\ref*{lem:mdiscrep-bias-step1},~\ref*{lem:bounded-POs},~\ref*{lem:mdiscrep-bias-step2},~\ref*{lem:mdiscrep-bias-step3}~and~\ref*{lem:mdiscrep-bias-step4}}

\begin{lemma}\label{lem:mdiscrep-bias}
	\lemmamdiscrepbias{}
\end{lemma}

\newcommand{\altmdiscrep}{\mdiscrep^\dagger}

\begin{proof}
Recall that:
\begin{equation}
	\mdiscrep = \frac{1}{\avgpscore n} \sum_{i \in \Treated} \POci - \frac{1}{\avgpscore n} \sum_{i \in \Matched} \POci.
\end{equation}
Let:
\begin{equation}
	\altmdiscrep = \frac{1}{\maxf{1, \nTreated}} \sum_{i \in \Treated} \POci - \frac{1}{\maxf{1, \nTreated}} \sum_{i \in \Matched} \POci,
\end{equation}
so that:
\begin{equation}
	\E{\mdiscrep} = \E{\altmdiscrep} + \E{\mdiscrep - \altmdiscrep}.
\end{equation}
Lemma~\ref{lem:mdiscrep-bias-step1} shows that $\lim_{n \to \infty} \E{\mdiscrep - \altmdiscrep} = 0$, which implies:
\begin{equation}
	\liminf_{n \to \infty} \E{\mdiscrep} = \liminf_{n \to \infty} \E{\altmdiscrep}
	\qquad\text{and}\qquad
	\limsup_{n \to \infty} \E{\mdiscrep} = \limsup_{n \to \infty} \E{\altmdiscrep}.
\end{equation}

Using the law of total expectation, write:
\begin{multline}
	\E{\altmdiscrep}
	=
	\Pr{\nTreated = 0} \E[\big]{\altmdiscrep \given \nTreated = 0} + \Pr{1 \leq \nTreated \leq \nControls} \E[\big]{\altmdiscrep \given 1 \leq \nTreated \leq \nControls}
	\\
	+ \Pr{\nTreated > \nControls} \E[\big]{\altmdiscrep \given \nTreated > \nControls}.
\end{multline}
Note $\E[\big]{\altmdiscrep \given \nTreated = 0} = 0$ and:
\begin{equation}
	\Pr{1 \leq \nTreated \leq \nControls} = 1 - \Pr{\nTreated = 0} - \Pr{\nTreated > \nControls},
\end{equation}
so:
\begin{multline}
	\E{\altmdiscrep}
	=
	\E[\big]{\altmdiscrep \given 1 \leq \nTreated \leq \nControls}
	- \Pr{\nTreated = 0} \E[\big]{\altmdiscrep \given 1 \leq \nTreated \leq \nControls}
	\\
	+ \Pr{\nTreated > \nControls} \paren[\Big]{\E[\big]{\altmdiscrep \given \nTreated > \nControls} - \E[\big]{\altmdiscrep \given 1 \leq \nTreated \leq \nControls}}.
\end{multline}
Lemma~\ref{lem:mdiscrep-bias-step2} therefore implies:
\begin{multline}
	\liminf_{n \to \infty} \E{\altmdiscrep} = \liminf_{n \to \infty} \E[\big]{\altmdiscrep \given 1 \leq \nTreated \leq \nControls}
	\qquad\text{and}\qquad
	\\
	\limsup_{n \to \infty} \E{\altmdiscrep} = \limsup_{n \to \infty} \E[\big]{\altmdiscrep \given 1 \leq \nTreated \leq \nControls}.
\end{multline}

By a similar argument:
\begin{multline}
	\E{\attest}
	=
	\Pr{\nTreated = 0} \E[\big]{\attest \given \nTreated = 0} + \Pr{1 \leq \nTreated \leq \nControls} \E[\big]{\attest \given 1 \leq \nTreated \leq \nControls}
	\\
	+ \Pr{\nTreated > \nControls} \E[\big]{\attest \given \nTreated > \nControls}.
\end{multline}
Recall $\attest = 0$ when $\nTreated = 0$ or $\nTreated > \nControls$, so:
\begin{equation}
	\E{\attest}
	=
	\E[\big]{\attest \given 1 \leq \nTreated \leq \nControls}
	- \bracket[\big]{\Pr{\nTreated = 0} + \Pr{\nTreated > \nControls}} \E[\big]{\attest \given 1 \leq \nTreated \leq \nControls}.
\end{equation}
Lemma~\ref{lem:mdiscrep-bias-step3} then implies:
\begin{multline}
	\liminf_{n \to \infty} \E{\attest} = \liminf_{n \to \infty} \E[\big]{\attest \given 1 \leq \nTreated \leq \nControls}
	\qquad\text{and}\qquad
	\\
	\limsup_{n \to \infty} \E{\attest} = \limsup_{n \to \infty} \E[\big]{\attest \given 1 \leq \nTreated \leq \nControls}.
\end{multline}
Lemma~\ref{lem:mdiscrep-bias-step4} completes the proof by showing:
\begin{equation}
	\E[\big]{\attest \given 1 \leq \nTreated \leq \nControls} - \att = \E[\big]{\altmdiscrep \given 1 \leq \nTreated \leq \nControls}. \tag*{\qedhere}
\end{equation}
\end{proof}

\begin{lemma}\label{lem:mdiscrep-bias-step1}
	Under Condition~\refmain{cond:reg-cond}:
	\begin{equation}
		\lim_{n \to \infty} \E[\Bigg]{\paren[\bigg]{\frac{1}{n \avgpscore} - \frac{1}{\maxf{1, \nTreated}}} \paren[\bigg]{\sum_{i \in \Treated} \POci - \sum_{i \in \Matched} \POci}} = 0.
	\end{equation}
\end{lemma}

\begin{proof}
Rearrange the factors as:
\begin{multline}
	\paren[\bigg]{\frac{1}{\avgpscore n} - \frac{1}{\maxf{1, \nTreated}}} \paren[\bigg]{\sum_{i \in \Treated} \POci - \sum_{i \in \Matched} \POci}
	\\
	= \frac{\maxf{1, \nTreated} - \avgpscore n}{\maxf{1, \nTreated}} \paren[\bigg]{\frac{1}{\avgpscore n} \sum_{i \in \Treated} \POci - \frac{1}{\avgpscore n} \sum_{i \in \Matched} \POci},
\end{multline}
and then bound the expectation as:
\begin{multline}
	\E[\Bigg]{\paren[\bigg]{\frac{\maxf{1, \nTreated} - \avgpscore n}{\maxf{1, \nTreated}}} \paren[\bigg]{\frac{1}{\avgpscore n} \sum_{i \in \Treated} \POci - \frac{1}{\avgpscore n} \sum_{i \in \Matched} \POci}}
	\\
	\leq \E[\Bigg]{\frac{\abs{\maxf{1, \nTreated} - \avgpscore n}}{\maxf{1, \nTreated}} \paren[\bigg]{\frac{1}{\avgpscore n} \sum_{i \in \Treated} \abs{\POci} + \frac{1}{\avgpscore n} \sum_{i \in \Matched} \abs{\POci}}}.
\end{multline}
Note that $\Matched \subseteq \Controls$, so:
\begin{equation}
	\sum_{i \in \Matched} \abs{\POci}
	\leq \sum_{i \in \Controls} \abs{\POci},
\end{equation}
and use the law of iterated expectations to get:
\begin{multline}
	\E[\Bigg]{\frac{\abs{\maxf{1, \nTreated} - \avgpscore n}}{\maxf{1, \nTreated}} \paren[\bigg]{\frac{1}{\avgpscore n} \sum_{i \in \Treated} \abs{\POci} + \frac{1}{\avgpscore n} \sum_{i \in \Controls} \abs{\POci}}}
	\\
	=
	\E[\Bigg]{\frac{\abs{\maxf{1, \nTreated} - \avgpscore n}}{\maxf{1, \nTreated}} \E[\bigg]{\frac{1}{\avgpscore n} \sum_{i \in \Treated} \abs{\POci} + \frac{1}{\avgpscore n} \sum_{i \in \Controls} \abs{\POci} \given \Treated}}.
\end{multline}
The set $\Treated$ contains no more information about $\POci$ than $\Wi$, so:
\begin{equation}
	\E[\bigg]{\frac{1}{\avgpscore n} \sum_{i \in \Treated} \abs{\POci} + \frac{1}{\avgpscore n} \sum_{i \in \Controls} \abs{\POci} \given \Treated}
	= \frac{\nTreated}{\avgpscore n} \E[\big]{\abs{\POcpop} \given \Wpop = 1} + \frac{\nControls}{\avgpscore n} \E[\big]{\abs{\POcpop} \given \Wpop = 0}.
\end{equation}
The law of total expectation gives:
\begin{equation}
	\E[\big]{\abs{\POcpop}} = \avgpscore \E[\big]{\abs{\POcpop} \given \Wpop = 1} + \paren{1 - \avgpscore}\E[\big]{\abs{\POcpop} \given \Wpop = 0},
\end{equation}
so:
\begin{equation}
	\E[\big]{\abs{\POcpop} \given \Wpop = 1} \leq \frac{\E[\big]{\abs{\POcpop}}}{\avgpscore}.
\end{equation}
Lemma~\ref{lem:bounded-POs} ensures that $\E[\big]{\abs{\POcpop}}$ exists.
By a similar argument:
\begin{equation}
	\E[\big]{\abs{\POcpop} \given \Wpop = 0} \leq \frac{\E[\big]{\abs{\POcpop}}}{1 - \avgpscore} \leq \frac{\E[\big]{\abs{\POcpop}}}{\avgpscore}.
\end{equation}
Because $\nTreated + \nControls = n$:
\begin{equation}
	\frac{\nTreated}{\avgpscore n} \E[\big]{\abs{\POcpop} \given \Wpop = 1} + \frac{\nControls}{\avgpscore n} \E[\big]{\abs{\POcpop} \given \Wpop = 0}
	\leq
	\frac{\E[\big]{\abs{\POcpop}}}{\avgpscore^2}.
\end{equation}
It follows that:
\begin{multline}
	\E[\Bigg]{\frac{\abs{\maxf{1, \nTreated} - \avgpscore n}}{\maxf{1, \nTreated}} \E[\bigg]{\frac{1}{\avgpscore n} \sum_{i \in \Treated} \abs{\POci} + \frac{1}{\avgpscore n} \sum_{i \in \Controls} \abs{\POci} \given \Treated}}
	\\
	\leq
	\E[\Bigg]{\frac{\abs{\maxf{1, \nTreated} - \avgpscore n}}{\maxf{1, \nTreated}}} \frac{\E[\big]{\abs{\POcpop}}}{\avgpscore^2}.
\end{multline}

Consider the expectation for large samples.
In particular, consider when $\avgpscore^2 n / 4 \geq \logf{n} > 1$.
By the law of total expectation:
\begin{align}
	&\E[\Bigg]{\frac{\abs{\maxf{1, \nTreated} - \avgpscore n}}{\maxf{1, \nTreated}}}
	\\
	&\quad =
	\Pr[\Big]{\abs[\big]{\nTreated - \avgpscore n} < \sqrt{n \logf{n}}} \E[\Bigg]{\frac{\abs{\maxf{1, \nTreated} - \avgpscore n}}{\maxf{1, \nTreated}} \given \abs[\big]{\nTreated - \avgpscore n} < \sqrt{n \logf{n}}}
	\\
	&\qquad +
	\Pr[\Big]{\abs[\big]{\nTreated - \avgpscore n} \geq \sqrt{n \logf{n}}} \E[\Bigg]{\frac{\abs{\maxf{1, \nTreated} - \avgpscore n}}{\maxf{1, \nTreated}} \given \abs[\big]{\nTreated - \avgpscore n} \geq \sqrt{n \logf{n}}}.
\end{align}

The first probability is crudely bounded as:
\begin{equation}
	\Pr[\Big]{\abs[\big]{\nTreated - \avgpscore n} < \sqrt{n \logf{n}}} \leq 1.
\end{equation}
Recall that $\avgpscore^2 n / 4 \geq \logf{n}$, which together with $\abs[\big]{\nTreated - \avgpscore n} < \sqrt{n \logf{n}}$, implies that $\nTreated > \avgpscore n / 2$.
Use this to bound the first expectation as:
\begin{equation}
	\E[\Bigg]{\frac{\abs{\maxf{1, \nTreated} - \avgpscore n}}{\maxf{1, \nTreated}} \given \abs[\big]{\nTreated - \avgpscore n} \leq \sqrt{n \logf{n}}}
	\leq
	\frac{2}{\avgpscore}\sqrt{\frac{\logf{n}}{n}}.
\end{equation}

Bound the second probability using Hoeffding's inequality (Theorem~\ref{thm:hoeffdings-inequality}):
\begin{equation}
	\Pr[\Big]{\abs[\big]{\nTreated - \avgpscore n} \geq \sqrt{n \logf{n}}} = \Pr[\Big]{\abs[\big]{\nTreated / n - \avgpscore} \geq \sqrt{\logf{n} / n}} \leq 2 \expf{-2\logf{n}} = \frac{2}{n^2},
\end{equation}
and the second expectation is bounded as:
\begin{equation}
	\E[\Bigg]{\frac{\abs{\maxf{1, \nTreated} - \avgpscore n}}{\maxf{1, \nTreated}} \given \abs[\big]{\nTreated - \avgpscore n} \geq \sqrt{n \logf{n}}}
	\leq
	n.
\end{equation}
Taken together, when $\avgpscore^2 n / 4 \geq \logf{n} > 1$:
\begin{equation}
	\E[\Bigg]{\frac{\abs{\maxf{1, \nTreated} - \avgpscore n}}{\maxf{1, \nTreated}}} \leq \frac{2}{\avgpscore}\sqrt{\frac{\logf{n}}{n}} + \frac{2}{n} \leq \frac{4}{\avgpscore}\sqrt{\frac{\logf{n}}{n}}.
\end{equation}

Returning to the full expression:
\begin{equation}
	\E[\Bigg]{\frac{\abs{\maxf{1, \nTreated} - \avgpscore n}}{\maxf{1, \nTreated}}} \frac{\E[\big]{\abs{\POcpop}}}{\avgpscore^2}
	\leq
\frac{4 \E[\big]{\abs{\POcpop}}}{\avgpscore^3} \sqrt{\frac{\logf{n}}{n}}. \tag*{\qedhere}
\end{equation}
\end{proof}

\begin{lemma}\label{lem:bounded-POs}
	Under Condition~\refmain{cond:reg-cond}, $\E[\big]{\abs{\POcpop}}$ exists.
\end{lemma}

\begin{proof}
Condition~\refmain{cond:reg-cond} states that $\E{\Ypop^2}$ exists.
Together with Lyapunov's inequality, this implies that $\E[\big]{\abs{\Ypop}}$ exists as well.
By the law of total expectation:
\begin{equation}
	\E[\big]{\abs{\Ypop}} = \avgpscore \E[\big]{\abs{\Ypop} \given \Wpop = 1} + \paren{1 - \avgpscore}\E[\big]{\abs{\Ypop} \given \Wpop = 0} \geq \paren{1 - \avgpscore} \E[\big]{\abs{\Ypop} \given \Wpop = 0},
\end{equation}
so $\E[\big]{\abs{\POcpop} \given \Wpop = 0} = \E[\big]{\abs{\Ypop} \given \Wpop = 0}$ exists.

Using the law of iterated expectations and unconfoundedness with respect to the propensity score (Theorem~\ref{thm:pscore-balancing}):
\begin{equation}
	\E[\big]{\abs{\POcpop} \given \Wpop = 0}
	= \E[\Big]{\E[\big]{\abs{\POcpop} \given \rPSpop, \Wpop = 0} \given \Wpop = 0}
	= \E[\Big]{\E[\big]{\abs{\POcpop} \given \rPSpop} \given \Wpop = 0},
\end{equation}
Assuming for the moment that $\E[\big]{\abs{\POcpop} \given \Wpop = 1}$ exists, then:
\begin{equation}
	\E[\big]{\abs{\POcpop} \given \Wpop = 1}
	= \E[\Big]{\E[\big]{\abs{\POcpop} \given \rPSpop} \given \Wpop = 1}
	= \frac{1 - \avgpscore}{\avgpscore} \E[\bigg]{\frac{\rPSpop}{1 - \rPSpop}\E[\big]{\abs{\POcpop} \given \rPSpop} \given \Wpop = 0},
\end{equation}
because $\rPSpop / \paren{1 - \rPSpop}$ is the ratio of treated units to control units at each value of the propensity score in the population, which is what the expectation is marginalizing over.
However, because the control units are more numerous in the population, the $\paren{1 - \avgpscore} / \avgpscore$ factor is needed for normalization.

Condition~\refmain{cond:reg-cond} states that the propensity score is bounded away from one.
Hence, $\rPSpop \leq \suppPSp < 1$ with probability one, where $\suppPSp = \sup \suppPS$ is the upper bound of the support of $\rPSpop$.
Thus, with probability one:
\begin{equation}
	\frac{\rPSpop}{1 - \rPSpop} \leq \frac{1}{1 - \suppPSp},
\end{equation}
and:
\begin{multline}
	\frac{1 - \avgpscore}{\avgpscore} \E[\bigg]{\frac{\rPSpop}{1 - \rPSpop}\E[\big]{\abs{\POcpop} \given \rPSpop} \given \Wpop = 0}
	\leq
	\frac{1}{\avgpscore \paren{1 - \suppPSp}} \E[\Big]{\E[\big]{\abs{\POcpop} \given \rPSpop} \given \Wpop = 0}
	\\
	=
	\frac{\E[\big]{\abs{\POcpop} \given \Wpop = 0}}{\avgpscore \paren{1 - \suppPSp}},
\end{multline}
where $\avgpscore > 0$ according to Condition~\refmain{cond:reg-cond}.
The conclusion is that $\E[\big]{\abs{\POcpop} \given \Wpop = 1}$ exists.
It follows from the law of total expectation that:
\begin{equation}
	\avgpscore \E[\big]{\abs{\POcpop} \given \Wpop = 1} + \paren{1 - \avgpscore} \E[\big]{\abs{\POcpop} \given \Wpop = 0}
	= \E[\big]{\abs{\POcpop}},
\end{equation}
exists as well.
\end{proof}

\begin{lemma}\label{lem:mdiscrep-bias-step2}
	Under Condition~\refmain{cond:reg-cond}:
	\begin{multline}
		\lim_{n \to \infty} \Pr{\nTreated = 0} \E[\big]{\altmdiscrep \given 1 \leq \nTreated \leq \nControls} = 0
		\qquad\text{and}
		\\
		\lim_{n \to \infty} \Pr{\nTreated > \nControls} \paren[\Big]{\E[\big]{\altmdiscrep \given \nTreated > \nControls} - \E[\big]{\altmdiscrep \given 1 \leq \nTreated \leq \nControls}} = 0,
	\end{multline}
	where $\altmdiscrep$ is defined in the proof of Lemma~\ref{lem:mdiscrep-bias}.
\end{lemma}

\begin{proof}
Starting with the probabilities, note that $\Pr{\nTreated = 0} = \paren{1 - \avgpscore}^n$.
Next:
\begin{equation}
	\Pr{\nTreated > \nControls} = \Pr{\nTreated / n > 1 / 2} = \Pr[\big]{\nTreated / n - \avgpscore > \paren{1 - 2\avgpscore} / 2}.
\end{equation}
Note that $\avgpscore < 1 / 2$, so by Hoeffding's inequality (Theorem~\ref{thm:hoeffdings-inequality}):
\begin{equation}
	\Pr[\big]{\nTreated / n - \avgpscore > \paren{1 - 2\avgpscore} / 2} \leq \expf[\big]{-n\paren{1 - 2\avgpscore}^2 / 2}.
\end{equation}
It follows that:
\begin{equation}
	\lim_{n \to \infty} \Pr{\nTreated = 0} = 0
	\qquad\text{and}\qquad
	\lim_{n \to \infty} \Pr{\nTreated > \nControls} = 0.
\end{equation}

Now consider the expectations.
Note $\Matched \subseteq \Controls$, so:
\begin{equation}
	\abs[\big]{\E{\altmdiscrep \given \Treated}} \leq \E[\bigg]{\frac{1}{\nTreated} \sum_{i \in \Treated} \abs{\POci} + \frac{1}{\nTreated} \sum_{i \in \Controls} \abs{\POci} \given \Treated},
\end{equation}
when $\nTreated \geq 1$.
The set $\Treated$ contains no more information about $\POci$ than $\Wi$, so:
\begin{equation}
	\E[\bigg]{\frac{1}{\nTreated} \sum_{i \in \Treated} \abs{\POci} + \frac{1}{\nTreated} \sum_{i \in \Controls} \abs{\POci} \given \Treated}
	= \E[\big]{\abs{\POcpop} \given \Wpop = 1} + \frac{\nControls}{\nTreated}\E[\big]{\abs{\POcpop} \given \Wpop = 0}.
\end{equation}
As shown in the proof of Lemma~\ref{lem:mdiscrep-bias-step1}:
\begin{equation}
	\E[\big]{\abs{\POcpop} \given \Wpop = 0} \leq \frac{\E[\big]{\abs{\POcpop}}}{\avgpscore} \qquad\text{and}\qquad \E[\big]{\abs{\POcpop} \given \Wpop = 1} \leq \frac{\E[\big]{\abs{\POcpop}}}{\avgpscore},
\end{equation}
so:
\begin{equation}
	\E[\big]{\abs{\POcpop} \given \Wpop = 1} + \frac{\nControls}{\nTreated}\E[\big]{\abs{\POcpop} \given \Wpop = 0}
	\leq
	\paren[\bigg]{1 + \frac{\nControls}{\nTreated}} \frac{\E[\big]{\abs{\POcpop}}}{\avgpscore}.
\end{equation}
One of the expectations is thus bounded as:
\begin{equation}
	\abs[\big]{\E{\altmdiscrep \given \nTreated > \nControls}} \leq \E[\bigg]{1 + \frac{\nControls}{\nTreated} \given \nTreated > \nControls} \frac{\E[\big]{\abs{\POcpop}}}{\avgpscore} \leq \frac{2 \E[\big]{\abs{\POcpop}}}{\avgpscore}.
\end{equation}

For the other expectation, write:
\begin{equation}
	\abs[\big]{\E{\altmdiscrep \given 1 \leq \nTreated \leq \nControls}} \leq \E[\bigg]{\frac{n}{\nTreated} \given 1 \leq \nTreated \leq \nControls} \frac{\E[\big]{\abs{\POcpop}}}{\avgpscore}.
\end{equation}
Consider samples large enough to satisfy $\avgpscore n \geq 2$, and:
\begin{multline}
	\E[\bigg]{\frac{n}{\nTreated} \given 1 \leq \nTreated \leq \nControls}
	= \frac{n \Pr[\big]{1 \leq \nTreated \leq \avgpscore n / 2}}{\Pr[\big]{1 \leq \nTreated \leq \nControls}} \E[\bigg]{\frac{1}{\nTreated} \given 1 \leq \nTreated \leq \avgpscore n / 2}
	\\
	+ \frac{\Pr[\big]{\avgpscore n / 2 < \nTreated \leq \nControls}}{\Pr[\big]{1 \leq \nTreated \leq \nControls}} \E[\bigg]{\frac{n}{\nTreated} \given \avgpscore n / 2 < \nTreated \leq \nControls}.
\end{multline}
By Hoeffding's inequality (Theorem~\ref{thm:hoeffdings-inequality}):
\begin{equation}
	n \Pr[\big]{1 \leq \nTreated \leq \avgpscore n / 2} \leq n \Pr[\big]{\nTreated \leq \avgpscore n / 2} = n \Pr[\big]{\nTreated / n - \avgpscore \leq -\avgpscore / 2} \leq n \expf{-n\avgpscore^2 / 2},
\end{equation}
so $\lim_{n \to \infty} n \Pr[\big]{1 \leq \nTreated \leq \avgpscore n / 2} = 0$.
The start of the proof implies:
\begin{equation}
	\lim_{n \to \infty} \Pr[\big]{1 \leq \nTreated \leq \nControls} = 1,
\end{equation}
because $\Pr[\big]{1 \leq \nTreated \leq \nControls} = 1 - \Pr{\nTreated = 0} - \Pr{\nTreated > \nControls}$.
Bound the other parts as:
\begin{multline}
	\E[\bigg]{\frac{1}{\nTreated} \given 1 \leq \nTreated \leq \avgpscore n / 2} \leq 1,
	\qquad
	\Pr[\big]{\avgpscore n / 2 < \nTreated \leq \nControls} \leq 1
	\qquad\text{and}
	\\
	\E[\bigg]{\frac{n}{\nTreated} \given \avgpscore n / 2 < \nTreated \leq \nControls} \leq \frac{2}{\avgpscore}.
\end{multline}
Hence:
\begin{multline}
	\limsup_{n \to \infty} \E[\bigg]{\frac{n}{\nTreated} \given 1 \leq \nTreated \leq \nControls} \leq \frac{2}{\avgpscore}
	\qquad\text{and}\qquad
	\\
	\limsup_{n \to \infty} \abs[\big]{\E{\altmdiscrep \given 1 \leq \nTreated \leq \nControls}} \leq \frac{2 \E[\big]{\abs{\POcpop}}}{\avgpscore^2}. \tag*{\qedhere}
\end{multline}
\end{proof}

\begin{lemma}\label{lem:mdiscrep-bias-step3}
	Under Condition~\refmain{cond:reg-cond}:
	\begin{equation}
		\lim_{n \to \infty} \bracket[\big]{\Pr{\nTreated = 0} + \Pr{\nTreated > \nControls}} \E[\big]{\attest \given 1 \leq \nTreated \leq \nControls} = 0.
	\end{equation}
\end{lemma}

\begin{proof}
The proof follows the structure of the proof of Lemma~\ref{lem:mdiscrep-bias-step2} closely.
It was there shown that:
\begin{equation}
	\lim_{n \to \infty} \Pr{\nTreated = 0} = 0
	\qquad\text{and}\qquad
	\lim_{n \to \infty} \Pr{\nTreated > \nControls} = 0.
\end{equation}
It remains to show that the expectation is bounded.
Recall that $\Matched = \setb{\mopt{i} \suchthat i \in \Treated}$ when $1 \leq \nTreated \leq \nControls$, so in that case:
\begin{equation}
	\attest = \frac{1}{\nTreated} \sum_{i \in \Treated} \paren[\big]{\Yi - \Y{\mopt{i}}} = \frac{1}{\nTreated} \sum_{i \in \Treated} \Yi + \frac{1}{\nTreated} \sum_{i \in \Matched} \Yi.
\end{equation}
As above, note $\Matched \subseteq \Controls$, so when $1 \leq \nTreated \leq \nControls$:
\begin{equation}
	\abs[\big]{\E{\attest \given \Treated}} \leq \E[\bigg]{\frac{1}{\nTreated} \sum_{i \in \Treated} \abs{\Yi} + \frac{1}{\nTreated} \sum_{i \in \Controls} \abs{\Yi} \given \Treated}.
\end{equation}
The set $\Treated$ contains no more information about $\Yi$ than $\Wi$, so:
\begin{equation}
	\E[\bigg]{\frac{1}{\nTreated} \sum_{i \in \Treated} \abs{\Yi} + \frac{1}{\nTreated} \sum_{i \in \Controls} \abs{\Yi} \given \Treated}
	= \E[\big]{\abs{\Ypop} \given \Wpop = 1} + \frac{\nControls}{\nTreated}\E[\big]{\abs{\Ypop} \given \Wpop = 0}.
\end{equation}
and by the same argument as in the previous proofs:
\begin{equation}
	\E[\big]{\abs{\Ypop} \given \Wpop = 1} + \frac{\nControls}{\nTreated}\E[\big]{\abs{\Ypop} \given \Wpop = 0} \leq
	\paren[\bigg]{1 + \frac{\nControls}{\nTreated}} \frac{\E[\big]{\abs{\Ypop}}}{\avgpscore},
\end{equation}
where Condition~\refmain{cond:reg-cond} ensures that $\E[\big]{\abs{\Ypop}}$ exists.
It follows that:
\begin{equation}
	\abs[\big]{\E{\attest \given 1 \leq \nTreated \leq \nControls}} \leq \E[\bigg]{\frac{n}{\nTreated} \given 1 \leq \nTreated \leq \nControls} \frac{\E[\big]{\abs{\Ypop}}}{\avgpscore}.
\end{equation}
The first expectation on the right-hand side was shown to be asymptotically bounded in the proof of Lemma~\ref{lem:mdiscrep-bias-step2}, so:
\begin{equation}
	\limsup_{n \to \infty} \abs[\big]{\E{\attest \given 1 \leq \nTreated \leq \nControls}} \leq \frac{2 \E[\big]{\abs{\Ypop}}}{\avgpscore^2}. \tag*{\qedhere}
\end{equation}
\end{proof}

\begin{lemma}\label{lem:mdiscrep-bias-step4}
	\begin{equation}
		\E[\big]{\attest \given 1 \leq \nTreated \leq \nControls} - \att = \E[\bigg]{\frac{1}{\nTreated} \sum_{i \in \Treated} \POci - \frac{1}{\nTreated} \sum_{i \in \Matched} \POci \given 1 \leq \nTreated \leq \nControls}.
	\end{equation}
\end{lemma}

\begin{proof}
In was shown in the proof of Lemma~\ref{lem:mdiscrep-bias-step3} that:
\begin{equation}
	\attest = \frac{1}{\nTreated} \sum_{i \in \Treated} \Yi - \frac{1}{\nTreated} \sum_{i \in \Matched} \Yi,
\end{equation}
when $1 \leq \nTreated \leq \nControls$.
Add and subtract $\nTreated^{-1}\sum_{i \in \Treated} \POci$ from the estimator to get:
\begin{multline}
	\E[\big]{\attest \given 1 \leq \nTreated \leq \nControls} = \E[\bigg]{\frac{1}{\nTreated} \sum_{i \in \Treated} \POci - \frac{1}{\nTreated} \sum_{i \in \Matched} \POci \given 1 \leq \nTreated \leq \nControls}
	\\
	+ \E[\bigg]{\frac{1}{\nTreated} \sum_{i \in \Treated} \POti - \frac{1}{\nTreated} \sum_{i \in \Treated} \POci \given 1 \leq \nTreated \leq \nControls}.
\end{multline}
where $\Yi = \POti$ for $i \in \Treated$ and $\Yi = \POci$ for $i \in \Matched \subseteq \Controls$ was used.
As above, $\Treated$ contains no more information about $\POci$ or $\POti$ than $\Wi$, so when $\nTreated \geq 1$:
\begin{equation}
	\E[\bigg]{\frac{1}{\nTreated} \sum_{i \in \Treated} \POti \given \Treated} = \E{\POtpop \given \Wpop = 1}
	\quad\text{and}\quad
	\E[\bigg]{\frac{1}{\nTreated} \sum_{i \in \Treated} \POci \given \Treated} = \E{\POcpop \given \Wpop = 1}.
\end{equation}
It follows that:
\begin{equation}
	\E[\bigg]{\frac{1}{\nTreated} \sum_{i \in \Treated} \POti - \frac{1}{\nTreated} \sum_{i \in \Treated} \POci \given 1 \leq \nTreated \leq \nControls} = \E{\POtpop - \POcpop \given \Wpop = 1} = \att. \tag*{\qedhere}
\end{equation}
\end{proof}

\subsection{Proofs of Lemmas~\ref*{lem:pscore-mdiscrep-term1},~\ref*{lem:pscore-mdiscrep-term1-no-upper-unmatched},~\ref*{lem:pscore-mdiscrep-term1-no-upper-unmatched-middle-part},~\ref*{lem:pscore-mdiscrep-term1-no-upper-unmatched-atom-point}~and~\ref*{lem:pscore-mdiscrep-term1-upper-balance}}

\begin{lemma}\label{lem:pscore-mdiscrep-term1}
	Under Conditions~\refmain{cond:reg-cond},~\refmain{cond:pscore-lipschitz} and~\refmain{cond:left-closed}:
	\begin{equation}
		\lim_{n \to \infty} \E[\bigg]{\frac{1}{\avgpscore n} \sum_{i \in \ControlsPSp} \POci - \frac{1}{\avgpscore n} \sum_{i \in \MatchedPSp} \POci} = 0.
	\end{equation}
\end{lemma}

\begin{proof}
The matching only depends on $\W{1}, \dotsc, \W{n}$ and $\rPS{1}, \dotsc, \rPS{n}$, so $\MatchedPSp$ is determined by those variables, and:
\begin{multline}
	\E[\bigg]{\frac{1}{\avgpscore n} \sum_{i \in \MatchedPSp} \POci}
	=
	\E[\Bigg]{\E[\bigg]{\frac{1}{\avgpscore n} \sum_{i \in \MatchedPSp} \POci \given \W{1}, \dotsc, \W{n}, \rPS{1}, \dotsc, \rPS{n}}}
	\\
	= \E[\bigg]{\frac{1}{\avgpscore n} \sum_{i \in \MatchedPSp} \E[\big]{\POci \given \Wi, \rPSi}}
	= \E[\bigg]{\frac{1}{\avgpscore n} \sum_{i \in \MatchedPSp} \EPOPSc{\rPSi}},
\end{multline}
where the last equality follows from unconfoundedness with respect to the propensity score (Theorem~\ref{thm:pscore-balancing}).
By a similar argument:
\begin{equation}
	\E[\bigg]{\frac{1}{\avgpscore n} \sum_{i \in \ControlsPSp} \POci}
	= \E[\bigg]{\frac{1}{\avgpscore n} \sum_{i \in \ControlsPSp} \EPOPSc{\rPSi}}.
\end{equation}

Partition $\ControlsPSp$ and $\MatchedPSp$ as:
\begin{multline}
	\ControlsPSp = \paren[\big]{\ControlsPSp \setminus \paren{\ControlsPSp \cap \MatchedPSp}} \cup \paren{\ControlsPSp \cap \MatchedPSp}
	\qquad\text{and}
	\\
	\MatchedPSp = \paren[\big]{\MatchedPSp \setminus \paren{\ControlsPSp \cap \MatchedPSp}} \cup \paren{\ControlsPSp \cap \MatchedPSp}.
\end{multline}
Observe that the operands in both unions are disjoint, so we can write:
\begin{equation}
	\sum_{i \in \ControlsPSp} \EPOPSc{\rPSi} - \sum_{i \in \MatchedPSp} \EPOPSc{\rPSi}
	= \sum_{i \in \ControlsPSp \setminus \paren{\ControlsPSp \cap \MatchedPSp}} \EPOPSc{\rPSi} - \sum_{i \in \MatchedPSp \setminus \paren{\ControlsPSp \cap \MatchedPSp}} \EPOPSc{\rPSi},
\end{equation}
because the two terms summing over $\ControlsPSp \cap \MatchedPSp$ cancel.
It follows that:
\begin{equation}
	\frac{1}{\avgpscore n} \sum_{i \in \ControlsPSp} \EPOPSc{\rPSi} - \frac{1}{\avgpscore n} \sum_{i \in \MatchedPSp} \EPOPSc{\rPSi}
	\leq \frac{1}{\avgpscore n} \sum_{i \in \mathbf{A}} \abs{\EPOPSc{\rPSi}},
\end{equation}
where $\mathbf{A} = \paren[\big]{\ControlsPSp \setminus \paren{\ControlsPSp \cap \MatchedPSp}} \cup \paren[\big]{\MatchedPSp \setminus \paren{\ControlsPSp \cap \MatchedPSp}}$.

Condition~\refmain{cond:reg-cond} implies that $\EPOPSc{a}$ exists for some $a \in \suppPS$.
By Condition~\refmain{cond:pscore-lipschitz}, $\EPOPSc{p}$ is Lipschitz continuous on $\suppPS$.
Together with $\suppPS \subset \bracket{0, 1}$, this implies that $\EPOPSc{p}$ is bounded on $\suppPS$.
Let $c_\mu \in \Reals$ be this bound.
That is, $\abs{\EPOPSc{p}} \leq c_\mu$ for all $p \in \suppPS$.
It follows that:
\begin{equation}
	\E[\bigg]{\frac{1}{\avgpscore n} \sum_{i \in \mathbf{A}} \abs{\EPOPSc{\rPSi}}}
	\leq
	\frac{c_\mu}{\avgpscore} \E[\bigg]{\frac{\card{\mathbf{A}}}{n}}.
\end{equation}

Because the union operands in $\mathbf{A}$ are disjoint:
\begin{equation}
	\card{\mathbf{A}} = \card[\big]{\ControlsPSp \setminus \paren{\ControlsPSp \cap \MatchedPSp}} + \card[\big]{\MatchedPSp \setminus \paren{\ControlsPSp \cap \MatchedPSp}},
\end{equation}
and because $\paren{\ControlsPSp \cap \MatchedPSp} \subset \ControlsPSp$ and $\paren{\ControlsPSp \cap \MatchedPSp} \subset \MatchedPSp$:
\begin{align}
	\card{\mathbf{A}}
	&= \card{\ControlsPSp} - \card{\ControlsPSp \cap \MatchedPSp} + \card{\MatchedPSp} - \card{\ControlsPSp \cap \MatchedPSp}
	\\
	&= 2\paren[\big]{\card{\ControlsPSp} - \card{\ControlsPSp \cap \MatchedPSp}} + \paren[\big]{\card{\MatchedPSp} - \card{\ControlsPSp}}
	\\
	&= 2\card{\ControlsPSp \setminus \MatchedPSp} + \paren[\big]{\card{\MatchedPSp} - \card{\ControlsPSp}},
\end{align}
where the last equality follows from:
\begin{equation}
	\card{\ControlsPSp \setminus \MatchedPSp} = \card{\ControlsPSp \setminus \paren{\ControlsPSp \cap \MatchedPSp}} = \card{\ControlsPSp} - \card{\ControlsPSp \cap \MatchedPSp}.
\end{equation}
The proof is completed by:
\begin{equation}
	\frac{c_\mu}{\avgpscore} \E[\bigg]{\frac{\card{\mathbf{A}}}{n}}
	=
	\frac{2c_\mu}{\avgpscore} \E[\bigg]{\frac{\card{\ControlsPSp \setminus \MatchedPSp}}{n}}
	+
	\frac{c_\mu}{\avgpscore} \E[\bigg]{\frac{\card{\MatchedPSp} - \card{\ControlsPSp}}{n}},
\end{equation}
and Lemmas~\ref{lem:pscore-mdiscrep-term1-no-upper-unmatched} and~\ref{lem:pscore-mdiscrep-term1-upper-balance}.
\end{proof}

\begin{lemma}\label{lem:pscore-mdiscrep-term1-no-upper-unmatched}
	Under Conditions~\refmain{cond:reg-cond},~\refmain{cond:pscore-lipschitz} and~\refmain{cond:left-closed}:
	\begin{equation}
		\lim_{n \to \infty} \E[\bigg]{\frac{\card{\ControlsPSp \setminus \MatchedPSp}}{n}} = 0.
	\end{equation}
\end{lemma}

\newcommand{\UpperEpsiPSpoint}{\varepsilon^+_n}
\newcommand{\LowerEpsiPSpoint}{\varepsilon^-_n}

\DeclarePairedDelimiterXPP\IntervalPSpoint[1]{K_n}{\lparen}{\rparen}{}{#1}

\begin{proof}
Units in $\ControlsPSp \setminus \MatchedPSp$ are control units with $\rPSi \geq \PSpoint$ not in $\MatchedPSp$, so:
\begin{multline}
	\E[\bigg]{\frac{\card{\ControlsPSp \setminus \MatchedPSp}}{n}}
	= \E[\bigg]{\avgin \indicator[\big]{\IMatchedPSpi = 0, \Wi = 0, \rPSi \geq \PSpoint}}
	\\
	= \avgin \Pr[\big]{\IMatchedPSpi = 0, \Wi = 0, \rPSi \geq \PSpoint},
\end{multline}
where $\IMatchedPSpi = \indicator{i \in \MatchedPSp}$.
The summation on the right-hand side of the equation cannot be removed at this point because the matching may not be symmetric with respect to the unit indices.
For example, if tie breaking is done by picking units with lower indices as matches, then $\Pr{\IMatchedPSp{1} = 0, \W{1} = 0, \rPS{1} \geq \PSpoint}$ may be less than $\Pr{\IMatchedPSp{n} = 0, \W{n} = 0, \rPS{n} \geq \PSpoint}$.
Furthermore, the probability cannot be written with respect to the population distribution because $\IMatchedPSpi$ is only defined in the sample.

The proof is completed immediately if $\Pr{\rPSpop \geq \PSpoint} = 0$ because $\ControlsPSp$ is empty with probability one in that case.
Next consider when $\PSpoint = \suppPSp$.
This means:
\begin{equation}
	\avgin \Pr[\big]{\IMatchedPSpi = 0, \Wi = 0, \rPSi \geq \PSpoint}
	=
	\avgin \Pr[\big]{\IMatchedPSpi = 0, \Wi = 0, \rPSi = \PSpoint},
\end{equation}
and Lemma~\ref{lem:pscore-mdiscrep-term1-no-upper-unmatched-atom-point} immediately completes the proof because $\Pr{\Wpop = 1 \given \rPSpop \geq \PSpoint} = 1 / 2$ when $\PSpoint = \suppPSp$.
The rest of the proof considers the case when $\Pr{\rPSpop \geq \PSpoint} > 0$ and $\PSpoint < \suppPSp$.
It cannot be that $\suppPSp < 1 / 2$ here because $\setb{p \suchthat \Pr{\Wpop = 1 \given \rPSpop \geq p} \geq 1 / 2}$ would then be empty and $\PSpoint = 1 / 2$, which contradicts $\PSpoint < \suppPSp$.
Similarly, $\suppPSp = 1 / 2$ implies $\PSpoint = 1 / 2$.
It must thus be $\suppPSp > 1 / 2$, and then $\PSpoint < 1 / 2$, so this is the case to be considered.

Let $\paren{\UpperEpsiPSpoint}$ be a sequence in $\Reals^+$ such that $\UpperEpsiPSpoint \to 0$ at a sufficiently slow rate so to satisfy:
\begin{equation}
	\Pr{\rPSpop \geq \suppPSp - \UpperEpsiPSpoint} \geq \sqrt{\logf{n} / n},
\end{equation}
for sufficiently large $n$.
Similarly, let $\paren{\LowerEpsiPSpoint}$ be a sequence in $\Reals^+$ such that $\LowerEpsiPSpoint \to 0$ and:
\begin{equation}
	\Pr{\Wpop = 1 \given \rPSpop \geq \PSpoint + \LowerEpsiPSpoint} \geq 1 / 2 + \sqrt{\logf{n} / n},
\end{equation}
for sufficiently large $n$.
Finally, let:
\begin{equation}
	\IntervalPSpoint{p}
	=
	\begin{cases}
		0 & \text{if } p < \PSpoint,
		\\
		1 & \text{if } p = \PSpoint,
		\\
		2 & \text{if } \PSpoint < p < \PSpoint + \LowerEpsiPSpoint,
		\\
		3 & \text{if } \PSpoint + \LowerEpsiPSpoint \leq p \leq \suppPSp - \UpperEpsiPSpoint,
		\\
		4 & \text{if } \suppPSp - \UpperEpsiPSpoint < p < \suppPSp,
		\\
		5 & \text{if } p = \suppPSp.
	\end{cases}
\end{equation}
In other words, $\IntervalPSpoint{p}$ partitions the support of $\rPSpop$ into six groups based on the quantities defined above.

Write the quantity under consideration as:
\begin{align}
	\avgin \Pr[\big]{\IMatchedPSpi = 0, \Wi = 0, \rPSi \geq \PSpoint}
	&= \avgin \Pr[\big]{\IMatchedPSpi = 0, \Wi = 0, \IntervalPSpoint{\rPSi} = 3}
	\\
	&\qquad + \avgin \Pr[\big]{\IMatchedPSpi = 0, \Wi = 0, \IntervalPSpoint{\rPSi} \in \setb{1, 5}}
	\\
	&\qquad + \avgin \Pr[\big]{\IMatchedPSpi = 0, \Wi = 0, \IntervalPSpoint{\rPSi} \in \setb{2, 4}}.
\end{align}
Lemmas~\ref{lem:pscore-mdiscrep-term1-no-upper-unmatched-middle-part} demonstrates that the first term converges to zero, and Lemma~\ref{lem:pscore-mdiscrep-term1-no-upper-unmatched-atom-point} does the same for the second term.
This is because Condition~\refmain{cond:left-closed} implies $\Pr{\Wpop = 1 \given \rPSpop \geq \PSpoint} = 1 / 2$, and $\suppPSp > 1 / 2$ implies $\Pr{\Wpop = 1 \given \rPSpop \geq \suppPSp} > 1 / 2$, so the premise of Lemma~\ref{lem:pscore-mdiscrep-term1-no-upper-unmatched-atom-point} holds for the second term.
For the third term, write:
\begin{equation}
	\Pr[\big]{\IMatchedPSpi = 0, \Wi = 0, \IntervalPSpoint{\rPSi} \in \setb{2, 4}}
	\leq \Pr[\big]{\IntervalPSpoint{\rPSpop} \in \setb{2, 4}},
\end{equation}
so:
\begin{multline}
	\avgin \Pr[\big]{\IMatchedPSpi = 0, \Wi = 0, \IntervalPSpoint{\rPSi} \in \setb{2, 4}}
	\\
	\leq \Pr[\big]{\PSpoint < \rPSpop < \PSpoint + \LowerEpsiPSpoint} + \Pr[\big]{\suppPSp - \UpperEpsiPSpoint < \rPSpop < \suppPSp},
\end{multline}
and $\LowerEpsiPSpoint \to 0$ and $\UpperEpsiPSpoint \to 0$ ensure that also this term converges to zero.
\end{proof}

\begin{lemma}\label{lem:pscore-mdiscrep-term1-no-upper-unmatched-middle-part}
	Given $\PSpoint < 1 / 2 < \suppPSp$ and Conditions~\refmain{cond:reg-cond},~\refmain{cond:pscore-lipschitz} and~\refmain{cond:left-closed}:
	\begin{equation}
		\lim_{n \to \infty} \avgin \Pr[\big]{\IMatchedPSpi = 0, \Wi = 0, \IntervalPSpoint{\rPSi} = 3} = 0,
	\end{equation}
	where $\IntervalPSpoint{p}$ is defined in the proof of Lemma~\ref{lem:pscore-mdiscrep-term1-no-upper-unmatched}.
\end{lemma}

\DeclarePairedDelimiterXPP\PSBalCount[2]{B_{#1}}{\lparen}{\rparen}{}{#2}

\DeclarePairedDelimiterXPP\AvgPSBalCount[1]{\bar{H}}{\lparen}{\rparen}{}{#1}
\DeclarePairedDelimiterXPP\UnitAvgPSBalCount[2]{H_{#1}}{\lparen}{\rparen}{}{#2}
\DeclarePairedDelimiterXPP\ExpoPS[1]{S}{\lparen}{\rparen}{}{#1}

\begin{proof}
Note that:
\begin{equation}
	\Pr[\big]{\IMatchedPSpi = 0, \Wi = 0, \IntervalPSpoint{\rPSi} = 3}
	\leq \Pr[\big]{\IMatchedPSpi = 0 \given \Wi = 0, \IntervalPSpoint{\rPSi} = 3},
\end{equation}
and write:
\begin{equation}
	\Pr[\big]{\IMatchedPSpi = 0 \given \Wi = 0, \IntervalPSpoint{\rPSi} = 3}
	= \E[\Big]{\Pr[\big]{\IMatchedPSpi = 0 \given \Wi = 0, \rPSi} \given \Wi = 0, \IntervalPSpoint{\rPSi} = 3}.
\end{equation}

Let $\PSBalCount{i}{p} = \sum_{j \neq i} \paren{2\Wj - 1} \indicator{\rPSj \geq p}$ count the balance of treated and control units with propensity scores greater than or equal to $p$ excluding unit $i$.
For example, if there are $25$ treated units and $19$ control units with $\rPSj \geq p$ for $j \neq i$, then $\PSBalCount{i}{p} = 25 - 19 = 6$.
Use the law of total probability to write:
\begin{align}
	&\Pr[\big]{\IMatchedPSpi = 0 \given \Wi = 0, \rPSi = p}
	\\
	&\qquad \qquad =
	\Pr[\big]{\PSBalCount{i}{p} \geq 1 \given \Wi = 0, \rPSi = p} \Pr[\big]{\IMatchedPSpi = 0 \given \Wi = 0, \rPSi = p, \PSBalCount{i}{p} \geq 1}
	\\
	&\qquad \qquad \qquad +
	\Pr[\big]{\PSBalCount{i}{p} \leq 0 \given \Wi = 0, \rPSi = p} \Pr[\big]{\IMatchedPSpi = 0 \given \Wi = 0, \rPSi = p, \PSBalCount{i}{p} \leq 0}.
\end{align}
Bound two of the factors as:
\begin{equation}
	\Pr[\big]{\PSBalCount{i}{p} \geq 1 \given \Wi = 0, \rPSi = p} \leq 1
	\quad\text{and}\quad
	\Pr[\big]{\IMatchedPSpi = 0 \given \Wi = 0, \rPSi = p, \PSBalCount{i}{p} \leq 0} \leq 1,
\end{equation}
to get:
\begin{multline}
	\Pr[\big]{\IMatchedPSpi = 0 \given \Wi = 0, \rPSi = p}
	\leq \Pr[\big]{\IMatchedPSpi = 0 \given \Wi = 0, \rPSi = p, \PSBalCount{i}{p} \geq 1}
	\\
	+ \Pr[\big]{\PSBalCount{i}{p} \leq 0 \given \Wi = 0, \rPSi = p}.
\end{multline}

Now for the key step of the proof, namely showing that:
\begin{equation}
	\Pr[\big]{\IMatchedPSpi = 0 \given \Wi = 0, \rPSi = p, \PSBalCount{i}{p} \geq 1} = 0,
\end{equation}
for all $i \in \Sample$.
There are five scenarios to consider:
\begin{enumerate}[label=(\alph*)]
	\item $\nTreated = 0$,
	\item $\nTreated > \nControls$,
	\item $1 \leq \nTreated \leq \nControls$ and $i \not\in \Matched$,
	\item $1 \leq \nTreated \leq \nControls$ and $i \in \MatchedPSm$,
	\item $1 \leq \nTreated \leq \nControls$ and $i \in \MatchedPSp$.
\end{enumerate}
The first scenario can be ignored because $\PSBalCount{i}{p} \geq 1$ implies that at least one treated unit exists in the sample.
In the second scenario, $\MatchedPSp = \ControlsPSp$.
This implies $\IMatchedPSpi = 1$ because $\ControlsPSp = \setb{i \in \Controls \suchthat \rPSi \geq \PSpoint}$ and we are only considering control units with $\rPSi \geq \PSpoint$.

In the third scenario, $\IMatchedPSpi = 0$, but such a matching cannot be optimal.
In particular, $\PSBalCount{i}{p} \geq 1$ means that there is at least one treated unit $k$ with $\rPS{k} \geq p$ that is matched with a control unit $j$ with $\rPSj < p$.
Because unit $i$ is unmatched, we could match unit $k$ with $i$ without otherwise changing the matching, and the sum of within-match propensity score differences would then change by:
\begin{equation}
	\paren{\rPS{k} - \rPSi} - \paren{\rPS{k} - \rPSj} = \rPSj - \rPSi < 0,
\end{equation}
because $\rPSj < p$ and $\rPSi = p$.
Hence, the matching in the third scenario cannot be optimal.

The fourth scenario follows a similar argument.
Also in this scenario, $\IMatchedPSpi = 0$, but again such a matching cannot be optimal.
As before, $\PSBalCount{i}{p} \geq 1$ means that there is at least one treated unit $k$ with $\rPS{k} \geq p$ that is matched with a control unit $j$ with $\rPSj < p$.
Because $i \in \MatchedPSm$, there exists a treated unit $\ell$ with $\rPS{\ell} < p$ that is matched with $i$.
Taken together:
\begin{equation}
	\maxf{\rPS{\ell}, \rPS{\mopt{k}}} < p \leq \minf{\rPS{k}, \rPS{\mopt{\ell}}},
\end{equation}
which means that $\moptsym$ contains crossing matches, but Lemma~\refmain{lem:no-crossing-matches} tells us that no such matching is optimal.

The conclusion is that the only possible scenarios are the second and fifth, and then $\IMatchedPSpi = 1$.
It follows that:
\begin{equation}
	\Pr[\big]{\IMatchedPSpi = 0 \given \Wi = 0, \rPSi = p, \PSBalCount{i}{p} \geq 1} = 0,
\end{equation}
as desired, which gives:
\begin{equation}
	\Pr[\big]{\IMatchedPSpi = 0 \given \Wi = 0, \rPSi = p}
	\leq \Pr[\big]{\PSBalCount{i}{p} \leq 0 \given \Wi = 0, \rPSi = p}.
\end{equation}
Note that $\PSBalCount{i}{p}$ does not depend on $\Wi = 0$ or $\rPSi = p$ other than through the value $p$ because unit $i$ is excluded from the count in $\PSBalCount{i}{p}$.
It follows that:
\begin{equation}
	\Pr[\big]{\PSBalCount{i}{p} \leq 0 \given \Wi = 0, \rPSi = p} = \Pr[\big]{\PSBalCount{i}{p} \leq 0}.
\end{equation}
Note that $\Pr[\big]{\PSBalCount{i}{p} \leq 0} = \Pr[\big]{\PSBalCount{j}{p} \leq 0}$ for all $i, j \in \Sample$ and $p \in \suppPS$ because the probability does not depend on the matching and the observations are otherwise identically distributed.
The rest of the proof uses $\Pr[\big]{\PSBalCount{i}{p} \leq 0}$ for $i = 1$ to represent all units in $\Sample$.

Consider a normalized version of $\PSBalCount{1}{p}$:
\begin{equation}
	\AvgPSBalCount{p} = \frac{1}{n - 1} \sum_{i = 2}^n \UnitAvgPSBalCount{i}{p},
	\qquad\text{where}\qquad
	\UnitAvgPSBalCount{i}{p} =
	\begin{cases}
		\Wi & \text{if } \rPSi \geq p,
		\\
		1 / 2 & \text{if } \rPSi < p.
	\end{cases}
\end{equation}
In particular:
\begin{equation}
	\AvgPSBalCount{p} = \frac{1}{2} + \frac{\PSBalCount{1}{p}}{2\paren{n - 1}}.
\end{equation}
Consider its expectation:
\begin{equation}
	\E[\big]{\AvgPSBalCount{p}} = \frac{1}{2} + \frac{\Pr{\Wpop = 1, \rPSpop \geq p} - \Pr{\Wpop = 0, \rPSpop \geq p}}{2}.
\end{equation}
Define $\ExpoPS{p} = \Pr{\rPSpop \geq p} \bracket[\big]{\Pr{\Wpop = 1 \given \rPSpop \geq p} - 1 / 2}$, so:
\begin{align}
	&\Pr{\Wpop = 1, \rPSpop \geq p} - \Pr{\Wpop = 0, \rPSpop \geq p}
	\\
	&\qquad \qquad = \Pr{\rPSpop \geq p} \bracket[\big]{\Pr{\Wpop = 1 \given \rPSpop \geq p} - \Pr{\Wpop = 0 \given \rPSpop \geq p}}
	\\
	&\qquad \qquad = \Pr{\rPSpop \geq p} \bracket[\big]{2\Pr{\Wpop = 1 \given \rPSpop \geq p} - 1}
	\\
	&\qquad \qquad = 2 \ExpoPS{p},
\end{align}
and $\E[\big]{\AvgPSBalCount{p}} = 1 / 2 + \ExpoPS{p}$.
It follows that:
\begin{equation}
	\Pr[\big]{\PSBalCount{1}{p} \leq 0} = \Pr[\big]{\AvgPSBalCount{p} \leq 1 / 2} = \Pr[\Big]{\AvgPSBalCount{p} - \E[\big]{\AvgPSBalCount{p}} \leq - \ExpoPS{p}}.
\end{equation}
Apply Hoeffding's inequality (Theorem~\ref{thm:hoeffdings-inequality}) to get:
\begin{multline}
	\Pr[\Big]{\AvgPSBalCount{p} - \E[\big]{\AvgPSBalCount{p}} \leq - \ExpoPS{p}}
	\leq \expf[\big]{- \paren{n - 1} \bracket{\ExpoPS{p}}^2}
	\\
	= \expf[\big]{\bracket{\ExpoPS{p}}^2} \expf[\big]{- n \bracket{\ExpoPS{p}}^2}
	\leq 2 \expf[\big]{- n \bracket{\ExpoPS{p}}^2},
\end{multline}
where the last inequality follows from $\exp\paren[\big]{\bracket{\ExpoPS{p}}^2} \leq \exp\paren{1 / 4} \leq 2$.

Recapitulating what we have shown so far, for all $i \in \Sample$:
\begin{multline}
	\E[\Big]{\Pr[\big]{\IMatchedPSpi = 0 \given \Wi = 0, \rPSi} \given \Wi = 0, \IntervalPSpoint{\rPSi} = 3}
	\\
	\leq
	2 \E[\Big]{\expf[\big]{- n \bracket{\ExpoPS{\rPSpop}}^2} \given \Wpop = 0, \IntervalPSpoint{\rPSpop} = 3}.
\end{multline}

Recall that $\PSpoint < 1 / 2 < \suppPSp$, which means:
\begin{equation}
	\Pr{\rPSpop \geq 1 / 2} > 0
	\qquad\text{and}\qquad
	\Pr{\Wpop = 1 \given \rPSpop \geq 1 / 2} > 1 / 2.
\end{equation}
Also recall that $\IntervalPSpoint{p} = 3$ means $\PSpoint + \LowerEpsiPSpoint \leq p \leq \suppPSp - \UpperEpsiPSpoint$.
The rest of the proof considers sufficiently large $n$ so that $\PSpoint + \LowerEpsiPSpoint < 1 / 2 < \suppPSp - \UpperEpsiPSpoint$.
Such samples exist because $\LowerEpsiPSpoint \to 0$ and $\UpperEpsiPSpoint \to 0$.

Consider the events $\PSpoint + \LowerEpsiPSpoint \leq \rPSpop \leq 1 / 2$ and $1 / 2 < \rPSpop \leq \suppPSp - \UpperEpsiPSpoint$.
The function $\ExpoPS{p} = \Pr{\rPSpop \geq p} \bracket[\big]{\Pr{\Wpop = 1 \given \rPSpop \geq p} - 1 / 2}$ is key here.
Note that $\Pr{\rPSpop \geq p}$ is non-negative and decreasing in $p$, and $\Pr{\Wpop = 1 \given \rPSpop \geq p}$ is non-negative and increasing in $p$.
Thus, for any $p$ such that $\PSpoint + \LowerEpsiPSpoint \leq p \leq 1 / 2$:
\begin{equation}
	\ExpoPS{p} \geq \Pr{\rPSpop \geq 1 / 2} \bracket[\big]{\Pr{\Wpop = 1 \given \rPSpop \geq \PSpoint + \LowerEpsiPSpoint} - 1 / 2}.
\end{equation}
Furthermore, $\LowerEpsiPSpoint$ was defined so that:
\begin{equation}
	\Pr{\Wpop = 1 \given \rPSpop \geq \PSpoint + \LowerEpsiPSpoint} \geq 1 / 2 + \sqrt{\logf{n} / n},
\end{equation}
for sufficiently large $n$, and then:
\begin{equation}
	\ExpoPS{p} \geq \Pr{\rPSpop \geq 1 / 2} \sqrt{\logf{n} / n}.
\end{equation}
Similarly, by the definition of $\UpperEpsiPSpoint$, for any $p$ such that $1 / 2 < p \leq \suppPSp - \UpperEpsiPSpoint$:
\begin{multline}
	\ExpoPS{p}
	\geq \Pr{\rPSpop \geq \suppPSp - \UpperEpsiPSpoint} \bracket[\big]{\Pr{\Wpop = 1 \given \rPSpop \geq 1 / 2} - 1 / 2}
	\\
	\geq \sqrt{\logf{n} / n} \bracket[\big]{\Pr{\Wpop = 1 \given \rPSpop \geq 1 / 2} - 1 / 2},
\end{multline}
for sufficiently large $n$.

Let $C = \bracket[\Big]{\minf[\big]{\Pr{\rPSpop \geq 1 / 2}, \Pr{\Wpop = 1 \given \rPSpop \geq 1 / 2} - 1 / 2}}^2$, so:
\begin{equation}
	\ExpoPS{p} \geq \sqrt{C \logf{n} / n},
\end{equation}
for all $p$ such that $\IntervalPSpoint{p} = 3$ when $n$ is sufficiently large.
It follows that:
\begin{equation}
	2 \E[\Big]{\expf[\big]{- n \bracket{\ExpoPS{\rPSpop}}^2} \given \Wpop = 0, \IntervalPSpoint{\rPSpop} = 3}
	\leq
	2 \expf[\big]{- C \logf{n}}
	=
	\frac{2}{n^C}.
\end{equation}
As noted above, $\PSpoint < 1 / 2 < \suppPSp$ implies that $C > 0$.
\end{proof}

\begin{lemma}\label{lem:pscore-mdiscrep-term1-no-upper-unmatched-atom-point}
	Given $\Pr{\Wpop = 1 \given \rPSpop \geq p} \geq 1 / 2$:
	\begin{equation}
		\lim_{n \to \infty} \avgin \Pr[\big]{\IMatchedPSpi = 0, \Wi = 0, \rPSi = p} = 0.
	\end{equation}
\end{lemma}

\DeclarePairedDelimiterXPP\PSBalCountAtom[1]{B}{\lparen}{\rparen}{}{#1}
\DeclarePairedDelimiterXPP\CountAtom[1]{C}{\lparen}{\rparen}{}{#1}

\begin{proof}
The proof is completed immediately if $\Pr{\rPSpop = p} = 0$ because:
\begin{equation}
	\Pr[\big]{\IMatchedPSpi = 0, \Wi = 0, \rPSi = p} \leq \Pr{\rPSpop = p}.
\end{equation}
The rest of the proof considers the case when $\Pr{\rPSpop = p} > 0$.

Let $\PSBalCountAtom{p} = \sumin \paren{2\Wi - 1} \indicator{\rPSi \geq p}$ and let $\CountAtom{p} = \sumin \indicator{\Wi = 0, \rPSi = p}$.
By the same argument as in the proof of Lemma~\ref{lem:pscore-mdiscrep-term1-no-upper-unmatched-middle-part}, if $\PSBalCountAtom{p} \geq 0$, then $\IMatchedPSpi = 1$ must be true for all control units with $\rPSi = p$.
If $\PSBalCountAtom{p} < 0$, then some of these units may not be matched.
However, all treated units with $\rPSi \geq p$ will be matched with control units with $\rPSi \geq p$ if possible.
This means that at most:
\begin{equation}
	-\PSBalCountAtom{p} = \CountAtom{p} - \sumin \indicator{\Wi = 1, \rPSi = p} - \sumin \paren{2\Wi - 1} \indicator{\rPSi > p},
\end{equation}
control units with $\rPSi = p$ are unmatched, and:
\begin{equation}
	\sumin \indicator[\big]{\IMatchedPSpi = 0, \Wi = 0, \rPSi = p} \leq \maxf[\big]{0, \minf[\big]{\CountAtom{p}, - \PSBalCountAtom{p}}}.
\end{equation}

Write:
\begin{align}
	\avgin \Pr[\big]{\IMatchedPSpi = 0, \Wi = 0, \rPSi = p}
	&= \E[\bigg]{\avgin \indicator[\big]{\IMatchedPSpi = 0, \Wi = 0, \rPSi = p}}
	\\
	&\leq \E[\Big]{\maxf[\big]{0, \minf[\big]{\CountAtom{p}, - \PSBalCountAtom{p}}} / n}
	\\
	&\leq \E[\Big]{\maxf[\big]{0, - \PSBalCountAtom{p}} / n},
\end{align}
and:
\begin{align}
	&\E[\Big]{\maxf[\big]{0, - \PSBalCountAtom{p}} / n}
	\\
	&\qquad\qquad = \Pr[\Big]{\PSBalCountAtom{p} \leq -\sqrt{n \logf{n}}} \E[\Big]{\maxf[\big]{0, - \PSBalCountAtom{p}} / n \given \PSBalCountAtom{p} \leq -\sqrt{n \logf{n}}}
	\\
	&\qquad\qquad\qquad + \Pr[\Big]{\PSBalCountAtom{p} > -\sqrt{n \logf{n}}} \E[\Big]{\maxf[\big]{0, - \PSBalCountAtom{p}} / n \given \PSBalCountAtom{p} > -\sqrt{n \logf{n}}}.
\end{align}
Bound two of the factors as:
\begin{equation}
	\E[\Big]{\maxf[\big]{0, - \PSBalCountAtom{p}} / n \given \PSBalCountAtom{p} \leq -\sqrt{n \logf{n}}} \leq 1
	\qquad\text{and}\qquad
	\Pr[\Big]{\PSBalCountAtom{p} > -\sqrt{n \logf{n}}} \leq 1,
\end{equation}
so:
\begin{multline}
	\E[\Big]{\maxf[\big]{0, - \PSBalCountAtom{p}} / n}
	\leq \Pr[\Big]{\PSBalCountAtom{p} \leq -\sqrt{n \logf{n}}}
	\\
	+ \E[\Big]{\maxf[\big]{0, - \PSBalCountAtom{p}} / n \given \PSBalCountAtom{p} > -\sqrt{n \logf{n}}}.
\end{multline}

Consider the first term:
\begin{align}
	\Pr[\Big]{\PSBalCountAtom{p} \leq -\sqrt{n \logf{n}}}
	&= \Pr[\Big]{\PSBalCountAtom{p} / n \leq -\sqrt{\logf{n} / n}}
	\\
	&= \Pr[\Big]{\PSBalCountAtom{p} / n - \E{\PSBalCountAtom{p} / n} \leq - \E{\PSBalCountAtom{p} / n} - \sqrt{\logf{n} / n}}
	\\
	&\leq \Pr[\Big]{\PSBalCountAtom{p} / n - \E{\PSBalCountAtom{p} / n} \leq - \sqrt{\logf{n} / n}},
\end{align}
where the last inequality follows from:
\begin{align}
	\E{\PSBalCountAtom{p} / n}
	&= \Pr{\Wpop = 1, \rPSpop \geq p} - \Pr{\Wpop = 0, \rPSpop \geq p}
	\\
	&= \Pr{\rPSpop \geq p} \bracket[\big]{\Pr{\Wpop = 1 \given \rPSpop \geq p} - \Pr{\Wpop = 0 \given \rPSpop \geq p}}
	\\
	&= 2 \Pr{\rPSpop \geq p} \bracket[\big]{\Pr{\Wpop = 1 \given \rPSpop \geq p} - 1 / 2}
	\\
	&\geq 0,
\end{align}
which in turn holds because $\Pr{\rPSpop \geq p} > 0$ and $\Pr{\Wpop = 1 \given \rPSpop \geq p} \geq 1 / 2$.
Apply Hoeffding's inequality (Theorem~\ref{thm:hoeffdings-inequality}) to get:
\begin{equation}
	\Pr[\Big]{\PSBalCountAtom{p} / n - \E{\PSBalCountAtom{p} / n} \leq - \sqrt{\logf{n} / n}}
	\leq \expf[\big]{-2\logf{n}}
	= \frac{1}{n^2}.
\end{equation}

Complete the proof by noting:
\begin{equation}
	\E[\Big]{\maxf[\big]{0, - \PSBalCountAtom{p}} / n \given \PSBalCountAtom{p} > -\sqrt{n \logf{n}}}
	\leq
	\sqrt{\frac{\logf{n}}{n}}. \tag*{\qedhere}
\end{equation}
\end{proof}

\begin{lemma}\label{lem:pscore-mdiscrep-term1-upper-balance}
	Under Conditions~\refmain{cond:reg-cond} and~\refmain{cond:left-closed}:
	\begin{equation}
		\lim_{n \to \infty} \E[\bigg]{\frac{\card{\MatchedPSp} - \card{\ControlsPSp}}{n}} = 0.
	\end{equation}
\end{lemma}

\begin{proof}
Write:
\begin{equation}
	\E[\bigg]{\frac{\card{\MatchedPSp} - \card{\ControlsPSp}}{n}}
	=
	\E[\bigg]{\frac{\card{\MatchedPSp} - \card{\TreatedPSp}}{n}} + \E[\bigg]{\frac{\card{\TreatedPSp} - \card{\ControlsPSp}}{n}}.
\end{equation}
Consider the absolute value of the first expectation:
\begin{equation}
	\abs[\bigg]{\E[\bigg]{\frac{\card{\MatchedPSp} - \card{\TreatedPSp}}{n}}}
	= \Pr[\big]{\nTreated > \nControls} \abs[\bigg]{\E[\bigg]{\frac{\card{\MatchedPSp} - \card{\TreatedPSp}}{n} \given \nTreated > \nControls}}
	\leq \Pr[\big]{\nTreated > \nControls},
\end{equation}
because $\card{\MatchedPSp} = \card{\setb{\mopt{i} \suchthat i \in \TreatedPSp}} = \card{\TreatedPSp}$ when $\nTreated \leq \nControls$.
As noted in the proof of Lemma~\ref{lem:mdiscrep-bias-step2}:
\begin{equation}
	\lim_{n \to \infty} \Pr[\big]{\nTreated > \nControls} = 0,
\end{equation}
given Condition~\refmain{cond:reg-cond}.

Next:
\begin{equation}
	\E[\bigg]{\frac{\card{\TreatedPSp} - \card{\ControlsPSp}}{n}}
	=
	\Pr{\Wpop = 1, \rPSpop \geq \PSpoint} - \Pr{\Wpop = 0, \rPSpop \geq \PSpoint},
\end{equation}
because $\TreatedPSp = \setb{i \in \Treated \suchthat \rPSi \geq \PSpoint}$ and $\ControlsPSp = \setb{i \in \Controls \suchthat \rPSi \geq \PSpoint}$.
Condition~\refmain{cond:left-closed} implies that $\Pr{\Wpop = 1 \given \rPSpop \geq \PSpoint} = 1 / 2$, so:
\begin{align}
	&\Pr{\Wpop = 1, \rPSpop \geq \PSpoint} - \Pr{\Wpop = 0, \rPSpop \geq \PSpoint}
	\\
	&\qquad\qquad = \Pr{\rPSpop \geq \PSpoint} \bracket[\big]{\Pr{\Wpop = 1 \given \rPSpop \geq \PSpoint} - \Pr{\Wpop = 0 \given \rPSpop \geq \PSpoint}}
	\\
	&\qquad\qquad = 2\Pr{\rPSpop \geq \PSpoint} \bracket[\big]{\Pr{\Wpop = 1 \given \rPSpop \geq \PSpoint} - 1 / 2}
	\\
	&\qquad\qquad = 0,
\end{align}
because $\Pr{\Wpop = 1 \given \rPSpop \geq \PSpoint} + \Pr{\Wpop = 0 \given \rPSpop \geq \PSpoint} = 1$.
\end{proof}

\subsection{Proofs of Lemmas~\ref*{lem:pscore-mdiscrep-term2}~and~\ref*{lem:pscore-mdiscrep-term2-bin-overflow}}

\begin{lemma}\label{lem:pscore-mdiscrep-term2}
	Under Conditions~\refmain{cond:reg-cond},~\refmain{cond:pscore-lipschitz} and~\refmain{cond:left-closed}:
	\begin{equation}
		\lim_{n \to \infty} \E[\bigg]{\frac{1}{\avgpscore n} \sum_{i \in \TreatedPSm} \POci - \frac{1}{\avgpscore n} \sum_{i \in \MatchedPSm} \POci} = 0.
	\end{equation}
\end{lemma}

\newcommand{\eTreatedPSm}{\Treated_e}

\newcommand{\eallM}{\allM_e}

\newcommand{\emoptsym}{\matchsym^*_e}
\DeclarePairedDelimiterXPP\emopt[1]{\emoptsym}{\lparen}{\rparen}{}{#1}
\newcommand{\emopti}{\emopt{i}}
\newcommand{\emoptj}{\emopt{j}}

\newcommand{\mrestsym}{\matchsym_r}
\DeclarePairedDelimiterXPP\mrest[1]{\mrestsym}{\lparen}{\rparen}{}{#1}
\newcommand{\mresti}{\mrest{i}}
\newcommand{\mrestj}{\mrest{j}}

\newcommand{\LipschitzConst}{c_\mu}

\newcommand{\BinSeq}{b_n}
\newcommand{\BinWidth}{w_n}

\newcommand{\SampleBin}[1]{\Sample_{#1,n}}
\newcommand{\TreatedBin}[1]{\Treated_{#1,n}}
\newcommand{\ControlsBin}[1]{\Controls_{#1,n}}

\newcommand{\SampleBink}{\SampleBin{k}}
\newcommand{\TreatedBink}{\TreatedBin{k}}
\newcommand{\ControlsBink}{\ControlsBin{k}}

\begin{proof}
By the same argument as in the proof of Lemma~\ref{lem:pscore-mdiscrep-term1}, namely that the matching only depends on $\W{1}, \W{2}, \dotsc, \W{n}$ and $\rPS{1}, \rPS{2}, \dotsc, \rPS{n}$ and unconfoundedness with respect to the propensity score (Theorem~\ref{thm:pscore-balancing}):
\begin{equation}
	\E[\bigg]{\frac{1}{\avgpscore n} \sum_{i \in \TreatedPSm} \POci - \frac{1}{\avgpscore n} \sum_{i \in \MatchedPSm} \POci}
	= \E[\Bigg]{\frac{1}{\avgpscore n} \sum_{i \in \TreatedPSm} \EPOPSc{\rPSi} - \frac{1}{\avgpscore n} \sum_{i \in \MatchedPSm} \EPOPSc{\rPSi}}
\end{equation}

The unit index will now be extended beyond $\Sample$.
For any $i > n$, set $\rPSi = -1$.
Also extend $\EPOPSc{p}$ so that $\EPOPSc{-1} = 0$.
Define:
\begin{equation}
	\eTreatedPSm
	=
	\begin{cases}
		\TreatedPSm & \text{if } \nTreated \leq \nControls
		\\
		\TreatedPSm \cup \setb{n + i \suchthat 1 \leq i \leq \card{\ControlsPSm} - \card{\TreatedPSm}} & \text{if } \nTreated > \nControls
	\end{cases}
\end{equation}
so $\card{\Treated_e} = \card{\MatchedPSm}$ no matter whether $\nTreated \leq \nControls$ or $\nTreated > \nControls$, because:
\begin{equation}
	\MatchedPSm =
	\begin{cases}
		\setb{\mopt{i} \suchthat i \in \TreatedPSm} & \text{if } \nTreated \leq \nControls,
		\\
		\ControlsPSm & \text{if } \nTreated > \nControls.
	\end{cases}
\end{equation}
Because we defined $\EPOPSc{-1} = 0$, we can write:
\begin{equation}
	\frac{1}{\avgpscore n} \sum_{i \in \TreatedPSm} \EPOPSc{\rPSi} - \frac{1}{\avgpscore n} \sum_{i \in \MatchedPSm} \EPOPSc{\rPSi}
	=
	\frac{1}{\avgpscore n} \sum_{i \in \eTreatedPSm} \EPOPSc{\rPSi} - \frac{1}{\avgpscore n} \sum_{i \in \MatchedPSm} \EPOPSc{\rPSi}
\end{equation}

Let $\eallM$ be all injective functions from $\eTreatedPSm$ to $\Controls \setminus \MatchedPSp$.
Select a $\emoptsym \in \eallM$ satisfying:
\begin{equation}
	\emoptsym \in \argmin_{\mgensym \in \eallM} \sum_{i \in \eTreatedPSm} \abs{\rPSi - \rPS{\mgeni}}.
\end{equation}
If $\nTreated \leq \nControls$, then select $\emoptsym = \moptsym$, so $\emopti = \mopti$ for all $i \in \eTreatedPSm = \TreatedPSm$.
This is possible because $\MatchedPSm \subseteq \Controls \setminus \MatchedPSp$.
If $\nTreated > \nControls$, then $\MatchedPSm = \ControlsPSm = \Controls \setminus \MatchedPSp$.
This means that $\emoptsym$ is a bijection from $\Treated_e$ to $\MatchedPSm$ no matter whether $\nTreated \leq \nControls$ or $\nTreated > \nControls$, and:
\begin{multline}
	\frac{1}{\avgpscore n} \sum_{i \in \eTreatedPSm} \EPOPSc{\rPSi} - \frac{1}{\avgpscore n} \sum_{i \in \MatchedPSm} \EPOPSc{\rPSi}
	=
	\frac{1}{\avgpscore n} \sum_{i \in \eTreatedPSm} \paren[\big]{\EPOPSc{\rPSi} - \EPOPSc{\rPS{\emopti}}}
	\\
	\leq
	\frac{1}{\avgpscore n} \sum_{i \in \eTreatedPSm} \abs[\big]{\EPOPSc{\rPSi} - \EPOPSc{\rPS{\emopti}}}
\end{multline}

Condition~\refmain{cond:pscore-lipschitz} stipulates that $\EPOPSc{p}$ is Lipschitz continuous on the support of $\rPSpop$.
The task now is to extend Lipschitz continuity to also include the point $p = -1$.
Condition~\refmain{cond:reg-cond} implies that $\EPOPSc{a}$ exists for some $a \in \suppPS$.
By the triangle inequality, for any $p \in \suppPS$:
\begin{equation}
	\abs[\big]{\EPOPSc{-1} - \EPOPSc{p}} \leq \abs[\big]{\EPOPSc{-1} - \EPOPSc{a}} + \abs[\big]{\EPOPSc{a} - \EPOPSc{p}}
\end{equation}
Recall $\EPOPSc{-1} = 0$, and $a$ was picked so $\EPOPSc{a}$ existed, so $\abs[\big]{\EPOPSc{-1} - \EPOPSc{a}}$ exists.
Furthermore:
\begin{equation}
	\abs[\big]{\EPOPSc{a} - \EPOPSc{p}} \leq c,
\end{equation}
for all $p \in \suppPS$ because of Lipschitz continuity on $\suppPS \subseteq \bracket{0, 1}$ and $a \in \suppPS$.
The constant $c$ is the Lipschitz constant.
It follows that $\EPOPSc{p}$ is Lipschitz continuous on $\suppPS \cup \setb{-1}$ with Lipschitz constant $\LipschitzConst = c + \abs[\big]{\EPOPSc{a}}$.

By virtue of being Lipschitz continuous:
\begin{equation}
	\frac{1}{\avgpscore n} \sum_{i \in \eTreatedPSm} \abs[\big]{\EPOPSc{\rPSi} - \EPOPSc{\rPS{\emopti}}}
	\leq
	\frac{\LipschitzConst}{\avgpscore n} \sum_{i \in \eTreatedPSm} \abs[\big]{\rPSi - \rPS{\emopti}}
\end{equation}

Now for the key step of the proof.
Consider a weakly growing sequence $\paren{\BinSeq}$ in $\Naturals$ such that $\BinSeq \geq 1$ and $\BinSeq \to \infty$.
The growth rate is, however, sufficiently slow so that:
\begin{equation}
	\lim_{n \to \infty} \BinSeq \logf{n} / n = 0.
\end{equation}
Let $\BinWidth = \paren{\PSpoint - \suppPSm} / \BinSeq$.
For $k \in \setb{1, \dotsc, \BinSeq}$, let:
\begin{align}
	\SampleBink &= \setb{i \in \SamplePSm \suchthat \suppPSm + \paren{k - 1} \BinWidth \leq \rPSi < \suppPSm + k \BinWidth}
	\\
	\TreatedBink &= \setb{i \in \TreatedPSm \suchthat \suppPSm + \paren{k - 1} \BinWidth \leq \rPSi < \suppPSm + k \BinWidth}
	\\
	\ControlsBink &= \setb{i \in \ControlsPSm \suchthat \suppPSm + \paren{k - 1} \BinWidth \leq \rPSi < \suppPSm + k \BinWidth}
\end{align}

Recall that $\eallM$ contains all injective functions from $\eTreatedPSm$ to $\Controls \setminus \MatchedPSp$.
Consider a matching $\mrestsym \in \eallM$ such that $\setb{\mresti \suchthat i \in \TreatedBink} \subseteq \ControlsBink \setminus \MatchedPSp$ if $\card{\TreatedBink} \leq \card{\ControlsBink \setminus \MatchedPSp}$, and $\ControlsBink \setminus \MatchedPSp \subseteq \setb{\mresti \suchthat i \in \TreatedBink}$ if $\card{\TreatedBink} > \card{\ControlsBink \setminus \MatchedPSp}$.
In other words, $\mrestsym$ is such that units in $\TreatedBink$ are first matched with control units in $\ControlsBink$ not matched to a treated unit in $\TreatedPSp$ in $\moptsym$, and if there are not sufficient many such units, the remaining units in $\TreatedBink$ are matched arbitrarily.

Because $\emoptsym$ is an optimum in $\eallM$ and $\mrestsym \in \eallM$:
\begin{equation}
	\frac{\LipschitzConst}{\avgpscore n} \sum_{i \in \eTreatedPSm} \abs[\big]{\rPSi - \rPS{\emopti}}
	\leq
	\frac{\LipschitzConst}{\avgpscore n} \sum_{i \in \eTreatedPSm} \abs[\big]{\rPSi - \rPS{\mresti}}
\end{equation}

Let $\TreatedBin{0} = \setb{i \in \eTreatedPSm \suchthat \rPSi = -1}$, and for completeness, let $\SampleBin{0} = \TreatedBin{0}$ and $\ControlsBin{0} = \emptyset$.
This means that $\TreatedBin{0}, \TreatedBin{1}, \dotsc, \TreatedBin{\BinSeq}$ partition $\eTreatedPSm$, so:
\begin{equation}
	\frac{\LipschitzConst}{\avgpscore n} \sum_{i \in \eTreatedPSm} \abs[\big]{\rPSi - \rPS{\mresti}}
	=
	\frac{\LipschitzConst}{\avgpscore n} \sum_{i \in \TreatedBin{0}} \abs[\big]{\rPSi - \rPS{\mresti}}
	+
	\frac{\LipschitzConst}{\avgpscore n} \sum_{k = 1}^{\BinSeq} \sum_{i \in \TreatedBink} \abs[\big]{\rPSi - \rPS{\mresti}}
\end{equation}
Note $\abs{\rPSi - \rPS{\mresti}} \leq 2$ for $i \in \TreatedBin{0} = \setb{i \in \eTreatedPSm \suchthat \rPSi = -1}$, so:
\begin{equation}
	\frac{\LipschitzConst}{\avgpscore n} \sum_{i \in \TreatedBin{0}} \abs[\big]{\rPSi - \rPS{\mresti}}
	\leq
	\frac{2 \LipschitzConst \card{\TreatedBin{0}}}{\avgpscore n}
\end{equation}

By a similar argument, $\abs{\rPSi - \rPS{\mresti}} \leq 1$ for $i \in \TreatedBink$ and $k \geq 1$.
However, we sometimes have a sharper bound.
If $\card{\TreatedBink} \leq \card{\ControlsBink \setminus \MatchedPSp}$, then $\setb{\mresti \suchthat i \in \TreatedBink} \subseteq \ControlsBink \setminus \MatchedPSp$, so all matches are inside the bin, and $\abs{\rPSi - \rPS{\mresti}} \leq \BinWidth$ for $i \in \TreatedBink$.
If instead $\card{\TreatedBink} > \card{\ControlsBink \setminus \MatchedPSp}$, then $\abs{\rPSi - \rPS{\mresti}} \leq \BinWidth$ holds only for a subset of $\TreatedBink$ of size $\card{\ControlsBink \setminus \MatchedPSp}$.
Taken together, this means that $\abs{\rPSi - \rPS{\mresti}} \leq \BinWidth$ is true for $\minf[\big]{\card{\TreatedBink}, \card{\ControlsBink \setminus \MatchedPSp}}$ units in $\TreatedBink$, and $\abs{\rPSi - \rPS{\mresti}} \leq 1$ for the remaining units.
It follows that, for any $1 \leq k \leq \BinSeq$:
\begin{equation}
	\sum_{i \in \TreatedBink} \abs[\big]{\rPSi - \rPS{\mresti}}
	\leq
	\BinWidth \minf[\big]{\card{\TreatedBink}, \card{\ControlsBink \setminus \MatchedPSp}}
	+
	\paren[\Big]{\card{\TreatedBink} - \minf[\big]{\card{\TreatedBink}, \card{\ControlsBink \setminus \MatchedPSp}}}
\end{equation}
In the first term, bound the minimum as:
\begin{equation}
	\minf[\big]{\card{\TreatedBink}, \card{\ControlsBink \setminus \MatchedPSp}} \leq \card{\TreatedBink}.
\end{equation}
For the second term:
\begin{equation}
	\card{\TreatedBink} - \minf[\big]{\card{\TreatedBink}, \card{\ControlsBink \setminus \MatchedPSp}}
	=
	\maxf[\big]{0, \card{\TreatedBink} - \card{\ControlsBink \setminus \MatchedPSp}}
\end{equation}

Note that $\ControlsBink \setminus \MatchedPSp = \ControlsBink \setminus \paren{\ControlsBink \cap \MatchedPSp}$, so $\card{\ControlsBink \setminus \MatchedPSp} = \card{\ControlsBink} - \card{\ControlsBink \cap \MatchedPSp}$.
Write:
\begin{multline}
	\maxf[\big]{0, \card{\TreatedBink} - \card{\ControlsBink \setminus \MatchedPSp}}
	=
	\maxf[\big]{0, \card{\TreatedBink} - \card{\ControlsBink} + \card{\ControlsBink \cap \MatchedPSp}}
	\\
	\leq
	\maxf[\big]{0, \card{\TreatedBink} - \card{\ControlsBink}} + \card{\ControlsBink \cap \MatchedPSp}
\end{multline}
so that:
\begin{multline}
	\frac{\LipschitzConst}{\avgpscore n} \sum_{k = 1}^{\BinSeq} \sum_{i \in \TreatedBink} \abs[\big]{\rPSi - \rPS{\mresti}}
	\leq
	\frac{\LipschitzConst}{\avgpscore n} \sum_{k = 1}^{\BinSeq} \BinWidth \card{\TreatedBink}
	+
	\frac{\LipschitzConst}{\avgpscore n} \sum_{k = 1}^{\BinSeq} \maxf[\big]{0, \card{\TreatedBink} - \card{\ControlsBink}}
	\\
	+
	\frac{\LipschitzConst}{\avgpscore n} \sum_{k = 1}^{\BinSeq} \card{\ControlsBink \cap \MatchedPSp}
\end{multline}

The sets $\TreatedBin{1}, \TreatedBin{2}, \dotsc, \TreatedBin{\BinSeq}$ partition $\TreatedPSm$, so:
\begin{equation}
	\sum_{k = 1}^{\BinSeq} \BinWidth \card{\TreatedBink} = \BinWidth \card{\TreatedPSm}
\end{equation}
Similarly, $\ControlsBin{1}, \ControlsBin{2}, \dotsc, \ControlsBin{\BinSeq}$ partition $\ControlsPSm$, and:
\begin{multline}
	\ControlsPSm \cap \MatchedPSp
	=
	\paren{\ControlsBin{1} \cup \ControlsBin{2} \cup \dotsb \cup \ControlsBin{\BinSeq}} \cap \MatchedPSp
	\\
	=
	\paren{\ControlsBin{1} \cap \MatchedPSp} \cup \paren{\ControlsBin{2} \cap \MatchedPSp} \cup \dotsb \cup \paren{\ControlsBin{\BinSeq} \cap \MatchedPSp}
\end{multline}
so:
\begin{equation}
	\sum_{k = 1}^{\BinSeq} \card{\ControlsBink \cap \MatchedPSp}
	= \card{\ControlsPSm \cap \MatchedPSp}
\end{equation}
To continue, note that $\ControlsPSp$ and $\ControlsPSm$ partition $\Controls$, so:
\begin{equation}
	\card{\ControlsPSp \cap \MatchedPSp} + \card{\ControlsPSm \cap \MatchedPSp} = \card{\Controls \cap \MatchedPSp} = \card{\MatchedPSp}
\end{equation}
where the last equality follows from $\MatchedPSp \subset \Controls$.
This implies:
\begin{multline}
	\card{\ControlsPSm \cap \MatchedPSp}
	= \card{\MatchedPSp} - \card{\ControlsPSp \cap \MatchedPSp}
	= \paren[\big]{\card{\MatchedPSp} - \card{\ControlsPSp}} + \paren[\big]{\card{\ControlsPSp} - \card{\ControlsPSp \cap \MatchedPSp}}
	\\
	= \paren[\big]{\card{\MatchedPSp} - \card{\ControlsPSp}} + \card{\ControlsPSp \setminus \MatchedPSp}
\end{multline}
where the last equality follows from:
\begin{equation}
	\card{\ControlsPSp \setminus \MatchedPSp} = \card{\ControlsPSp \setminus \paren{\ControlsPSp \cap \MatchedPSp}} = \card{\ControlsPSp} - \card{\ControlsPSp \cap \MatchedPSp}
\end{equation}

Recapitulating what we have shown so far:
\begin{equation}
	\E[\bigg]{\frac{1}{\avgpscore n} \sum_{i \in \TreatedPSm} \POci - \frac{1}{\avgpscore n} \sum_{i \in \MatchedPSm} \POci}
	\leq
	\E[\bigg]{\frac{\LipschitzConst}{\avgpscore n} \sum_{i \in \eTreatedPSm} \abs[\big]{\rPSi - \rPS{\mresti}}}
\end{equation}
and:
\begin{multline}
	\E[\bigg]{\frac{\LipschitzConst}{\avgpscore n} \sum_{i \in \eTreatedPSm} \abs[\big]{\rPSi - \rPS{\mresti}}}
	\leq
	\frac{2 \LipschitzConst}{\avgpscore} \E[\bigg]{\frac{\card{\TreatedBin{0}}}{n}}
	+
	\frac{\LipschitzConst \BinWidth}{\avgpscore} \E[\bigg]{\frac{\card{\TreatedPSm}}{n}}
	+
	\frac{\LipschitzConst}{\avgpscore} \E[\bigg]{\frac{\card{\ControlsPSp \setminus \MatchedPSp}}{n}}
	\\
	+
	\frac{\LipschitzConst}{\avgpscore} \E[\bigg]{\frac{\card{\MatchedPSp} - \card{\ControlsPSp}}{n}}
	+
	\frac{\LipschitzConst}{\avgpscore} \E[\bigg]{\frac{1}{n}\sum_{k = 1}^{\BinSeq} \maxf[\big]{0, \card{\TreatedBink} - \card{\ControlsBink}}}
\end{multline}

For the first term, recall that $\TreatedBin{0} = \emptyset$ when $\nTreated \leq \nControls$, and $\card{\TreatedBin{0}} \leq n$ otherwise.
It follows:
\begin{equation}
	\E[\bigg]{\frac{\card{\TreatedBin{0}}}{n}} \leq \Pr[\big]{\nTreated > \nControls},
\end{equation}
which was shown to converge to zero given Condition~\refmain{cond:reg-cond} in the proof of Lemma~\ref{lem:mdiscrep-bias-step2}.
For the second term, note that $\card{\TreatedPSm} \leq n$, so:
\begin{equation}
	\frac{\LipschitzConst \BinWidth}{\avgpscore} \E[\bigg]{\frac{\card{\TreatedPSm}}{n}} \leq \frac{\LipschitzConst \BinWidth}{\avgpscore},
\end{equation}
which converges to zero because $\BinSeq \to \infty$ implies that $\BinWidth \to 0$.
Lemmas~\ref{lem:pscore-mdiscrep-term1-no-upper-unmatched} and~\ref{lem:pscore-mdiscrep-term1-upper-balance} demonstrate that the third and fourth terms converge to zero.
Lemma~\ref{lem:pscore-mdiscrep-term2-bin-overflow} completes the proof.
\end{proof}

\begin{lemma}\label{lem:pscore-mdiscrep-term2-bin-overflow}
	Under Conditions~\refmain{cond:reg-cond},~\refmain{cond:pscore-lipschitz} and~\refmain{cond:left-closed}:
	\begin{equation}
		\lim_{n \to \infty} \E[\bigg]{\frac{1}{n}\sum_{k = 1}^{\BinSeq} \maxf[\big]{0, \card{\TreatedBink} - \card{\ControlsBink}}} = 0.
	\end{equation}
	where $\BinSeq$, $\TreatedBink$ and $\ControlsBink$ are defined in the proof of Lemma~\ref{lem:pscore-mdiscrep-term2}.
\end{lemma}

\DeclarePairedDelimiterXPP\ExpCountBink[1]{\bar{T}_{k,n}}{\lparen}{\rparen}{}{#1}

\begin{proof}
Recall $\SampleBink = \TreatedBink \cup \ControlsBink$, so:
\begin{equation}
	\maxf[\big]{0, \card{\TreatedBink} - \card{\ControlsBink}}
	=
	\maxf[\big]{0, 2 \card{\TreatedBink} - \card{\SampleBink}}.
\end{equation}
Use the law of iterated expectations to write:
\begin{equation}
	\E[\Big]{\maxf[\big]{0, 2 \card{\TreatedBink} - \card{\SampleBink}}}
	=
	\E[\bigg]{\E[\Big]{\maxf[\big]{0, 2 \card{\TreatedBink} - \card{\SampleBink}} \given \card{\SampleBink}}},
\end{equation}
and:
\begin{align}
	&\E[\Big]{\maxf[\big]{0, 2 \card{\TreatedBink} - \card{\SampleBink}} \given \card{\SampleBink} = u}
	\\
	&\qquad =
	\Pr[\Big]{\card{\TreatedBink} < u / 2 + \sqrt{u\logf{n}} \given \card{\SampleBink} = u}
	\\
	&\qquad \qquad \qquad \times \E[\Big]{\maxf[\big]{0, 2 \card{\TreatedBink} - \card{\SampleBink}} \given \card{\TreatedBink} < u / 2 + \sqrt{u\logf{n}}, \card{\SampleBink} = u}
	\\
	&\qquad \qquad +
	\Pr[\Big]{\card{\TreatedBink} \geq u / 2 + \sqrt{u\logf{n}} \given \card{\SampleBink} = u}
	\\
	&\qquad \qquad \qquad \times \E[\Big]{\maxf[\big]{0, 2 \card{\TreatedBink} - \card{\SampleBink}} \given \card{\TreatedBink} \geq u / 2 + \sqrt{u\logf{n}}, \card{\SampleBink} = u}.
\end{align}
Bound the first probability as:
\begin{equation}
	\Pr[\Big]{\card{\TreatedBink} < u / 2 + \sqrt{u\logf{n}} \given \card{\SampleBink} = u} \leq 1,
\end{equation}
and the first expectation as:
\begin{equation}
	\E[\Big]{\maxf[\big]{0, 2 \card{\TreatedBink} - \card{\SampleBink}} \given \card{\TreatedBink} < u / 2 + \sqrt{u\logf{n}}, \card{\SampleBink} = u} \leq 2\sqrt{u\logf{n}},
\end{equation}
and the second expectation as:
\begin{equation}
	\E[\Big]{\maxf[\big]{0, 2 \card{\TreatedBink} - \card{\SampleBink}} \given \card{\TreatedBink} \geq u / 2 + \sqrt{u\logf{n}}, \card{\SampleBink} = u} \leq u.
\end{equation}

Consider the second probability when $u \geq 1$:
\begin{multline}
	\Pr[\Big]{\card{\TreatedBink} \geq u / 2 + \sqrt{u\logf{n}} \given \card{\SampleBink} = u}
	\\
	=
	\Pr[\Big]{\card{\TreatedBink} / u \geq 1 / 2 + \sqrt{\logf{n} / u} \given \card{\SampleBink} = u}.
\end{multline}
Next, let $\ExpCountBink{u} = \E[\big]{\card{\TreatedBink} \given \card{\SampleBink} = u}$, so:
\begin{multline}
	\Pr[\Big]{\card{\TreatedBink} / u \geq 1 / 2 + \sqrt{\logf{n} / u} \given \card{\SampleBink} = u}
	\\
	=
	\Pr[\Big]{\card{\TreatedBink} / u - \ExpCountBink{u} / u \geq 1 / 2 - \ExpCountBink{u} / u + \sqrt{\logf{n} / u} \given \card{\SampleBink} = u}.
\end{multline}
Note that:
\begin{equation}
	\ExpCountBink{u} / u
	= \Pr[\big]{\Wpop = 1 \given \suppPSm + \paren{k - 1} \BinWidth \leq \rPSpop < \suppPSm + k \BinWidth},
\end{equation}
because $\TreatedBink = \SampleBink \cap \Treated$ and $\SampleBink = \setb{i \in \SamplePSm \suchthat \suppPSm + \paren{k - 1} \BinWidth \leq \rPSi < \suppPSm + k \BinWidth}$.
Recall that $\suppPSm + \BinSeq \BinWidth = \PSpoint \leq 1 / 2$, so for all $k \in \setb{1, 2, \dotsc, \BinSeq}$:
\begin{equation}
	\Pr[\big]{\Wpop = 1 \given \suppPSm + \paren{k - 1} \BinWidth \leq \rPSpop < \suppPSm + k \BinWidth} \leq 1 / 2,
\end{equation}
and:
\begin{multline}
	\Pr[\Big]{\card{\TreatedBink} / u - \ExpCountBink{u} / u \geq 1 / 2 - \ExpCountBink{u} / u + \sqrt{\logf{n} / u} \given \card{\SampleBink} = u}
	\\
	\leq
	\Pr[\Big]{\card{\TreatedBink} / u - \ExpCountBink{u} / u \geq \sqrt{\logf{n} / u} \given \card{\SampleBink} = u}.
\end{multline}
Apply Hoeffding's inequality (Theorem~\ref{thm:hoeffdings-inequality}) to get:
\begin{equation}
	\Pr[\Big]{\card{\TreatedBink} / u - \ExpCountBink{u} / u \geq \sqrt{\logf{n} / u} \given \card{\SampleBink} = u}
	\leq \expf[\big]{-2\logf{n}}
	= \frac{1}{n^2}.
\end{equation}

Taken together:
\begin{equation}
	\E[\Big]{\maxf[\big]{0, 2 \card{\TreatedBink} - \card{\SampleBink}} \given \card{\SampleBink} = u}
	\leq 2\sqrt{u\logf{n}} + \frac{u}{n^2},
\end{equation}
and:
\begin{equation}
	\E[\Big]{\maxf[\big]{0, 2 \card{\TreatedBink} - \card{\SampleBink}}}
	=
	\E[\bigg]{2\sqrt{\card{\SampleBink} \logf{n}} + \frac{\card{\SampleBink}}{n^2}},
\end{equation}
so:
\begin{equation}
	\E[\bigg]{\frac{1}{n}\sum_{k = 1}^{\BinSeq} \maxf[\big]{0, \card{\TreatedBink} - \card{\ControlsBink}}}
	\leq
	\frac{1}{n^3}\sum_{k = 1}^{\BinSeq} \E[\big]{\card{\SampleBink}}
	+
	\frac{2\sqrt{\logf{n}}}{n}\sum_{k = 1}^{\BinSeq} \E[\bigg]{\sqrt{\card{\SampleBink}}}.
\end{equation}

First consider:
\begin{equation}
	\frac{1}{n^3}\sum_{k = 1}^{\BinSeq} \E[\big]{\card{\SampleBink}}
	= \frac{1}{n^3} \E[\bigg]{\sum_{k = 1}^{\BinSeq} \card{\SampleBink}}
	= \frac{\E[\big]{\card{\SamplePSm}}}{n^3}
	= \frac{\Pr{\rPSpop < \PSpoint}}{n^2}
	\leq \frac{1}{n^2},
\end{equation}
because $\SampleBin{1}, \SampleBin{2}, \dotsc, \SampleBin{\BinSeq}$ partition $\SamplePSm$.
It follows that the first term converges to zero.

Next, use Jensen's inequality and concavity of the square root to get:
\begin{equation}
	\frac{2\sqrt{\logf{n}}}{n}\sum_{k = 1}^{\BinSeq} \E[\bigg]{\sqrt{\card{\SampleBink}}}
	\leq
	\frac{2\sqrt{\logf{n}}}{n}\sum_{k = 1}^{\BinSeq} \sqrt{\E[\big]{\card{\SampleBink}}}.
\end{equation}
Use Jensen's inequality once more:
\begin{equation}
	\frac{2\sqrt{\logf{n}}}{n}\sum_{k = 1}^{\BinSeq} \sqrt{\E[\big]{\card{\SampleBink}}}
	\leq
	\frac{2 \sqrt{\BinSeq \logf{n}}}{n} \sqrt{\sum_{k = 1}^{\BinSeq}\E[\big]{\card{\SampleBink}}},
\end{equation}
and, finally:
\begin{equation}
	\frac{2 \sqrt{\BinSeq \logf{n}}}{n} \sqrt{\sum_{k = 1}^{\BinSeq}\E[\big]{\card{\SampleBink}}}
	=
	\frac{2 \sqrt{\BinSeq \logf{n}}}{n} \sqrt{n \Pr{\rPSpop < \PSpoint}}
	\leq
	2 \sqrt{\frac{\BinSeq \logf{n}}{n}},
\end{equation}
which implies that also this term converges to zero because $\BinSeq$ was defined in the proof Lemma~\ref{lem:pscore-mdiscrep-term2} so that:
\begin{equation}
	\lim_{n \to \infty} \frac{\BinSeq \logf{n}}{n} = 0. \tag*{\qedhere}
\end{equation}
\end{proof}

\subsection{Proof of Lemma~\ref*{lem:pscore-mdiscrep-term3}}

\begin{lemma}\label{lem:pscore-mdiscrep-term3}
	Under Condition~\refmain{cond:left-closed}:
	\begin{multline}
		\E[\bigg]{\frac{1}{\avgpscore n} \sum_{i \in \TreatedPSp} \POci - \frac{1}{\avgpscore n} \sum_{i \in \ControlsPSp} \POci}
		\\
		=
		\frac{\Pr{\rPSpop \geq \PSpoint}}{2\avgpscore} \paren[\Big]{\E[\big]{\POcpop \given \Wpop = 1, \rPSpop \geq \PSpoint} - \E[\big]{\POcpop \given \Wpop = 0, \rPSpop \geq \PSpoint}}.
	\end{multline}
\end{lemma}

\begin{proof}
The proof is completed immediately if $\Pr{\rPSpop \geq \PSpoint} = 0$ because $\TreatedPSp$ and $\ControlsPSp$ are then empty with probability one.
The rest of the proof considers the case when $\Pr{\rPSpop \geq \PSpoint} > 0$.

By the same argument as in the previous proofs, $i \in \TreatedPSp$ provides no more information about $\POci$ than $\Wi = 1$ and $\rPSi \geq \PSpoint$, so:
\begin{equation}
\E[\bigg]{\frac{1}{\avgpscore n} \sum_{i \in \TreatedPSp} \POci}
= \frac{\E[\big]{\nTreatedPSp}}{\avgpscore n} \E[\big]{\POcpop \given \Wpop = 1, \rPSpop \geq \PSpoint}.
\end{equation}
Similarly, $i \in \ControlsPSp$ provides no more information about $\POci$ than $\Wi = 0$ and $\rPSi \geq \PSpoint$, and:
\begin{equation}
\E[\bigg]{\frac{1}{\avgpscore n} \sum_{i \in \ControlsPSp} \POci}
= \frac{\E[\big]{\nControlsPSp}}{\avgpscore n} \E[\big]{\POcpop \given \Wpop = 0, \rPSpop \geq \PSpoint}.
\end{equation}
Note that $\E[\big]{\nTreatedPSp} = n \Pr{\Wpop = 1, \rPSi \geq \PSpoint}$, and similarly for $\E[\big]{\nControlsPSp}$ so:
\begin{multline}
\frac{\E[\big]{\nTreatedPSp}}{\avgpscore n} = \frac{\Pr{\rPSi \geq \PSpoint} \Pr{\Wpop = 1 \given \rPSi \geq \PSpoint}}{\avgpscore}
\qquad\text{and}
\\
\frac{\E[\big]{\nControlsPSp}}{\avgpscore n} = \frac{\Pr{\rPSi \geq \PSpoint} \Pr{\Wpop = 0 \given \rPSi \geq \PSpoint}}{\avgpscore}.
\end{multline}

Condition~\refmain{cond:left-closed} and $\Pr{\rPSpop \geq \PSpoint} > 0$ imply that $\Pr{\Wpop = 1 \given \rPSi \geq \PSpoint} = 1 / 2$.
It follows that:
\begin{equation}
	\Pr{\Wpop = 0 \given \rPSi \geq \PSpoint}
	= 1 - \Pr{\Wpop = 1 \given \rPSi \geq \PSpoint}
	= 1 / 2
	= \Pr{\Wpop = 1 \given \rPSi \geq \PSpoint},
\end{equation}
and:
\begin{equation}
\frac{\Pr{\rPSi \geq \PSpoint} \Pr{\Wpop = 1 \given \rPSi \geq \PSpoint}}{\avgpscore}
= \frac{\Pr{\rPSi \geq \PSpoint} \Pr{\Wpop = 0 \given \rPSi \geq \PSpoint}}{\avgpscore}
= \frac{\Pr{\rPSi \geq \PSpoint}}{2 \avgpscore}. \tag*{\qedhere}
\end{equation}
\end{proof}

\end{document}